\documentclass[aps,prx,twocolumn,groupedaddress,showpacs,superscriptaddress,scrartcl,longbibliography]{revtex4-2}
\usepackage{xcolor,lineno,mathtools,comment,mathrsfs,natbib,amssymb,amsmath,graphicx,epstopdf,hyperref,amsthm,dsfont,enumerate,bm}

\newtheorem*{lemma6*}{Lemma 6}
\newtheorem*{theorem1*}{Theorem 1}
\newtheorem{lemma}{Lemma}
\newtheorem{lemmaa}{Lemma}
\newtheorem*{lemmaaA*}{Lemma A}
\newtheorem*{lemmaaB*}{Lemma B}

\newtheorem{definition}{Definition}
\newtheorem{theorem}{Theorem}

\newtheorem{principle}{Principle}
\newtheorem{postulate}{Postulate}
\newtheorem*{postulate1'*}{Postulate 1'}
\newtheorem*{postulate1*}{Postulate 1}
\newtheorem*{postulate3'*}{Postulate 3'}
\newtheorem{postulatea}{Postulate}

\newcommand{\reff}[1]{{\color{black}{#1}}}

\begin{document}
\interfootnotelinepenalty=10000

%\title{Spacetime symmetries and the qubit Bloch ball: A derivation of finite dimensional quantum theory}

%\title{Spacetime symmetries and the qubit Bloch ball: a physical derivation of finite dimensional quantum theory and the number of spatial dimensions in Minkowski spacetime}

\title{Spacetime symmetries and the qubit Bloch ball: a physical derivation of finite dimensional quantum theory and the number of spatial dimensions}

%\title{Spacetime symmetries and the qubit Bloch ball: a derivation of finite dimensional quantum theory and the number of spatial dimensions from Poincar{\'e} invariance and other physical postulates}

%\title{Spacetime symmetries and the qubit Bloch ball: deriving Minkowski spacetime and finite dimensional quantum theory}

%\title{An axiomatic derivation of Minkowski spacetime and finite dimensional quantum theory}

%title{Deriving finite dimensional quantum theory and the number of spatial dimensions of Minkowski spacetime}

\author{Dami\'an Pital\'ua-Garc\'ia}
%\email[]{dpitaluag@gmail.com}
\email[]{D.Pitalua-Garcia@damtp.cam.ac.uk}
\affiliation{Centre for Quantum Information and Foundations, DAMTP,
Centre for Mathematical Sciences, University of Cambridge,
Wilberforce Road, Cambridge, CB3 0WA, United Kingdom}

\date{\today}

%\pacs{03.65.Ud}

\begin{abstract}
Quantum theory and relativity are the pillar theories on which our understanding of physics is based. 
%Two important, long standing and independent research problems are reconstructing quantum theory from sensible physical conditions \cite{BN36,Mackeybook,JP63,DL70,E70,Ludwigbook,H01,DB09,CDP11,MM11,H11,TMSM12,MMAP13,BMU14}, and investigating whether spacetime and quantum theory are linked at the level of the Hilbert space formalism of quantum theory \cite{Weizsackerbook,MM13.1,DB13,HM14,GMD14}.
Poincar{\'e} invariance is a fundamental physical principle stating that the experimental results must be the same in all inertial reference frames in Minkowski spacetime. It is a basic condition imposed on quantum theory in order to construct quantum field theories, hence, it plays a fundamental role in the standard model of particle physics too. As is well known, Minkowski spacetime follows from clear physical principles, like the relativity principle and the invariance of the speed of light. Here, we reproduce such a derivation, but leave the number of spatial dimensions $n$ as a free variable. Then, assuming that spacetime is Minkowski in $1+n$ dimensions and within the framework of general probabilistic theories, we reconstruct the qubit Bloch ball and finite dimensional quantum theory, and obtain that the number of spatial dimensions must be $n=3$, from Poincar{\'e} invariance and other physical postulates. Our results suggest a fundamental physical connection between spacetime and quantum theory.\\

%postulates of quantum theory$|$ postulates of special relativity$|$ dimension of space $|$ quantum information

\begin{comment}
Minkowski spacetime in $1+n$ dimensions is derived from the relativity principle,
the invariance of the speed of light and other well known physical postulates. Within the formalism of general probabilistic theories, we analyse a massive particle whose momentum corresponds to a continuous variable and whose
internal degrees of freedom are assumed discrete and finite. Poincar{\'e} invariance states that, under
Poincar{\'e} transformations, the states and effects transform as representations of the Poincar{\'e} group
and that the outcome probabilities remain invariant. We show that, for theories including a
class of states with well defined momentum, the state space of minimal dimension for the particles'
internal degrees of freedom being consistent with Poincar{\'e} invariance is an Euclidean ball of
dimension $n$. Then, the qubit Bloch ball, $n=3$ and finite dimensional quantum theory follow from physical
postulates of previous reconstructions of quantum theory: tomographic
locality, continuous reversibility, existence of entanglement and universal encoding.
\end{comment}

\end{abstract}

\maketitle

\section{Introduction}
%Quantum theory and relativity are the pillar theories on which our current understanding of physics is based. 
While special relativity is clearly stated in terms of the relativity principle and the invariance of the speed of light \cite{LandauLifshitzbook,WeinbergGravitationandCosmologybook}, quantum theory is traditionally formulated in terms of abstract mathematical postulates, involving vectors in a complex Hilbert space \cite{vonNeumannbook,NielsenandChuangbook,WeinbergQMbook}.
%Poincar{\'e} invariance is a fundamental relativistic physical principle stating that the experimental results must be the same in all inertial reference frames \cite{Weinbergbook}. It is a basic condition imposed on quantum theory in order to construct quantum field theories, hence, it plays a fundamental role in the standard model of particle physics too \cite{Weinbergbook}.
%Two important, long standing and independent research problems are to reconstruct quantum theory from sensible physical conditions \cite{BN36,Mackeybook,JP63,DL70,E70,Ludwigbook,H01,DB09,CDP11,MM11,H11,TMSM12,MMAP13,BMU14}, and to investigate whether spacetime and quantum theory are linked at the level of the Hilbert space formalism of quantum theory \cite{Weizsackerbook,MM13.1,DB13,HM14,GMD14}.
Two important, long standing and independent research problems are reconstructing quantum theory from sensible physical conditions \cite{BN36,Mackeybook,JP63,DL70,E70,Ludwigbook,H01,CBH03,G08,R09,DB09,CDP10,R11,CDP11,MM11,H11,FS11,W12,TMSM12,F12,Z12,H13,MMAP13,BMU14,HW17,H17,H17.2,SSC18,W19,vdW19,T20,N20}, and investigating whether spacetime and quantum theory are linked at the level of the Hilbert space formalism \cite{Weizsackerbook,MM13.1,DB13,HM14,GMD17}.

Important motivations to reconstruct quantum theory from physical conditions are to understand quantum theory better
% \cite{H01,DB09,CDP11,MM11,H11,TMSM12,MMAP13,BMU14}
and to explore ways in which quantum theory could be modified by investigating variations of its foundational principles. An important reason for doing this is the problem of unifying gravity and quantum physics \cite{H07,H16}. Thus, it is compelling to explore physical principles that suggest clear connections between the mathematical structures of spacetime and quantum theory.

The framework of general probabilistic theories (GPTs) has minimal assumptions and includes classical and quantum theory as special cases (see e.g. \cite{H01,B07,DB09,CDP10,CDP11,MM11,H11,TMSM12,MMAP13}). It has allowed the investigation of quantum properties and protocols within a broader framework of theories, like entanglement \cite{B07,SB10}, the violation of Bell inequalities \cite{BLMPPR05,B07,JGBB11}, the no-cloning theorem \cite{B07}, the no-broadcasting theorem \cite{BBLW07}, entropy \cite{BBCLSSWW10,SW10,KBBM17}, teleportation \cite{B07,SB10,BBLW08,MPP15}, dense coding \cite{SB10,MPP15}, entanglement swapping \cite{SB10,MPP15}, communication \cite{B07,MP14}, computation \cite{B07,LB15,LH16,LS16,LS16.2,BLS18,G18,KM19,BBHL19} and cryptography \cite{B07,SS18,SS18.2,LPW18,SS20}, for instance. Furthermore, several interesting axiomatic reconstructions of the mathematical formalism of finite dimensional quantum theory have been proposed within this framework (e.g. \cite{H01,DB09,CDP11,MM11,H11,TMSM12,MMAP13,BMU14}).

%There exist several reconstructions of finite dimensional quantum theory within the framework of GPTs, from sets of axioms and postulates that are, arguably, physically sensible (e.g. \cite{}). As mentioned above, an important motivation to derive quantum theory from sets of physically motivated postulates is, by relaxing the postulates, to investigate extensions of quantum theory that could be unified with gravity. We thus find surprising that, to our knowledge, at the moment of writing this paper there is not any proposed physical principle or postulate in the axiomatic reconstructions of quantum theory suggesting a clear connection with the whole mathematical structure of Minkowski or curved spacetime.

To the best of our knowledge, at the moment of writing this paper there is not any proposed physical principle or postulate in the axiomatic reconstructions of quantum theory within the GPTs framework suggesting a clear connection between the mathematical formalism of quantum theory and the whole mathematical structure of Minkowski or curved spacetime. In particular, there is not any proposed postulate \reff{in the GPTs framework} that exploits the whole group of symmetry transformations in Minkowski spacetime, given by the proper orthochronous Poincar{\'e} group, i.e. the spacetime translations, space rotations and Lorentz boosts. \reff{However, we note that Svetlichny \cite{S00} has derived the Hilbert space formalism of quantum mechanics from Poincar{\'e} invariance and other postulates in the framework of quantum logic.}

Physical principles and postulates inspired by the causality of spacetime have been proposed and investigated in the GPTs framework (e.g. \cite{CDP10,ic,MMAP13}). For example, causality, i.e. the condition that the probability of preparing a system in a given state is independent of what measurement is applied after the preparation \cite{CDP10}, is a standard assumption made in GPTs. The no-signalling principle, which allows consistency with relativistic causality, is also a standard assumption in GPTs. No-signalling says that the outcome probabilities of any measurement applied on $A$ are independent of the measurement applied on $B$ for any bipartite system $AB$ in an arbitrary state. Extensions of the no-signalling principle called ``Information Causality'' \cite{ic} and ``No Simultaneous Encoding'' \cite{MMAP13} have been proposed to derive some quantum properties.

%Although physical principles and postulates inspired by the causality of spacetime have been proposed and investigated in the framework of general probabilistic theories (GPTs) \cite{}, to the best of our knowledge, at the moment of writing this paper there are not other proposed physical postulates that exploit the whole structure of spacetime. For example, causality, i.e. the condition that the probability of preparing a system in a given state is independent of what measurement is applied after the preparation \cite{}, is a standard assumption made in the framework of GPTs. The no-signalling principle, stating that the outcome probabilities of any measurement applied on $A$ are independent of the measurement applied on $B$ for any bipartite system $AB$ in an arbitrary state, which allows consistency with relativistic causality, is also a standard assumption in the GPTs framework. The extensions of the causality condition called ``Information Causality'' \cite{ic} and ``No Simultaneous Encoding'' \cite{MMAP13} have been proposed to derive some properties of quantum theory. 

Information causality \cite{ic} roughly states that the transmission of $m$ classical bits by a first party, Alice, to a second party, Bob, cannot increase Bob's information about Alice's data in more than $m$ bits. A quantum version of this principle \cite{QIC} states that the transmission of $m$ qubits by Alice to Bob cannot increase Bob's quantum information about Alice's data in more than a value of $m$. 
Information causality was used to derive the quantum Tsirelson bound \cite{C80} on the violation of the Clauser-Horne-Shimony-Holt (CHSH) Bell inequality \cite{CHSH69}.
This principle reduces to the no-signalling principle when $m=0$. Although no-signalling allows the satisfaction of relativistic causality and is thus motivated by the structure of spacetime, information causality is not clearly motivated by any property of spacetime.

No Simultaneous Encoding roughly states that there exists an elementary system called ``the gbit'', which if used to perfectly encode one classical bit cannot simultaneously encode any further information. This postulate was used together with other postulates to reconstruct finite dimensional quantum theory \cite{MMAP13}.  Although this postulate might be a sensible variation of information causality, it is not clearly motivated by any feature of spacetime either.

Refs. \cite{MM13.1,DB13} investigated connections between the number of spatial dimensions and the mathematical structure of finite dimensional quantum theory in the framework of GPTs. But these papers did not consider any relativistic effects. In particular, they did not consider Lorentz boosts.

Assuming that space has $n$ dimensions and that there exists a physical system that allows to encode any direction in space, Ref. \cite{MM13.1} derived from a set of postulates that $n=3$ and that a pair of such systems must be described by quantum theory. However, some of the postulate proposed by Ref. \cite{MM13.1} are not clearly physically sensible. For example, the second postulate roughly says that if a state perfectly encodes a spatial direction, it cannot encode any further information. We note that this postulate is very similar to No Simultaneous Encoding. As for No Simultaneous Encoding, we do not see any compelling reasons to consider that this postulate is physically sensible.

Assuming that space is Euclidean and isotropic with $n$ dimensions, and using a set of postulates, Ref. \cite{DB13} reconstructed finite dimensional quantum theory and that $n=3$. However, Ref. \cite{DB13} made two strong assumptions: 1) there exists an elementary system whose state space is an Euclidean ball of dimension $d$, which generalizes the qubit Bloch ball of three dimensions; and 2) the dimension of space equals the dimension of the Euclidean ball, i.e. $d=n$.

Refs. \cite{HM14,GMD17} investigated connections between the mathematical structures of Minkowski spacetime, including relativistic effects, and finite dimensional quantum theory in the GPTs framework. By establishing a task in which two distant parties must synchronize their description of local physics using quantum communication, and with the help of some assumptions, Ref. \cite{HM14} derived the group of Lorentz transformations of Minkowski spacetime. We think that deriving properties of Minkowski spacetime from quantum theory is a very interesting problem. However, here we take the view that the reverse problem, in which properties of quantum theory are derived from Minkowski spacetime, is physically more compelling. The main reason for our point of view is that, as already mentioned, Minkowski spacetime follows clearly from well established physical principles, but quantum theory is commonly stated in terms of abstract mathematical postulates. Although we appreciate that there are many axiomatic reconstructions of quantum theory, we believe that all of them make some assumptions that are arguably not completely physically sensible.

Ref. \cite{GMD17} considered the following thought experiment. A particle passes through a beam splitter and is then superposed along two different paths. The experiment is described in two different inertial reference frames that are connected by a Lorentz transformation $\Lambda$. It was derived that the state space $\mathcal{S}$ corresponding to the path superpositions for the particle is the qubit Bloch ball. We find this approach very interesting. However, there were several strong assumptions in this derivation. 

First, Ref. \cite{GMD17} assumed that the state space $\mathcal{S}$ is an Euclidean ball of arbitrary dimension $d$, generalizing the qubit Bloch ball. Second, it was assumed that under a Lorentz transformation the path superposition state transforms as a finite dimensional representation of the Lorentz group. Because the only finite dimensional unitary representation of the Lorentz group is the trivial representation \cite{W39,Weinbergbook}, Ref. \cite{GMD17} concluded that the path superposition state does not transform under Lorentz transformations. However, by considering infinite dimensional degrees of freedom, the path superposition state could, in principle, transform nontrivialy under Lorentz transformations.
Third, it was assumed that any pure state on the Euclidean sphere could be prepared by setting up the beam splitter appropriately and applying local reversible transformations on the two branches of the interferometer. Finally, it was assumed that the group of reversible transformations on both branches are equal.

The only physical property from Minkowski spacetime used by Ref. \cite{GMD17} was the simultaneity of relativity, according to which the time order of two spacelike separated events can be inverted in inertial reference frames that are connected by some Lorentz boost. The whole group of symmetry transformations in Minkowski spacetime, given by the proper orthochronous Poincar{\'e} group, was not used by Ref. \cite{GMD17}.

The main physical contribution of this paper to the literature of GPTs and reconstructions of quantum theory is to propose a physical postulate that suggests a clear connection between the group of symmetry transformations in Minkowski spacetime, given by the proper orthochronous Poincar{\'e} group $\mathfrak{Poin}$, and the Hilbert space structure of finite dimensional quantum theory. Our Postulate \hyperref[P1first]{1}, ``Poincar{\'e} Structure'', assumes that spacetime is Minkowski in $1+n$ dimensions and roughly states that there exists a type of massive particle $\mathcal{P}$ described by a GPT satisfying: 1) the states and measurements transform as nontrivial representations of $\mathfrak{Poin}$ and the outcome probabilities remain invariant under transformations from $\mathfrak{Poin}$; and 2) the particle has internal degrees of freedom of finite dimension that if are required to transform as a representation of a subgroup of $\mathfrak{Poin}$ must do so nontrivially. We restrict the group of symmetry transformations to be the proper orthochronous Poincar{\'e} group $\mathfrak{Poin}$ and not the full  Poincar{\'e} group $\mathfrak{Poin}_\text{full}$, which includes the transformations of space inversion and time reversal, because physics is observed to be perfectly invariant under $\mathfrak{Poin}$ but not under $\mathfrak{Poin}_\text{full}$ \cite{Weinbergbook}. In quantum theory, a particle of the type $\mathcal{P}$ can be an electron, with the internal degrees of freedom corresponding to the spin, for instance.

We propose other postulates in addition to Postulate \hyperref[P1first]{1} and derive the qubit Bloch ball, finite dimensional quantum theory and that the number of spatial dimensions must be $n=3$. In our derivation we assume that spacetime is Minkowski in $1+n$ dimensions. As is well known, Minkowski spacetime is derived from clear physical principles \cite{LandauLifshitzbook,WeinbergGravitationandCosmologybook}. For completeness of this paper, we present such a derivation here, and leave the number of spatial dimensions $n$ as a free variable.

Our Postulate \hyperref[P2]{2}, ``Existence of a Classical Limit'', states that there exists a classical limit for the particle of the type $\mathcal{P}$, in the sense that two conditions hold: 1) there exists a class of states in which the particle has classical well defined $1+n$ momentum $p$; and 2) a classical bit can be encoded in the particle's internal degrees of freedom. To our knowledge, continuous and spacetime degrees of freedom, like the momentum of a particle, have not been considered in previous axiomatic reconstructions of quantum theory.

The conditions 1) of Postulates \hyperref[P1first]{1} and \hyperref[P2]{2} allow us to use Wigner's method of induced representations \cite{W39,Weinbergbook} to obtain our main technical result, Lemma \ref{newlemma}, which roughly says that there exists a class of states for a particle of the type $\mathcal{P}$ in which the internal degrees of freedom must transform as a representation of the group $\text{SO}(n)$. 

Our Postulate \hyperref[P3]{3}, ``Minimality of the Elementary System'', says that there exists an elementary system having the state space with the smallest nontrivial finite dimension $d_\text{elem}$ in nature, which can be physically implemented in some internal degrees of freedom of a particle of the type $\mathcal{P}$, and which satisfies that $d_\text{elem}$ achieves the minimum value that is consistent with the other postulates. This is arguably a strong assumption, but it is weaker than Hardy's \cite{H01} ``Simplicity Axiom'' stating that, for any physical system, the dimension of its state space takes the minimum value that is consistent with the other axioms. Furthermore, we believe this is a reasonable assumption. We think it is physically sensible to assume that physical theories should have mathematical structures that are as simple as possible, while still describing a broad range of physical phenomena. We think that a reasonable measure of mathematical simplicity for a theory is given by the number of real degrees of freedom needed to describe the elementary system.

Our Postulates \hyperref[P4]{4} -- \hyperref[P7]{7} have been used in previous reconstructions of finite dimensional quantum theory (e.g. \cite{H01,DB09,MM11,CDP11,H11,TMSM12,MMAP13}) and are arguably physically sensible. Postulate \hyperref[P4]{4}, ``Continuous Reversibility'', says that for every pair of pure states there exists a continuous reversible transformation that transforms one into the other \cite{H01,DB09,MM11,TMSM12,MMAP13}. Postulate \hyperref[P5]{5}, ``Tomographic Locality'', says that the state of a composite system is totally characterized by the outcome probabilities of the local measurements on the subsystems \cite{H01,B07,DB09,CDP11,MM11,H11,TMSM12,MMAP13}. Postulate \hyperref[P6]{6}, ``Existence of Entanglement'', says that the state space of any bipartite system contains at least one entangled state \cite{TMSM12,MMAP13}. Postulate \hyperref[P7]{7}, ``Universal Encoding'', says that for any physical system, any state of finite dimension can be reversibly encoded in a sufficiently large number of elementary systems \cite{MMAP13}.

Assuming from the beginning that spacetime is Minkowski in $1+3$ dimensions, we show in Lemma \ref{qubitlemma} from Postulates \hyperref[P1first]{1} -- \hyperref[P4]{4} that the states, measurements and reversible transformations of the elementary system are equivalent to those of the qubit Bloch ball. Using this result and the results of Refs. \cite{TMSM12,MMAP13}, we reconstruct finite dimensional quantum theory in Theorem \ref{firsttheorem}, from Postulates \hyperref[P1first]{1} -- \hyperref[P7]{7}.

On the other hand, if we assume that spacetime is Minkowski in $1+n$ dimensions, leaving $n$ as a free variable, we show in Lemma \ref{adlemma1} from Postulates \hyperref[P1first]{1} -- \hyperref[P3]{3} that the states and measurements for the elementary system correspond to an Euclidean ball of dimension $n$, which generalizes the qubit Bloch ball. Then, using this result and the results of Refs. \cite{TMSM12,MMAP13,MMPA14}, we show in Theorem \ref{secondtheorem} from Postulates \hyperref[P1first]{1} -- \hyperref[P7]{7} that the elementary system is the qubit, the number of spatial dimensions is $n=3$, and any physical system of any finite dimension can be described by finite dimensional quantum theory.

%As is well known, Minkowski spacetime is derived from a set of physical principles \cite{}. For completeness of this paper, we present such a derivation in section \ref{}, and leave the number of spatial dimensions $n$ as a free variable.

Our Postulate \hyperref[P1first]{1} can be extended for arbitrary spacetimes with arbitrary groups of symmetry transformations $\mathfrak{G}$, as we do in Postulate \hyperref[P1']{1'}, ``Structure from the Spacetime Symmetries''. We do not do it in this paper, but we think it would be very interesting to investigate the implications of this postulate for quantum theory in curved spacetimes or in modifications of general relativity. As mentioned above, a compelling motivation for this is the problem of unifying gravity and quantum theory.

The rest of this paper is organized as follows. We present an introduction to the framework of GPTs and give some examples in section \ref{GPTs}. Section \ref{derivingMinkowski} gives an introduction to Minkowski spacetime and the Poincar{\'e} group in $1+n$ dimensions, and presents a well known derivation of Minkowski spacetime from physical principles, where we leave the number of spatial dimensions $n$ as a free variable. In section \ref{secmainpostulate} we present and discuss our main postulates, Postulates \hyperref[P1first]{1} and \hyperref[P1']{1'}. Section \ref{particlepostulates} introduces and discusses our other postulates. In section \ref{qubitsection} we obtain our main technical result, Lemma \ref{newlemma}, and use it to derive  in Lemma \ref{adlemma1} that the elementary system corresponds to an Euclidean ball of dimension $n$, and to reconstruct the qubit Bloch ball in Lemma \ref{qubitlemma} if it is assumed that $n=3$. Our Theorems \ref{firsttheorem} and \ref{secondtheorem}, reconstructing finite dimensional quantum theory and the number of spatial dimensions, are given and proved in section \ref{quantumtheorysection}. We conclude discussing our results and presenting some open problems in section \ref{discussion}.

\section{General probabilistic theories}
\label{GPTs}
%\subsection{Our model: general probabilistic theories}

%We initially consider a massive particle in spacetime. We then consider the case of $N$ massive particles, for arbitrary $N\in\mathbb{N}$. We assume that the particles' states are described within the framework of general probabilistic theories (GPTs), which is a broad framework of probabilistic theories that includes classical and quantum theory as special cases. We assume that a GPT describes the particles' spacetime degrees of freedom, like their spacetime coordinates and their momentum, as well as some internal degrees of freedom, like the spin. We focus on the momentum degrees of freedom, which in general may be continuous dimensional, and on the internal degrees of freedom, which we assume to be discrete and of finite dimension. This is motivated by the fact that in quantum theory, the particles' spin degrees of freedom are finite dimensional. Thus, we assume that the particles' internal degrees of freedom are described by a finite dimensional GPT, while the particles' global degrees of freedom, which include the momentum and the internal degrees of freedom, are described in general by a continuous dimensional GPT.

\subsection{Finite dimensions}

We provide a brief introduction to the framework of general probabilistic theories (GPTs) of finite dimensions \cite{H01,B07,DB09,CDP10,CDP11,MM11,H11,TMSM12,MMAP13,MPP15}.
% which we apply to describe the state space $\mathcal{S}$ for the particles internal degrees of freedom.
A GPT predicts the outcome probabilities for all possible experiments that can be implemented in any physical system described by the theory \cite{H01}. An experiment comprises \emph{preparations} and \emph{operations} \cite{B07}. A physical system is prepared in a state $\zeta$, then an operation is applied, in general changing $\zeta$ to another state $\zeta_j$ with probability $q_j$, hence preparing $\zeta_j$ with probability $q_j$. The outcome $j$ of the operation can be observed and recorded at the end of the experiment. Thus, an operation can be regarded as a measurement or as a transformation \cite{B07}.

%\subsubsection{States}

Given a preparation procedure, the \emph{state} $\zeta\in\mathcal{S}$ is a mathematical object that allows us to compute the outcome probabilities for all possible measurements \cite{H01}. With this definition, a state could comprise the list of outcome probabilities for all possible measurements. An important assumption in the framework of GPTs is that the states can be determined by the outcome probabilities of a finite set of measurements, called \emph{fiducial measurements}, with finite sets of possible outcomes. We note that this holds in classical probabilistic theory and quantum theory of finite dimension. For example, any quantum state of a qubit can be determined by the outcome probabilities of three different measurements, along the $x$, $y$ and $z$ axes in the Bloch sphere, for instance.

Stated as the ``Probabilities'' axiom by Hardy \cite{H01}, in the framework of GPTs we assume that by preparing an ensemble of $N$ identical systems in the same state $\zeta$ and performing the same measurement on each system of the ensemble, the relative frequencies obtained tend to the same value, which we call probability, in the limit that $N$ tends to infinity. This is such a fundamental assumption that it is commonly considered as part of the background framework of GPTs. 

The state space $\mathcal{S}$ for a physical system is the set of possible states in which the system can be prepared or transformed. We assume that $\mathcal{S}$ is closed and convex. Convexity means that one should be able to prepare any convex combination $\sum_j q_j\zeta_j$ of states $\zeta_j\in\mathcal{S}$ by preparing $\zeta_j$ with probability $q_j$, where $q_j\geq 0$ and $\sum_j q_j =1$. It follows that $\mathcal{S}$ can be embedded in a vector space $V$ over $\mathbb{R}$. The extreme points of $\zeta$ are called \emph{pure states}. These cannot be written as convex combinations of other states. The state space $\mathcal{S}$ is the convex hull of the pure states. The \emph{mixed states} are the states  that are not pure \cite{H01,B07}.

%\subsubsection{Measurements and Transformations}

We assume \emph{causality}, stating that the probability of preparing a system in a given state is independent of what measurements will be implemented after the preparation \cite{CDP10}. As stated by Refs. \cite{CDP10,CDP11}, causality is an axiom that does not need to hold in general physical theories, for example in theories of quantum gravity without definite causal structures. However, this principle is so fundamental in our current understanding of physics that it is commonly assumed in the framework of GPTs.

The outcome probabilities are given by maps $\varepsilon: \mathcal{S} \rightarrow [0,1]$, called \emph{effects}. A measurement is a set of effects that adds to the \emph{unit effect} $u$, which satisfies $u(\zeta)=1$ for all $\zeta\in\mathcal{S}$. Causality implies that there is a single unit effect $u$ for a given system \cite{CDP10}. The \emph{zero effect} $\varepsilon_\bold{0}$ satisfies $\varepsilon_\bold{0}(\zeta)=0$ for all $\zeta\in\mathcal{S}$; it corresponds to measuring a property that occurs with zero probability for all states. The set of effects $\mathcal{E}$ includes both $u$ and $\varepsilon_\bold{0}$.

The set $\mathcal{T}$ of allowed \emph{transformations} is a set of maps $\tau$ that must leave the state space invariant, that is, $\tau: \mathcal{S}\rightarrow \mathcal{S}$. The set of reversible transformations $\mathcal{R}\equiv\{\tau\in\mathcal{T}\vert \tau^{-1}\in\mathcal{T}, (\tau^{-1}\circ \tau)(\zeta)=\zeta~ \forall 
\zeta\in\mathcal{S}\}$ is an important subset of the allowed transformations. As explained below, the set of transformations $\mathcal{T}$ can be considered as a set of linear maps on $V$. It is straightforward to see from this property that reversible transformations take pure states into pure states. 

Similarly to $\mathcal{S}$, we assume that $\mathcal{E}$ and $\mathcal{T}$ are convex sets. One way to prepare the state $\zeta=\sum_j q_j\zeta_j$ is to prepare $\zeta_j$ with probability $q_j$ and then forget $j$. Thus, the outcome probabilities and the transformed state must satisfy $\varepsilon(\zeta)= \sum_j q_j\varepsilon(\zeta_j)$ and $\tau(\zeta)= \sum_j q_j \tau (\zeta_j)$, respectively. It follows that the set of effects $\mathcal{E}$ and the set of transformations $\mathcal{T}$ can be considered as sets of linear maps on $V$ \cite{H01,B07}. The set of all linear maps from $V$ to $\mathbb{R}$ is called the \emph{dual space} of $V$ and is denoted by $V^*$. 

In general, we have $\mathcal{E}\subseteq \mathcal{E}_{\text{norm}}\equiv \{ \varepsilon\in V^*\vert 0\leq \varepsilon (\zeta) \leq 1 ~ \forall \zeta\in\mathcal{S}\}$, where $\mathcal{E}_{\text{norm}}$ is the set of \emph{normalized}, or \emph{proper}, effects.  We see that $\mathcal{E}_{\text{norm}}$ is the set of effects that give valid outcome probabilities for all states $\zeta\in\mathcal{S}$. In general, the set of effects $\mathcal{E}$ can be a proper subset of $\mathcal{E}_{\text{norm}}$. The \emph{no-restriction hypothesis} \cite{CDP10} states that $\mathcal{E}=\mathcal{E}_{\text{norm}}$. This a rather strong and unjustified assumption, which has been used in some axiomatic reconstructions of finite dimensional quantum theory (e.g. \cite{CDP10,MM11,MMAP13}) and whose relaxation has been investigated (see e.g Ref. \cite{JL13}). We do not make this assumption here.
%An important observation is that if $\varepsilon\in\mathcal{E}$ and $\tau\in\mathcal{T}$, then a consistent theory must satisfy that $\varepsilon\circ T\in\mathcal{E}$. This is because applying a transformation $\tau$ and then a measurement that includes the effect $\varepsilon$ is clearly a valid operation that corresponds to a measurement, hence it must correspond to a valid measurement.

Since $V$ is finite dimensional, $V$ and $V^*$ are isomorphic to $\mathbb{R}^{d+1}$, for some $d\in\mathbb{N}$ (the case $d=0$ being trivial). It follows that the states and effects are given by $\zeta, \varepsilon \in \mathbb{R}^{d+1}$ with $\zeta^{\text{t}}=(\zeta_0,\zeta_1,\ldots,\zeta_d)$, $\varepsilon^{\text{t}}=(\varepsilon_0,\varepsilon_1,\ldots,\varepsilon_d)$, where `t' denotes transposition. The map $\varepsilon(\zeta)$ is an Euclidean dot product $\varepsilon(\zeta)= \varepsilon\cdot \zeta =\sum_{j=0}^{d}\varepsilon_j\zeta_j$, and the transformations $\tau$ correspond to $(d+1)\times (d+1)$ real matrices \cite{H01,B07}. Without loss of generality we take the unit effect, the zero effect and the state space by $u=\bigl(\begin{smallmatrix}
1\\ \bold{0}
\end{smallmatrix} \bigr)$, $\varepsilon_{\bold{0}}=\bigl(\begin{smallmatrix}
0\\ \bold{0}
\end{smallmatrix} \bigr)$ and $\mathcal{S}=\bigl\{\zeta\equiv\bigl(\begin{smallmatrix}
1\\ \tilde{\zeta}
\end{smallmatrix} \bigr)\vert  \tilde{\zeta}\in\tilde{\mathcal{S}}\bigr\}$, respectively, where $\bold{0}$ is the null vector in $\mathbb{R}^d$, $\tilde{\mathcal{S}}\subset \mathbb{R}^d$ is a convex set, and $d\in\mathbb{N}$ takes the minimum value such that $\tilde{\mathcal{S}}\subset \mathbb{R}^d$. We call $d$ the \emph{dimension} of $\mathcal{S}$, or the dimension of the GPT.

\subsubsection{Composite systems}
\label{app:b}
%In this section we describe how to characterize the state spaces for composite systems. In particular we use the condition of tomographic locality (Postulate 5).

Consider two systems $A$ and $B$ with respective state spaces $\mathcal{S}_A\subset\mathbb{R}^{d_A+1}$ and $\mathcal{S}_B\subset\mathbb{R}^{d_B+1}$, and respective effect spaces $\mathcal{E}_A$ and $\mathcal{E}_B$. Let $\mathcal{S}_{AB}$ and $\mathcal{E}_{AB}$ be the state space and the effect space of the composite system $AB$, respectively. A few conditions are imposed on $\mathcal{S}_{AB}$ and $\mathcal{E}_{AB}$ in the literature of GPTs \cite{H01,B07,BBLW07,CDP10,MPP15}. 

The first condition is the \emph{no-signalling principle}, stating that for any state $\phi\in\mathcal{S}_{AB}$, the outcome probabilities of measurements performed on $A$ should be independent of the measurements implemented on $B$ and vice versa. If this condition were violated, by choosing and performing local measurements on $A$ and $B$ at spacelike separation, information could be communicated faster than light. This would violate relativistic causality. It is natural for us to impose the no-signalling principle, in particular because we will be explicitly assuming in this paper that physics takes place in Minkowski spacetime.

The second condition is \emph{tomographic locality}, which we will present later as Postulate \hyperref[P5]{5}. This condition says that any state $\phi\in\mathcal{S}_{AB}$ can be determined by the  outcome probabilities of local measurements on $A$ and $B$, which are given by $(\varepsilon\otimes \varepsilon')(\phi)$ for some  $\varepsilon\in\mathcal{E}_A$ and $\varepsilon'\in\mathcal{E}_B$. It follows from these conditions that $\mathcal{S}_{AB}\subset\mathbb{R}^{d_{A}+1}\otimes \mathbb{R}^{d_{B}+1}$ \cite{B07}.

Additionally, it is assumed that product states $\zeta\otimes \zeta'$ and product effects $\varepsilon\otimes \varepsilon'$ are allowed, for all $\zeta\in\mathcal{S}_A$, $\zeta'\in\mathcal{S}_B$, $\varepsilon\in\mathcal{E}_A$, and $\varepsilon'\in\mathcal{E}_B$. This is because the possibility of preparing system $A$ in some state $\zeta$ or applying a measurement that includes some effect $\varepsilon$ should be independent of whether there is or not another system $B$. Formally, these conditions can be expressed by $\mathcal{S}_A\otimes_{\text{min}}\mathcal{S}_B\subseteq \mathcal{S}_{AB}$ and $\mathcal{E}_A\otimes_{\text{min}}\mathcal{E}_B\subseteq \mathcal{E}_{AB}$, where  $\mathcal{S}_A\otimes_{\text{min}}\mathcal{S}_B\equiv \text{convex hull}\{\zeta\otimes \zeta'\vert \zeta\in\mathcal{S}_A,\zeta'\in\mathcal{S}_B\}$ is the \emph{minimal tensor product}, and where we define an analogous quantity for the space of effects by $\mathcal{E}_A\otimes_{\text{min}}\mathcal{E}_B\equiv \text{convex hull}\{\varepsilon\otimes \varepsilon'\vert \varepsilon\in\mathcal{E}_A,\varepsilon'\in\mathcal{E}_B\}$. Together with the conditions above, these assumptions imply that $\mathcal{S}_{AB}\subseteq \mathcal{S}_A\otimes_{\text{max}}\mathcal{S}_B$, where $\mathcal{S}_A\otimes_{\text{max}}\mathcal{S}_B\equiv \{\phi\vert (u_A\otimes u_B)(\phi) = 1, (\varepsilon\otimes \varepsilon')( \phi)\geq 0, ~ \forall \varepsilon\in\mathcal{E}_A,\varepsilon'\in\mathcal{E}_B\}$ is the \emph{maximal tensor product}.

The unit effect acting on $AB$ is $u_{AB}=u_A\otimes u_B$, which is a property that can be deduced from the causality condition, introduced above \cite{CDP10}. In general, we have $\mathcal{S}_A\otimes_{\text{min}}\mathcal{S}_B\subseteq \mathcal{S}_{AB}\subseteq\mathcal{S}_A\otimes_{\text{max}}\mathcal{S}_B$. A state is \emph{separable} if it can be written as a convex combination of product states.
A state is \emph{entangled} if it is not separable. Thus, by definition, the set of separable states corresponds to $\mathcal{S}_A\otimes_{\text{min}}\mathcal{S}_B$.
If $A$ or $B$ is a classical system then $\mathcal{S}_{AB}=\mathcal{S}_A\otimes_{\text{min}}\mathcal{S}_B= \mathcal{S}_A\otimes_{\text{max}}\mathcal{S}_B$ and thus there are not entangled states \cite{B07,BBLW07}.

When considering composite systems $AB$, the set of allowed transformations $\mathcal{\tau}_A$ on system $A$ must satisfy that for every system $B$ and for every transformation $\tau_A\in\mathcal{T}_A$, it holds that $\tau_A\otimes I_B : \mathcal{S}_{AB}\rightarrow \mathcal{S}_{AB}$, where $I_B$ is the identity map acting on system $B$. This condition corresponds to the fact that in quantum theory the maps must be completely positive. Additionally, we require that $\tau_A:\mathcal{S}_{A}\rightarrow \mathcal{S}_{A}$ for every $\tau_A\in\mathcal{T}_A$, as previously discussed \cite{B07}.

The considerations above apply to an arbitrary number $N$ of systems. For example, a system composed of three subsystems $A$, $B$ and $C$ can be considered as a system composed of two subsystems $AB$ and $C$.

\subsection{Examples of finite dimensional general probabilistic theories}

\subsubsection{Finite dimensional classical probabilistic theory}

Classical probabilistic theory of finite dimensions is an example of a finite dimensional GPT. In this case, for a classical system with $N+1$ possible outcomes, the state space $\mathcal{S}$ and the space of effects $\mathcal{E}$ have $N+1$ pure states $\zeta_0,\zeta_1,\ldots,\zeta_N$ and $N+1$ extremal effects $\varepsilon_0,\varepsilon_1,\ldots,\varepsilon_N$, respectively, satisfying $\sum_{i=0}^N\varepsilon_i=u$ and $\varepsilon_i(\zeta_j)=\delta_{i,j}$, for all $i,j\in\{0,1,\ldots,N\}$. That is, there is a single measurement, given by the set of effects $\{\varepsilon_j\}_{j=0}^N$, that perfectly distinguishes the pure states.

It follows that the state space $\mathcal{S}$ satisfies that $\tilde{\mathcal{S}}$ is a regular $N-$simplex in $\mathbb{R}^{N}$. That is, $\mathcal{S}$ is the convex hull of the pure states $\zeta_i =\bigl(\begin{smallmatrix}
1\\ \tilde{\zeta}_i
\end{smallmatrix} \bigr)$, for all $i\in\{0,1,2,\ldots,N\}$, where the set of vectors $\{\tilde{\zeta}_i\}_{i=0}^N$ represents the vertices of a $N-$simplex in $\mathbb{R}^{N}$. Thus, we have that the GPT's dimension is $d=N$. The space of effects $\mathcal{E}$ is the convex hull of the unit effect $u$, the zero effect $\varepsilon_{\bold{0}}$ and the extremal effects $\varepsilon_j$, for all $j\in\{0,1,2,\ldots,N\}$. The set of reversible transformations $\mathcal{R}$ permutes the pure states.

The classical bit corresponds to the case of two possible outcomes, i.e. $N=1$. In this case, we have 
\begin{equation}
\label{classicalbit}
\zeta_i=\begin{pmatrix}\!
1\!\\ \!(-1)^{i+1}\!
\end{pmatrix} , \quad\!\!\!\!\! \varepsilon_i=\frac{1}{2}\begin{pmatrix}
\!1\!\\ \!(-1)^{i+1}\!
\end{pmatrix},  \quad\!\!\!\!\! u=\begin{pmatrix}
\!1\!\\ \!0\!
\end{pmatrix}, \quad\!\!\!\!\! \varepsilon_\bold{0}=\begin{pmatrix}
\!0\!\\ \!0\!
\end{pmatrix},
\end{equation}
for all $i\in\{0,1\}$. The state space $\mathcal{S}$ is given by the convex hull of $\zeta_0$ and $\zeta_1$, which gives a line segment (see Fig. \ref{figbit}). The space of effects $\mathcal{E}$ is given by the convex hull of $\varepsilon_0$, $\varepsilon_1$, $u$ and $\varepsilon_\bold{0}$. There is a single reversible transformation, given by the reflection
\begin{equation}
\label{classicalbitreversible}
R=\begin{pmatrix}
1 & 0\\ 0 & -1\end{pmatrix}.
\end{equation}

\begin{figure}
\includegraphics[scale=0.26]{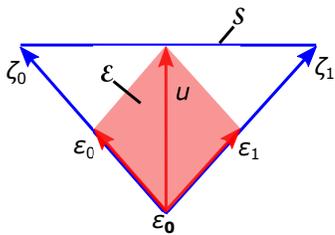}
 \caption{\label{figbit} \textbf{The GPT of a classical bit.} The state space $\mathcal{S}$ (blue horizontal line) is the convex hull of the pure states $\zeta_0$ and $\zeta_1$ (blue long diagonal arrows). The space of effects $\mathcal{E}$ (red filled area) is the convex hull of the zero effect $\varepsilon_{\bold{0}}$ (the origin), the unit effect $u$ (red vertical arrow), and the extremal effects $\varepsilon_0$ and $\varepsilon_1$ (red short diagonal arrows). The states $\zeta_0$ and $\zeta_1$, and the effects $\varepsilon_0$, $\varepsilon_1$, $u$ and $\varepsilon_{\bold{0}}$ are given by (\ref{classicalbit}), and all have the same origin. The vertical and horizontal axes represent the first and second entries of these vectors, respectively.}
\end{figure}

%An explicit representation of the vectors $\tilde{\zeta}_i$ is given, for example, by $\tilde{\zeta}_0=(1,0,\ldots,0)$,  $\tilde{\zeta}_1=(0,1,0,\ldots,0)\ldots$, $\tilde{\zeta}_N=(0,\ldots,0,1)$

\subsubsection{Finite dimensional quantum theory}

Finite dimensional quantum theory can be formulated in the framework of finite dimensional GPTs presented above too. A quantum density matrix $\rho$ acting on a complex Hilbert space $\mathcal{H}$ of finite dimension $d_\text{H}$ is a $d_\text{H}\times d_\text{H}$ Hermitian matrix with unit trace, and hence has $d_\text{H}^2-1$ real degrees of freedom. The corresponding GPT states $\zeta\in\mathcal{S}$ and effects $\varepsilon\in\mathcal{E}$ are associated to the $d=d_\text{H}^2-1$ real degrees of freedom of density matrices $\rho$ and measurement operators $M$ acting on $\mathcal{H}$, respectively, when expressed in a basis for the space of Hermitian matrices on $\mathcal{H}$ \cite{H01}. The reversible transformations correspond to the maps $\rho\rightarrow U\rho U^{\dagger}$ on density matrices, where $U\in\text{SU}(d_\text{H})$.

The case $d_\text{H}=2$ corresponds to the qubit. The qubit can be described by the GPT of a $3-$dimensional Euclidean ball, where the state space $\mathcal{S}$ is given by the Bloch ball and the set of pure states is given by the Bloch sphere. 

\subsubsection{Euclidean $d-$balls}

The GPTs for the qubit Bloch ball and more general $d-$dimensional Euclidean balls
%, also called Euclidean $d-$balls,
are illustrated in Fig. \ref{fig1}. Since Euclidean $d-$balls are natural generalizations of the qubit Bloch ball, these theories have been investigated before. For example, Ref. \cite{MPP15} used these theories to provide examples of hyperdense coding, where transmission of a system $A$ that is entangled with a system $B$ held by the receiver can communicate more than twice the amount of bits that the system $A$ alone can communicate, violating the quantum bound achieved by quantum superdense coding \cite{sdc}. Ref. \cite{MMPA14} used some physical conditions, namely continuous reversibility and tomographic locality (Postulates \hyperref[P4]{4} and \hyperref[P5]{5} presented in section\ref{particlepostulates}), to show that the only Euclidean $d-$balls with bipartite entanglement correspond to $d=3$, i.e to the qubit Bloch ball. This result was used in the reconstruction of finite dimensional quantum theory of Ref. \cite{MMAP13} and will be used here too. Other derivations of finite dimensional quantum theory have also used Euclidean $d-$balls (e.g. \cite{DB09,MM11}).

\begin{figure}
\includegraphics[scale=0.26]{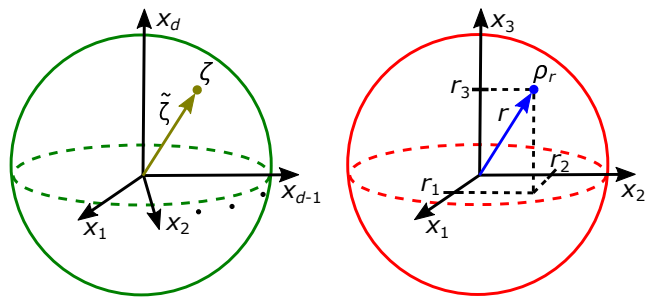}
 \caption{\label{fig1} \textbf{The GPTs of an Euclidean $d-$ball and the qubit Bloch ball.} Left: for all $d\in\mathbb{N}$, the GPT state space defining 
an \emph{Euclidean $d-$ball} is $\mathcal{S}_{\text{ball}}^{(d)}\equiv\bigl\{\zeta=\bigl(\begin{smallmatrix}
1\\ \tilde{\zeta}
\end{smallmatrix} \bigr)\vert \tilde{\zeta}\in\mathbb{R}^d,\lVert \tilde{\zeta} \rVert\leq 1\bigr\}$, where `$\lVert \cdot \rVert$' denotes the Euclidean norm. Its set of pure states is given by the \emph{Euclidean $d-$sphere}: $\mathcal{S}_{\text{sphere}}^{(d)}\equiv\bigl\{\zeta=\bigl(\begin{smallmatrix}
1\\ \tilde{\zeta}
\end{smallmatrix} \bigr)\vert \tilde{\zeta}\in\mathbb{R}^d,\lVert \tilde{\zeta} \rVert= 1\bigr\}$.  The corresponding set of effects $\mathcal{E}_{\text{ball}}^{(d)}$ is the convex hull of the zero effect $\varepsilon_{\bold{0}}\equiv\bigl(\begin{smallmatrix}
0\\ \bold{0}
\end{smallmatrix} \bigr)$, the unit effect $u\equiv\bigl(\begin{smallmatrix}1\\ \bold{0}
\end{smallmatrix} \bigr)$, and the extremal effects $\varepsilon_{v}\equiv\frac{1}{2}v$, where $v\equiv\bigl(\begin{smallmatrix}
1\\ \tilde{v}
\end{smallmatrix} \bigr)\in \mathcal{S}_{\text{sphere}}^{(d)}$, and where $\bold{0}\in\mathbb{R}^d$ is the null vector. Right: a qubit density matrix can be expressed by $\rho_r=\frac{1}{2}(I+r\cdot{\sigma})$, where $r\cdot \sigma=\sum_{j=1}^3r_j\sigma_j$, $\{\sigma_j\}_{j=1}^3$ are the
Pauli matrices and $r=(r_1,r_2,r_3)^{\text{t}}\in\mathbb{R}^3$ is the \emph{Bloch vector} \cite{NielsenandChuangbook}. The set of qubit states corresponds to the \emph{Bloch ball}: $\lVert r\rVert \leq 1$. The set of pure qubit states corresponds to the \emph{Bloch sphere}: $\lVert r\rVert = 1$. The Bloch ball and Bloch sphere are respectively represented by the GPT state spaces $\mathcal{S}_{\text{BB}}\equiv \mathcal{S}_{\text{ball}}^{(3)}$ and $\mathcal{S}_{\text{BS}}\equiv \mathcal{S}_{\text{sphere}}^{(3)}$. The qubit measurement statistics are reproduced by states from $\mathcal{S}_{\text{BB}}$ and measurements with effects from $\mathcal{E}_{\text{BB}}\equiv \mathcal{E}_{\text{ball}}^{(3)}$. For example, the probability that a qubit pure state $\rho_{\tilde{\zeta}}$ with unit Bloch vector $\tilde{\zeta}$ is projected into a pure state $\rho_{\tilde{v}}$ with unit Bloch vector $\tilde{v}$ is $\text{tr}(\rho_{\tilde{\zeta}}\rho_{\tilde{v}})=\frac{1}{2}(1+\tilde{\zeta}\cdot \tilde{v})$, which equals $\varepsilon_{v}(\zeta)$. The qubit unitary dynamics SU$(2)$ corresponds to the GPT set of reversible transformations $\mathcal{R}_{\text{BB}}\equiv\Bigl\{\tau\equiv\Bigl(\begin{smallmatrix}
  1 & 0 \\
  0 & \tilde{\tau} 
 \end{smallmatrix}\Bigr)\big\vert \tilde{\tau}\in\text{SO}(3)\Bigr\}$, which connects all pure states in $\mathcal{S}_{\text{BB}}$, i.e. all states in $\mathcal{S}_{\text{BS}}$.
}
\end{figure}

\subsubsection{Polygon theories}
\label{secpolygon}
Another simple example of GPTs is given by the \emph{polygon theories} introduced in Ref. \cite{JGBB11}. In these theories, the state space $\mathcal{S}$ of a single system is such that $\tilde{\mathcal{S}}$ is a regular polygon of $N$ vertices (see Fig. \ref{figpol}). The dimension of these GPTs is therefore $d=2$. For a given integer $N\geq3$, $\mathcal{S}$ is the convex hull of $N$ pure states
\begin{equation}
\label{polygonstates}
\zeta_i=\begin{pmatrix}
1\\ r_N \cos\bigl(\frac{2\pi(i+1)}{N}\bigr) \\ r_N \sin\bigl(\frac{2\pi(i+1)}{N}\bigr)
\end{pmatrix},
\end{equation}
for all $i\in\{0,1,\ldots,N-1\}$, where $r_N=\sqrt{\sec(\frac{\pi}{N})}$. The unit effect and the zero effect are respectively
\begin{equation}
\label{polygonunit}
u=\begin{pmatrix}
1\\ 0 \\ 0
\end{pmatrix} \quad \text{ and } \quad\varepsilon_\bold{0}=\begin{pmatrix}
0\\ 0 \\ 0
\end{pmatrix}.
\end{equation}
If $N$ is even, the set of normalized effects $\mathcal{E}_\text{norm}$ is the convex hull of $\varepsilon_\bold{0}$, $u$ and the effects
\begin{equation}
\label{polygoneffectseven}
\varepsilon_i=\frac{1}{2}\begin{pmatrix}
1\\ r_N \cos\bigl(\frac{(2i+1)\pi}{N}\bigr) \\ r_N \sin\bigl(\frac{(2i+1)\pi}{N}\bigr)
\end{pmatrix}, 
\end{equation}
for all $i\in\{0,1,\ldots,N-1\}$. If $N$ is odd, $\mathcal{E}_\text{norm}$ is the convex hull of $\varepsilon_\bold{0}$, $u$ and the effects
\begin{equation}
\label{polygoneffectsodd}
\varepsilon_i=\frac{1}{1+r_N^2}\begin{pmatrix}
1\\ r_N \cos\bigl(\frac{2\pi(i+1)}{N}\bigr) \\ r_N \sin\bigl(\frac{2\pi(i+1)}{N}\bigr)
\end{pmatrix} 
\end{equation}
and $\bar{\varepsilon}_i= u -\varepsilon_i$, for all $i\in\{0,1,\ldots,N-1\}$. As mentioned above, the space of effects $\mathcal{E}$ must satisfy $\mathcal{E}\subseteq\mathcal{E}_\text{norm}$. Let us consider in this example that $\mathcal{E}=\mathcal{E}_\text{norm}$.

\begin{figure}
\includegraphics[scale=0.26]{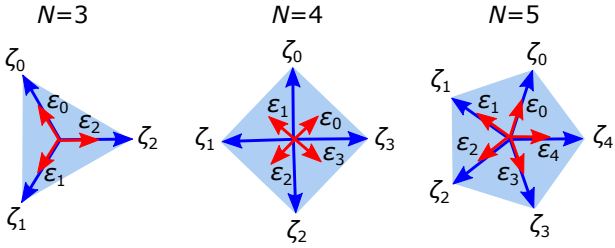}
 \caption{\label{figpol} \textbf{Polygon GPTs.} We illustrate the second and third entries of the pure states $\zeta_0,\zeta_1,\ldots,\zeta_N$ (blue long arrows) given by (\ref{polygonstates}) and the effects $\varepsilon_0,\varepsilon_1,\ldots,\varepsilon_N$ (red short arrows) given by (\ref{polygoneffectseven}) and (\ref{polygoneffectsodd}), for the cases $N=3,4,5$. The state space $\mathcal{S}$ (blue filled area) is the convex hull of the pure states. The cases $N=3$ and $N=4$ correspond to a classical trit and to half of a Poposcu-Rohrlich box, respectively. The state space $\mathcal{S}$ and space of effects $\mathcal{E}$ for these theories remain invariant under rotations $R(j\theta)$ in the illustrated plane by angles $j\theta$ with $\theta=\frac{2\pi}{N}$, for all $j\in\{0,1,\ldots,N-1\}$ (see (\ref{polygonrot})).}
\end{figure}

An important property of these theories that we will use in an example in section \ref{secmainpostulate} (see Fig. \ref{figtoyspacetime}) is that a rotation $R(j\theta)$ of an angle $j\theta$ in the plane of the last two dimensions, with $\theta=\frac{2\pi}{N}$, satisfies
\begin{eqnarray}
\label{polygonrot}
R(j\theta)\zeta_i&=&\zeta_{i+j \text{ mod }N},\nonumber\\
R(j\theta)\varepsilon_i&=&\varepsilon_{i+j \text{ mod }N},\nonumber\\
R(j\theta)\bar{\varepsilon}_i&=&\bar{\varepsilon}_{i+j \text{ mod }N},
\end{eqnarray}
for all $i\in\{0,1,\ldots,N-1\}$ and arbitrary integer $j$. Thus, we see that the rotations $R(j\theta)$ leave $\mathcal{S}$ and $\mathcal{E}$ invariant, and when applied on both $\mathcal{S}$ and $\mathcal{E}$ leave the outcome probabilities $\varepsilon(\zeta)$ invariant. It is clear that $R(j\theta)=R\bigl((j\text{ mod } N)\theta\bigr)$, for any integer $j$. Thus, these rotations are given by the set $\{R(j\theta)\}_{j=0}^{N-1}$. We note that this set includes their inverse transformations, as 
\begin{equation}
\label{polygonrotinv}
R(j\theta)R\bigl(\bigl((N-j) \text{ mod } N\bigl)\theta\bigr)= I,
\end{equation}
for all $j\in\{0,1,\ldots,N-1\}$. It follows that the set $\{R(j\theta)\}_{j=0}^{N-1}$ is a group of reversible transformations.

A classical trit corresponds to the case $N=3$. In this case we can see that $\varepsilon_0+\varepsilon_1+\varepsilon_2=u$ and $\varepsilon_j(\zeta_i)=\delta_{i,j}$, for all $i,j\in\{0,1,2\}$. Thus, the measurement $\{\varepsilon_j\}_{j=0}^2$ completely determines the pure states. Another important example is given by $N=4$, as explained below. Moreover, the limit $N\rightarrow\infty$ corresponds to the equatorial plane of the qubit Bloch ball, i.e. to a qubit in a real Hilbert space.

By defining entangled states for a pair of system locally described by the same polygon theory, Ref. \cite{JGBB11} investigated the consequences that gradually weakening the superposition principle,
from $N\rightarrow\infty$ to $N=3$, have on the degree of violation of the Clauser-Horne-Shimony-Holt (CHSH) Bell inequality \cite{CHSH69}. The communication capabilities of polygon theories were investigated in Ref. \cite{MP14}.

\subsubsection{Box world}

The case $N=4$ in the polygon theories described above represents a particular system of a GPT called \emph{box world} \cite{B07,SB10,AS14}, introduced by Barrett \cite{B07} as \emph{generalized non-signalling theory}, in which all non-signalling correlations are achieved. An entangled state of two systems, where each systems is described by a polygon theory with $n=4$ is called a Popescu-Rohrlich (PR) box \cite{PR94}. A PR box achieves the maximum violation of the CHSH Bell inequality \cite{CHSH69} that is mathematically possible, while satisfying the no-signalling principle, and thus obtaining a violation greater than the quantum Tsirelson bound \cite{C80}. 

Box world has greater capabilities than quantum theory for some tasks. For example, as mentioned above, it is more non-local than quantum theory, in the sense that it allows correlations that violate Bell inequalities to higher values than those achieved by quantum correlations. It can also solve some communication complexity tasks trivially, i.e. with the transmission of a single bit \cite{vD05}.

On the other hand, box world is more restricted than quantum theory for other tasks. For instance, it has trivial reversible dynamics comprising only combinations of local operations, which relabel the measurements and outcomes, and permutations of local systems, and thus does not have reversible interactions between different systems \cite{GMCD10,AS14}. Furthermore, it does not have entanglement swapping, teleportation or dense coding \cite{SB10}.

\subsection{Continuous dimensions}
\label{continuous}
In this paper, we need to consider continuous dimensional GPTs in order to describe 
the momentum of a massive particle, which is a continuous dimensional physical variable. Continuous dimensional GPTs have been investigated before (e.g. \cite{DL70,E70,H16}).
In this work, we only assume that continuous dimensional GPTs satisfy the following very basic properties. Let $\mathscr{S}$, $\mathscr{E}$ and $\mathscr{T}$ be the sets of states, effects and allowed transformations for a physical system of continuous dimension, respectively. 
%For the same reasons as in the finite dimensional case, we assume that $\mathscr{S}$ is convex and can be embedded in a vector space $W$ over $\mathbb{R}$, and that $\mathscr{E}$ and $\mathscr{T}$ are convex sets of linear maps acting on $W$ \cite{H01,B07}. 
We require that $\hat{E}: \mathscr{S}\rightarrow [0,1]$ and $\hat{\bold{T}}: \mathscr{S}\rightarrow \mathscr{S}$ for all $\hat{E}\in\mathscr{E}$ and all $\hat{\bold{T}}\in\mathscr{T}$.

%\section{A set of physical principles and postulates for spacetime}
%\label{spacetimepostulates}

%We assume that the physical principles and postulates presented below hold. It follows, as our main result, Theorem \ref{theorem1}, shows, that spacetime is Minkowski in $1+3$ dimensions and that any finite dimensional system can be described by finite dimensional quantum theory. 

%\subsection{Derivation of Minkowski spacetime in $1+n$ dimensions}

%\subsection{A physical derivation of Minkowski spacetime in $1+n$ dimensions}

\section{A physical derivation of Minkowski spacetime in $1+n$ dimensions}
\label{derivingMinkowski}

In this section we present a well known physical derivation of Minkowski spacetime from physical principles \cite{LandauLifshitzbook,WeinbergGravitationandCosmologybook}. We leave the number of spatial dimensions $n$ as a free variable.

\subsection{Our model for spacetime}
\label{model}
%\section{Our model: spacetime}

\reff{We consider that spacetime is mathematically described by a real pseudo-Riemannian manifold of dimension $n+1$, for some unspecified $n\in\mathbb{N}$. That is, we assume that spacetime is a real differentiable manifold with a smooth, symmetric, non-degenerate metric tensor in every spacetime point. Broadly speaking, this allows us to do calculus at every spacetime point and to define geometric properties in the neighbourhood of every spacetime point, e.g., distance and curvature \cite{Schutzbook}. This will suffice to derive Minkowski spacetime from physical principles below.

Our definition includes, as a special case, the spacetimes of general relativity, which are four-dimensional Lorentzian manifolds, i.e., four-dimensional pseudo-Riemannian manifolds with metric signature $(-,+,+,+)$ \cite{Schutzbook}. We note that time-orientable Lorentzian manifolds comprise an important class of spacetimes of general relativity, in which, broadly speaking, past and future causal relations can be assigned unambiguously for every pair of causally connected spacetime points \cite{Waldbook}. Nevertheless, there are solutions to Einstein's equations with closed timelike curves, for which past and future cannot be unambiguously defined, and hence are not time-orientable (e.g., the G\"{o}del metric \cite{G49}).

In addition to including Minkowski spacetime and the spacetimes of general relativity as special cases, our model also allows for spacetimes arising in extensions of general relativity \cite{H19}. However, in this paper we focus on Minkowski spacetime, with the particularity that the number of spatial dimensions is arbitrary.}

We consider that time has one dimension and space has $n$ dimensions, for an arbitrary and unspecified $n\in\mathbb{N}$. \reff{That is, we define the first dimension in the spacetime manifold as the \emph{time} dimension and the others as the \emph{spatial} dimensions.} In a reference frame $F$, we call a \emph{spacetime event with coordinates $x=(x_0,x_1,\ldots,x_n)$}, or simply a \emph{spacetime point $x$}, to an event occurring at a location with coordinates $x_1,\ldots,x_n$ in space and at a time $t$, where $x_0=tc$ and $c$ is the speed of light in vacuum. We use units in which $c=1$, hence, $x_0=t$. We define $\vec{x}=(x_1,\ldots,x_n)$. Thus, we can equivalently write $x=(x_0,\vec{x})$.

In this paper we only consider reference frames that are $\emph{inertial}$, i.e. that move with respect to each other at constant velocity. Thus, when we say that a spacetime event has coordinates $x$ in a reference frame $F$ and $x'$ in another reference frame $F'$, we implicitly assume that $F$ and $F'$ are inertial.

\subsection{Minkowski spacetime and the Poincar{\'e} group in $1+n$ dimensions}

\reff{Minkowski spacetime in $1+n$ dimensions is a straightforward generalization of the standard four-dimensional Minkowski spacetime, with the metric given by
the $(1+n)\times (1+n)$ matrix $\eta$ with entries
\begin{equation}
\label{metric}
\eta_{\mu\nu}=\begin{cases}
 -1, \text{ if } \mu=\nu=0,\\
 1, \text{ if } \mu=\nu>0,\\
0, \text{ otherwise},
\end{cases}
\end{equation}
for all $\mu,\nu\in\{0,1,\ldots,n\}$. As Lemma \ref{Minkowski} below shows, Minkowski spacetime follows from a set of physical principles.

Moreover, as Lemma \ref{Poincare} below shows, the non-singular coordinate transformations between inertial reference frames in Minkowski spacetime are Poincar{\'e} transformations. A \emph{Poincar{\'e} transformation} $P(a,\Lambda)$ is a transformation on $x$ of the form
\begin{equation}
\label{neweq:l1}
x_\mu\xrightarrow{P(a,\Lambda)} x'_\mu=\sum_{\nu=0}^{n}\Lambda_{\mu\nu} x_\nu+a_\mu,
\end{equation}
where $x_\mu, a_\mu,\Lambda_{\mu\nu}\in\mathbb{R}$, and where the following relation holds:
\begin{equation}
\label{newnew1}
\sum_{\mu=0}^n\sum_{\nu=0}^n\eta_{\mu\nu}\Lambda_{\mu\alpha}\Lambda_{\nu\beta}=\eta_{\alpha\beta},
\end{equation}
for all $\mu,\alpha,\beta\in\{0,1,\ldots,n\}$. A \emph{Lorentz transformation} is a Poincar{\'e} transformation $P(\vec{0},\Lambda)$.

The Poincar{\'e} transformations form the \emph{Poincar{\'e} group}, denoted here by $\mathfrak{Poin}_{\text{full}}$. The Lorentz transformations form the \emph{Lorentz group}, denoted here by $\mathfrak{L}_{\text{full}}$, which is a subgroup of $\mathfrak{Poin}_{\text{full}}$. The \emph{proper orhochronous Lorentz group} $\mathfrak{L}$ is the subgroup of $\mathfrak{L}_{\text{full}}$ comprising the Lorentz transformations that are continuously connected to the identity, i.e., the spatial rotations and the Lorentz boosts. The \emph{proper orthochronous Poincar{\'e} group} $\mathfrak{Poin}$ comprises the Poincar{\'e} transformations $P(a,\Lambda)$ with $\Lambda\in\mathfrak{L}$. We are only interested here in Poincar{\'e} transformations that are continuously connected to the identity, i.e., $P(a,\Lambda)\in \mathfrak{Poin}$.

It is easy to see from (\ref{neweq:l1}) that an arbitrary pair of Poincar{\'e} transformations compose as
\begin{equation}
\label{eq:l1.1}
P(a',\Lambda')\circ P(a,\Lambda) =P(a'+\Lambda' a,\Lambda' \Lambda),
\end{equation}
for all $a', a \in\mathbb{R}^{1+n}$ and all $\Lambda', \Lambda \in\mathfrak{L}_{\text{full}}$.

}

\subsection{A set of physical principles and postulates for Minkowski spacetime in $1+n$ dimensions}

As Lemmas \ref{Minkowski} and \ref{Poincare} below show, the following well established principles and postulates imply that spacetime is Minkowski in $1+n$ dimensions and the coordinate transformations between inertial reference frames are Poincar{\'e} transformations.

\begin{principle}[\textbf{Relativity Principle}]
\label{RP}
The laws of physics are identical in all inertial reference frames.
\end{principle}

\begin{principle}[\textbf{Constancy of the Speed of Light}]
\label{CSL}
The speed of light in vacuum is a constant $c$ in all inertial reference frames.
\end{principle}

\begin{principle}[\textbf{Homogeneity of Space and Time}]
\label{HST}
The laws of physics are identical at all locations in space and at all times.
\end{principle}

\begin{principle}[\textbf{Isotropy of Space}]
\label{IS}
The laws of physics are identical in all directions of space.
\end{principle}

\begin{postulatea}[\textbf{Euclidean Spatial Distance}]
\label{ESD}
The spatial distance $\lvert\vec{y}-\vec{x}\rvert$ between the space locations $\vec{x}$ and $\vec{y}$ of respective spacetime points $x$ and $y$ in a reference frame $F$ can be determined with Euclidean geometry, i.e. $\lvert\vec{y}-\vec{x}\rvert=\sqrt{\sum_{i=1}^n (y_i-x_i)^2}$.
\end{postulatea}

\begin{postulatea}[\textbf{Non-singularity of Coordinate Transformations}]
\label{NCT}
If an event in spacetime has coordinates $x$ in a reference frame $F$ and $x'$ in a reference frame $F'$ then the coordinate transformation $x\rightarrow x'$ is non-singular.
\end{postulatea}

\reff{We note that the \emph{cosmological principle} comprises Principles \ref{HST} and \ref{IS} applied on sufficiently large scales in the universe \cite{Colesbook}. It implies the \emph{Robertson-Walker}, also called \emph{Friedmann-Lema\^{i}tre-Robertson-Walker}, cosmological models \cite{Waldbook}.}

\subsection{A well known physical derivation of Minkowski spacetime in $1+n$ dimensions}

The following lemmas are well known in the literature. These are usually stated for Minkowski spacetime in $1+3$ dimensions, but their generalization to $n$ spatial dimensions is straightforward. \reff{In Appendix \ref{app},} we present proofs that are close to the ones given by Refs. \cite{LandauLifshitzbook} and \cite{WeinbergGravitationandCosmologybook}, respectively. \reff{We note that there exist other axiomatic derivations of Minkowski spacetime and the curved spacetimes of general relativity (e.g. \cite{EPS72}).}

\begin{lemmaa}
\label{Minkowski}
If Principles \ref{RP} -- \ref{IS} and Postulate \ref{ESD} hold, then spacetime is Minkowski in $1+n$ dimensions. 
\end{lemmaa}

\begin{lemmaa}
\label{Poincare}
Let spacetime be Minkowski in $1+n$ dimensions and let Postulate \ref{NCT} hold. If $x$ and $x'$ are the coordinates of a spacetime event in inertial reference frames $F$ and $F'$, respectively, then the coordinate transformation $x\rightarrow x'$ is a Poincar{\'e} transformation.
\end{lemmaa}

%\subsection{Derivation of finite dimensional quantum theory and $n=3$}

%\section{A physical derivation of finite dimensional quantum theory and the number of spatial dimensions}

\section{Our main postulate: spacetime symmetries in general probabilistic theories}
\label{secmainpostulate}
In this section we present our main physical contribution to the literature of general probabilistic theories and axiomatic reconstructions of quantum theory. We present a postulate that suggests that there is a fundamental connection between the mathematical structures of spacetime and quantum theory. More precisely, we suggest that the Hilbert space structure of finite dimensional quantum theory has, at least to some extent, its origin in the symmetries of Minkowski spacetime. To the best of our knowledge, a postulate similar to ours has not been considered before \reff{in the framework of GPTs. However, a similar postulate was proposed by Svetlichny \cite{S00} in the framework of quantum logic.}

%In a related but different approach, Ref. \cite{GMD17} has considered the simultaneity of relativity in general probabilistic theories. At the end of this section we will differentiate the approach of Ref. \cite{GMD17} with ours and will discuss some important shortcomings of Ref. \cite{GMD17}.

%In the following discussion we assume that spacetime is arbitrary\reff{, as given by our model of section \ref{model}}. We will later consider the particular case of Minkowski spacetime.

%\reff{In this section we mainly discuss the case of arbitrary spacetimes, given by our model in section \ref{model}. We will be explicit when discussing the particular case of Minkowski spacetime.}

\reff{
Broadly speaking, in general relativity, the principle of general covariance states that the laws of physics are invariant under arbitrary smooth coordinate transformations, i.e., under arbitrary diffeomorphisms. This implies in particular that Einstein's equations are the same in all reference frames. In the particular case of Minkowski spacetime, Poincar{\'e} invariance states that the laws of physics remain invariant under arbitrary changes of inertial reference frames, which are given by the proper orthochronous Poincar{\'e} transformations. The intuition behind general covariance and Poincar{\'e} invariance is that different reference frames merely provide different descriptions of the same physical events \cite{Waldbook}. We generalize these ideas for arbitrary spacetimes in the framework of GPTs below.

Consider an arbitrary spacetime, given by the model of section \ref{model}. Let $\mathcal{O}$ be an observer with a reference frame $F$. That is, $F$ defines a coordinate system in which $\mathcal{O}$ describes physical events in spacetime. A different reference frame $F'$ can be associated to another observer $\mathcal{O}'$, or to the same observer using a different coordinate system. Let $G$ be a diffeomorphism that transforms the spacetime manifold as described in $F$ to the spacetime manifold as described in $F'$. The descriptions of physical events in the reference frames $F$ and $F'$ are in general different. But, it is sensible that the laws of physics should not change in different frames. In particular, the probability of any physical event is expected to remain invariant by changing reference frames.

Thus, let us assume that there exists a group $\mathfrak{G}$ of transformations $G$ of reference frames $G: F\rightarrow F'$ that does not change the laws of physics. More precisely, let the outcome probabilities for arbitrary physical events in spacetime remain invariant under transformations $G\in\mathfrak{G}$. In this case, we say that $\mathfrak{G}$ is a \emph{group of spacetime symmetries}. As mentioned above, in Minkowski spacetime, $\mathfrak{G}=\mathfrak{Poin}$. Below we define invariant general probabilistic theories under $\mathfrak{G}$. }

Consider a GPT with set of states $\mathscr{S}$, set of effects $\mathscr{E}$ and set of allowed transformations $\mathscr{T}$. The GPT's dimension can be finite dimensional or continuous dimensional. We are explicitly using the notation for continuous dimensional GPTs introduced in section \ref{continuous}, as this will be useful in following sections. 

Let $F$ be a reference frame in spacetime. We consider a passive transformation $G^{-1}\in\mathfrak{G}$, which is the inverse of a transformation $G\in\mathfrak{G}$, applied on $F$ (see Fig. \ref{fig2}). We say that a state $Z\in\mathscr{S}$ and an effect $\hat{E}\in\mathscr{E}$ transform into $Z'$ and $\hat{E}'$ under $G$ if in the new reference frame $F'$, the state and effect are given by $Z'$ and $\hat{E}'$, respectively. We say that the GPT is \emph{invariant under $\mathfrak{G}$} if for any state $Z\in\mathscr{S}$ and for any effect $\hat{E}\in\mathscr{E}$ in the reference frame $F$, and for any $G\in\mathfrak{G}$, it holds that $Z$ and $\hat{E}$ transform as
\begin{eqnarray}
\label{x0.1}
Z\xrightarrow{G} Z'&\equiv& \hat{\bold{R}}^{\text{st}}(G)[Z],\nonumber\\
\hat{E}\xrightarrow{G} \hat{E}'&\equiv& \hat{\bold{R}}^{\text{ef}}(G)[\hat{E}],
\end{eqnarray}
under $G$, where $Z'\in\mathscr{S}$ and $\hat{E}'\in\mathscr{E}$, and where $\hat{\bold{R}}^{\text{st}}$ and $\hat{\bold{R}}^{\text{ef}}$ are representations of $\mathfrak{G}$; and the outcome probabilities remain invariant:
\begin{equation}
\label{x0.3}
\hat{E}'[Z']=\hat{E}[Z].
\end{equation}
%We could also impose that $\hat{\bold{R}}^{\text{st}}(G)\in\mathscr{T}$, for all $G\in\mathfrak{G}$, i.e. that the spacetime symmetry transformations $G\in\mathfrak{G}$ must induce allowed transformations in the state space $\mathscr{S}$.  However, we do not need to impose this extra condition for the results obtained in this paper.

\begin{figure}
\includegraphics[scale=0.32]{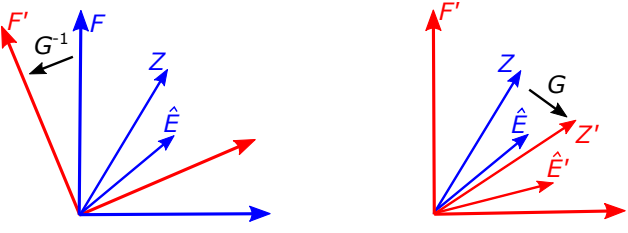}
 \caption{\label{fig2} \textbf{Transformations on reference frames.} A passive transformation $G^{-1}\in\mathfrak{G}$, which is the inverse of a transformation $G\in\mathfrak{G}$, is applied to the reference frame $F$. We illustrate a transformation $G\in\mathfrak{G}$ comprising a rotation in a two-dimensional space. Left: On the frame $F$, we illustrate by blue arrows a state $Z$ and an effect $\hat{E}$ of a given GPT; the frame $F$ (blue bold perpendicular axes) is transformed into a new frame $F'$ (red bold perpendicular axes). Right: in the frame $F'$, the state $Z$ and effect $E$ (blue arrows) are transformed into a new state $Z'$ and a new effect $\hat{E}'$ (red arrows), respectively. The GPT is invariant under $\mathfrak{G}$, if $Z'=\hat{\bold{R}}^{\text{st}}(G)[Z]$ and $\hat{E}'=\hat{\bold{R}}^{\text{ef}}(G)[\hat{E}]$, where $\hat{\bold{R}}^{\text{st}}$ and $\hat{\bold{R}}^{\text{ef}}$ are representations of $\mathfrak{G}$, and if it holds that the outcome probabilities remain invariant, that is, $\hat{E}'(Z')=\hat{E}(Z)$. We illustrate the simple case $\hat{\bold{R}}^{\text{st}}(G)=\hat{\bold{R}}^{\text{ef}}(G)=G$.} 
\end{figure}

For $a\in\{\text{st},\text{ef}\}$, $\hat{\bold{R}}^{a}$ is a representation of $\mathfrak{G}$ if it holds that
\begin{eqnarray}
\label{x0.2}
\hat{\bold{R}}^{a}(I)&=&\hat{\bold{I}}^{a},\nonumber\\
%\hat{\bold{R}}^{\text{ef}}(I)&=&\hat{\bold{I}}^{\text{ef}},\nonumber\\
\hat{\bold{R}}^{a}(G_2)\circ\hat{\bold{R}}^{a}(G_1)&=&\hat{\bold{R}}^{a}(G_2\circ G_1);
%\hat{\bold{R}}^{\text{ef}}(G_2)\circ\hat{\bold{R}}^{\text{ef}}(G_1)&=&\hat{\bold{R}}^{\text{ef}}(G_2\circ G_1);
\end{eqnarray}
for all $G_2,G_1\in\mathfrak{G}$, where $I$ is the identity element of $\mathfrak{G}$, $\hat{\bold{I}}^{\text{st}}$ is the identity acting on $\mathscr{S}$, and $\hat{\bold{I}}^{\text{ef}}$ is the identity acting on $\mathscr{E}$. $\hat{\bold{R}}^{a}$ is a trivial representation of $\mathfrak{G}$ if $\hat{\bold{R}}^{a}(G)=\hat{\bold{I}}^{a}$, for all $G\in\mathfrak{G}$ and for all $a\in\{\text{st},\text{ef}\}$.

We say that a GPT is \emph{nontrivially invariant under $\mathfrak{G}$} if the representations $\hat{\bold{R}}^{\text{st}}$ and $\hat{\bold{R}}^{\text{ef}}$ of $\mathfrak{G}$ above are not the trivial representations. We note from the invariance of probabilities that if a GPT is invariant under $\mathfrak{G}$ and one of the representations $\hat{\bold{R}}^{\text{st}}$ and $\hat{\bold{R}}^{\text{ef}}$ is trivial (nontrivial) then the other representation must also be trivial (nontrivial).

%In this paper we consider the special case in which $\mathfrak{G}$ is continuously connected to the identity. For example, we consider that spacetime is Minkowski in $n+1$ dimensions and that $\mathfrak{G}=\mathfrak{Poin}$, where $\mathfrak{Poin}$ is the group of Poincar{\'e} transformations in $n+1$ dimensions that are continuously connected to the identity, which is called the proper orthochronous Poincar{\'e} group, for all $n\in\mathbb{N}$. 
%We also consider the case that spacetime is Galilean in $n+1$ dimensions. In this case, we take $\mathfrak{G}=\mathfrak{Gal}$, where $\mathfrak{Gal}$ is the Galilei group in $n+1$ dimensions. Unless otherwise stated, in the following we consider that spacetime is Minkowski in $3+1$ dimensions and that $\mathfrak{G}=\mathfrak{Poin}$.

In order to provide a physical intuition of our definition of GPTs invariant under the group $\mathfrak{G}$ of spacetime symmetries let us consider now that spacetime is Minkowski in $1+3$ dimensions. Consider an idealized thought experiment in which in the universe there are only an observer with a set of coordinates defining a reference frame, a physical system $B$ represented by a black box and a finite set of detectors $D_1,D_2,\ldots,D_N$ that completely cover the area of a sphere surrounding the black box, and where the area of the detectors do not intersect.
%(see Fig. \ref{fig0}).
Suppose that the black box is at the origin of a reference frame $F$ and the sphere of detectors has its centre at the origin too.  In our idealized experiment we neglect the mass of the observer, black box and detectors, and assume that spacetime is exactly Minkowski.

The following description of the experiment takes place in the frame $F$.  When a button is pressed, the box emits a particle. With probability $P(i)$ the particle is detected only by the detector $D_i$, by emitting a signal, for instance, for all $i\in[N]=\{1,2,\ldots,N\}$. Since the detectors cover the whole area of a sphere surrounding the black box and their areas do not intersect, we have $\sum_{i=1}^NP(i)=1$. We assume that there is a physical state $Z$ described by a GPT for the black box that determines the probability distribution $\{P(i)\}_{i\in[N]}$. Thus, the detector $D_i$ is associated to an effect $\hat{E}_i$, for all $i\in[N]$. Therefore, the GPT description of the experiment tells us that 
\begin{equation}
\label{example}
P(i)=\hat{E}_i[Z],
\end{equation}
for all $i\in[N]$.
Consider that the observer can prepare the box in the same physical state $Z$ as many times as wished. The probability distribution (\ref{example}) will be observed each time in the reference frame $F$. By repeating the experiment a very large number of times, the observer can estimate the probabilities $P(i)$ with arbitrarily great precision, from the obtained frequencies, for all $i\in[N]$.

Now suppose that the observer changes his reference frame to $F'$ by applying a rotation $R^{-1}$, which is the inverse of a rotation $R$, to his set of coordinates.
Clearly, because spacetime is Minkowski and spatial rotations are symmetries of Minkowski spacetime, this situation is physically equivalent to rotating the black box and sphere of detectors by $R$. More precisely, in the frame $F'$, the black box and the sphere of detectors are rotated by $R$. Thus, the probability $P'(i)$ that the particle is detected by $D_i$ in the frame $F'$ satisfies 
\begin{equation}
\label{example2}
P'(i)=P(i),
\end{equation}
for all $i\in[N]$. In the GPT description, this means that the state $Z$ and the effects $\hat{E}_i$ are transformed in the frame $F'$ to a state $Z'$ and to effects $\hat{E}'_i$, for all $i\in[N]$. The probability distribution $P'$ in the frame $F'$ is given by
\begin{equation}
\label{example3}
P'(i)=\hat{E}'_i[Z'],
\end{equation}
for all $i\in[N]$. Thus, from (\ref{example}) -- (\ref{example3}), we have
\begin{equation}
\label{example4}
\hat{E}'_i[Z']=\hat{E}_i[Z],
\end{equation}
for all $i\in[N]$. This means that (\ref{x0.3}) holds in Minkowski spacetime with the coordinates transformation $G=R$ and the group of transformations $\mathfrak{G}=\text{SO}(3)$. Furthermore, since the spatial rotations form a group, the group $\text{SO}(3)$, then (\ref{x0.1}) must hold too with $\mathfrak{G}=\text{SO}(3)$. That is, under spatial rotations, the state $Z$ and effects $\hat{E}_i$ must transform as representations of the group of spatial rotations $\text{SO}(3)$.

More generally, for an arbitrary physical system with GPT state $Z$ and for an arbitrary experiment with GPT effects $\hat{E}$ in Minkowski spacetime, (\ref{x0.1}) and (\ref{x0.3}) must hold with $\mathfrak{G}$ being the set of symmetries of Minkowski spacetime. Although the coordinate transformations in Minkowski spacetime form the Poincar{\'e} group $\mathfrak{Poin}_{\text{full}}$, physics is observed to be perfectly invariant under the proper orthochronous Poincar{\'e} group $\mathfrak{Poin}$, but not perfectly invariant under the full Poincar{\'e} group $\mathfrak{Poin}_{\text{full}}$, which includes the discontinuous transformations of space inversion $P_{\text{inv}}$, time reversal $T_{\text{rev}}$ and $P_{\text{inv}}T_{\text{rev}}$ \cite{Weinbergbook}. Thus, in Postulate \hyperref[P1first]{1} below we consider that (\ref{x0.1}) and (\ref{x0.3}) hold in Minkowski spacetime with $\mathfrak{G}=\mathfrak{Poin}$.

\reff{In the rest of this paper we use the terms ``particle'', ``(four) momentum'' and ``mass'', as they are understood in classical physics in Minkowski spacetime, i.e. in special relativity. %In this paper, ``momentum'' is understood as the four momentum in Minkowski spacetime, which extends to $1+n$ momentum in Minkowski spacetime in $1+n$ dimensions.
This is justified by the following two facts. First, unless otherwise stated, in what follows we will assume that spacetime is Minkowski in $1+n$ dimensions, and the meaning of these terms can be straightforwardly extended to the case $n\neq 3$. Second, as explicitly stated in Postulate \hyperref[P2]{2} below, we will  also assume that the considered theory allows a class of states in which particles have classical momentum, and our derivation only works with this class of states.  Thus, in our derivation we will only consider particles having a well defined classical momentum in $1+n$ Minkowski spacetime. In this sense, our treatment for the spacetime degrees of
 freedom is classical, while the particles' internal degrees of freedom are treated more generally within the framework of finite dimensional GPTs. More precisely, in what follows, ``particle'' refers to the physical system under consideration, ``momentum'' refers to the $1+n$ momentum $p$ in $1+n$ Minkowski spacetime, which extends the concept of four momentum to the case $n\neq 3$, and the ``mass'' $m$ appears in the quantity $m^2=\sum_{\mu=0}^n\sum_{\nu=0}^n\eta_{\mu\nu}p_{\mu}p_{\nu}$, which remains invariant under Lorentz transformations. We note that in relativistic quantum mechanics the mass arises as an invariant number from the Casimir operator of the Poincar{\'e} group \cite{Waldbook}.}

In some of the following postulates we will refer to a particular type of physical system, defined as follows.

\begin{definition}[\textbf{Particles of the type $\mathcal{P}$}]
\label{particle}
Particles of the type $\mathcal{P}$ are massive particles with mass $m>0$. Each particle is described by the same GPT. The state space, space of effects and set of allowed transformations for a particle are denoted by $\mathscr{S}$, $\mathscr{E}$ and $\mathscr{T}$, respectively. These describe the spacetime degrees of freedom, like the momentum, which can be continuous dimensional, as well as the internal degrees of freedom, which are finite dimensional. The state space, space of effects and set of allowed transformations for a particle's internal degrees of freedom are denoted by $\mathcal{S}$, $\mathcal{E}$ and $\mathcal{T}$, respectively. These are described by a finite dimensional GPT with dimension $d$.
\end{definition}

%As stated by Theorem \ref{theorem1} below, from these postulates we derive the qubit Bloch ball, that any finite dimensional system can be described by finite dimensional quantum theory, and that the number of spatial dimensions must be $n=3$.

%obtain our our main result, Theorem \ref{theorem1}, which states that spacetime is Minkowski in $1+3$ dimensions and that any finite dimensional system can be described by finite dimensional quantum theory.  

%We present below the postulates from which finite dimensional quantum theory is reconstructed and the number of spatial dimensions $n=3$ is derived. We assume that spacetime is Minkowski in $3+1$ dimensions.

\begin{postulate}[\textbf{Poincar{\'e} Structure}]
\label{P1first}
If spacetime is Minkowski in $1+n$ dimensions then there exists an arbitrarily large number of particles of the type \hyperref[particle]{$\mathcal{P}$}. The GPT describing each particle satisfies the following two conditions.
\begin{enumerate}
\item\textbf{Nontrivial Poincar{\'e} Invariance}. The particle's GPT is nontrivially invariant under $\mathfrak{Poin}$.

\item \textbf{Nontrivial Structure}. If consistency with Nontrivial Poincar{\'e} Invariance requires a class of states from the state space $\mathcal{S}$ of any of the particle's internal degrees of freedom to transform as a representation of a subgroup of $\mathfrak{Poin}$ then such a representation must be nontrivial.
\end{enumerate}
\end{postulate}

We believe that this postulate is the main physical contribution of this paper to the literature of general probabilistic theories and reconstructions of quantum theory. To the best of our knowledge postulates similar to this one have not been considered before in derivations of quantum theory \reff{within the framework of GPTs}. This postulate suggests a connection between the mathematical structures of Minkowski spacetime and of finite dimensional quantum theory. As Lemma \ref{newlemma} given in section \ref{technical} shows, this postulate provides a first crucial step in establishing a relationship between the symmetries of spacetime and the state space of a massive particle's internal degrees of freedom of the type $\mathcal{P}$.

We state Postulate \hyperref[P1first]{1} for a particular type of physical system, a particle of a type that we have called $\mathcal{P}$, and which has mass $m>0$. This is motivated by the fact that in physics we have different types of physical systems, and in particular different types of elementary particles. A priori, different types of physical systems can behave in different ways. We only need to assume that this postulate holds for the particular type of physical systems considered. 

Regarding Postulate \hyperref[P1first]{1.1}, in general, we could assume that for any type of physical system, the GPT describing it is invariant under the proper orthochronous Poincar{\'e} group $\mathfrak{Poin}$. But, for some physical systems, the GPT states and effects could transform as trivial representations of $\mathfrak{Poin}$. That is, we cannot assume that for all physical systems the GPT is nontrivially invariant under $\mathfrak{Poin}$. The assumption that the representations $\hat{\bold{R}}^{\text{st}}$ and $\hat{\bold{R}}^{\text{ef}}$ are not the trivial ones at least for one type of physical system, i.e. for the particles of the type $\mathcal{P}$, is motivated by the observation that in quantum theory the spin degrees of freedom of massive particles arise due to Poincar{\'e} invariance \cite{W39,Weinbergbook}. One of our goals here is to investigate the structure of the spin degrees of freedom that follows from Poincar{\'e} invariance, independently of the mathematical structure of quantum theory.

Similarly, we cannot assume that Postulate \hyperref[P1first]{1.2} applies to any physical system. However, the assumption that Postulate \hyperref[P1first]{1.2} applies to some physical systems, the particles of type $\mathcal{P}$, is a natural condition given our motivations here. %\reff{In particular, the condition that the particle's internal degrees of freedom transform nontrivially is necessary to derive the structure of its state space $\mathcal{S}$.} 
As mentioned above, one of the goals in this work is to investigate the structure of the state space $\mathcal{S}$ of finite degrees of freedom
following from Poincar{\'e} invariance. We notice that Postulate \hyperref[P1first]{1.2} holds in the quantum case, where the states of massive particles' spin degrees of freedom corresponding to Hilbert spaces of finite dimension greater than one transform as irreducible (hence nontrivial) unitary representations of a subgroup of $\mathfrak{Poin}$ \cite{W39,Weinbergbook}.

Although in nature there exist particles with zero mass, e.g the photons, it is mathematically useful for our analysis to consider that the particles of the type $\mathcal{P}$ are massive. We restrict their mass to be positive because particles with negative mass are not observed in nature \cite{Weinbergbook}.

Although this postulate restricts to a particular type of physical system, we make a connection with arbitrary physical systems with any finite number of degrees of freedom by using Postulates \hyperref[P3]{3} and \hyperref[P7]{7}, introduced in  section \ref{particlepostulates}. These postulates roughly state that each particle of the type $\mathcal{P}$ encodes in its internal degrees of freedom an elementary system in the theory, and that an arbitrary physical system described by an arbitrary and finite number of degrees of freedom can be described by a sufficiently large number of elementary systems, respectively. For this reason, we require to assume in Postulate \hyperref[P1first]{1} that the number of particles of the type $\mathcal{P}$ can be arbitrarily large. This is motivated by the fact that in quantum theory, any quantum system of finite Hilbert space dimension can be described by a sufficiently large number of qubits. This is true of any type of physical systems used to encode the qubits, for example, polarization degrees of freedom of photons, spin degrees of freedom of elementary particles, energy degrees of freedom of atoms, etc.

The particles of the type $\mathcal{P}$ can be considered in quantum theory to be electrons, for example. The internal degrees of freedom of electrons are called spin and are represented by a Hilbert space of dimension 2. That is, a qubit can be encoded in the spin degrees of freedom of an electron. Since electrons and other elementary particles are more appropriately described as excitations of quantum fields, we can assume that there can be an infinite number of them. Thus, our assumption that we can have an arbitrarily large number of particles of the type $\mathcal{P}$ is justified by our current understanding of physics. We emphasize that this is just an example, that we are not assuming the particles $\mathcal{P}$ to be electrons or any other particular type of known elementary particle. In fact, we are not assuming quantum theory to hold, as our aim here is to reconstruct finite dimensional quantum theory within a broader class of probabilistic theories.

%We comment on Postulate \hyperref[P1first]{1.1}. First, the condition that Poincar{\'e} invariance is restricted only to $\mathfrak{Poin}$ is motivated by the fact that physics is observed to be perfectly invariant under the proper orthochronous Poincar{\'e} group $\mathfrak{Poin}$, but not perfectly invariant under the full Poincar{\'e} group $\mathfrak{Poin}_{\text{full}}$, which includes the discontinuous transformations of space inversion $P_{\text{inv}}$, time reversal $T_{\text{rev}}$ and $P_{\text{inv}}T_{\text{rev}}$ \cite{Weinbergbook}.

As mentioned in the introduction, an important motivation to reconstruct quantum theory from physical postulates is that by modifying the postulates we can investigate modifications of quantum theory. This is particularly relevant in the investigation of quantum gravity theories, in which potentially general relativity and/or quantum theory have to be modified. Thus, it is physically motivated to investigate variations of this postulate in which spacetime is not Minkowski. Poincar{\'e} invariance could then be replaced by invariance under the set of spacetime symmetries. For example, a modification of this postulate could be investigated for the spacetimes that are allowed by general relativity. Variations of this postulate could also be investigated for different spacetimes predicted by proposed modifications of general relativity, with the goal of investigating candidate theories for quantum gravity, for instance. With these motivations in mind, we can state a modification of Postulate \hyperref[P1first]{1} that holds in arbitrary hypothetical spacetimes.

%MAYBE FIND REFERENCES IN PARAGRAPH ABOVE 

\begin{postulate1'*}[\textbf{Structure from the Spacetime Symmetries}]
\label{P1'}
In a spacetime with group of symmetry transformations $\mathfrak{G}$ there exists a type of physical system described by a GPT satisfying the following two conditions.
\begin{enumerate}
\item\textbf{Nontrivial Invariance}. The GPT is nontrivially invariant under $\mathfrak{G}$.

\item \textbf{Nontrivial Structure}. If consistency with Nontrivial Invariance requires a class of states from the state space $\mathcal{S}$ of any of the internal degrees of freedom of the system to transform as a representation of a subgroup of $\mathfrak{G}$ then such a representation must be nontrivial.
\end{enumerate}
\end{postulate1'*}

In order to provide an intuition of how this postulate could be applied to different spacetimes, we present a toy example illustrated in Fig. \ref{figtoyspacetime}. \reff{This example does not fit within the model of spacetime given in section \ref{model}, but provides a simple illustration.} Consider a hypothetical spacetime with one time dimension $t$ and one spatial dimension $x$. Let $t$ and $x$ have discrete values given by
\begin{equation}
\label{discretespacetime}
t=ia \quad \text{ and } x=jb,
\end{equation}
\reff{in some reference frame,} where $a>0$ and $b>0$, for all integers $i$ and $j$. Suppose that the group of symmetries $\mathfrak{G}$ for this spacetime comprises only the space translations $T_k$ acting on a spacetime point like
\begin{equation}
\label{discretetranslations}
T_k(t,x)=(t,x+kb),
\end{equation}
where $k$ is an arbitrary integer. Suppose that there exists a physical system in this spacetime described by a polygon theory, discussed in section \ref{secpolygon}, for an arbitrary integer $N\geq 3$. Furthermore, suppose that under a translation $T_k$, the states $\zeta$ and effects $\varepsilon$ transform by applying a rotation $R(k\theta)$ of an angle $k\theta$ in the last two (of the three) vector entries of $\zeta$ and $\varepsilon$, i.e. on the two dimensions illustrated in Fig. \ref{figpol}, with $\theta=\frac{2\pi}{N}$ and for every integer $k$. As discussed in section \ref{secpolygon}, the rotations $R(k\theta)$ leave the state space $\mathcal{S}$ and the space of effects $\mathcal{E}$ invariant. For example, the set of pure states $\zeta_i$, given by (\ref{polygonstates}), and the set of extremal effects $\varepsilon_i$ and $\bar{\varepsilon}_i$, given by (\ref{polygoneffectseven}) and (\ref{polygoneffectsodd}), remain invariant, as given by (\ref{polygonrot}). It also holds that $\varepsilon'(\zeta')=\varepsilon(\zeta)$, for all $\zeta\in\mathcal{S}$ and all $\varepsilon\in\mathcal{E}$, where $\zeta'$ and $\varepsilon'$ are the transformed states and effects, respectively. It is straightforward to see that the rotations $R(k\theta)$ forms a group, and that this group is a nontrivial representation of the group $\mathfrak{G}$ of space translations described above. Thus, the system's GPT is nontrivially invariant under the group of spacetime symmetries $\mathfrak{G}$, and Postulate \hyperref[P1']{1'} holds.

\begin{figure}
\includegraphics[scale=0.26]{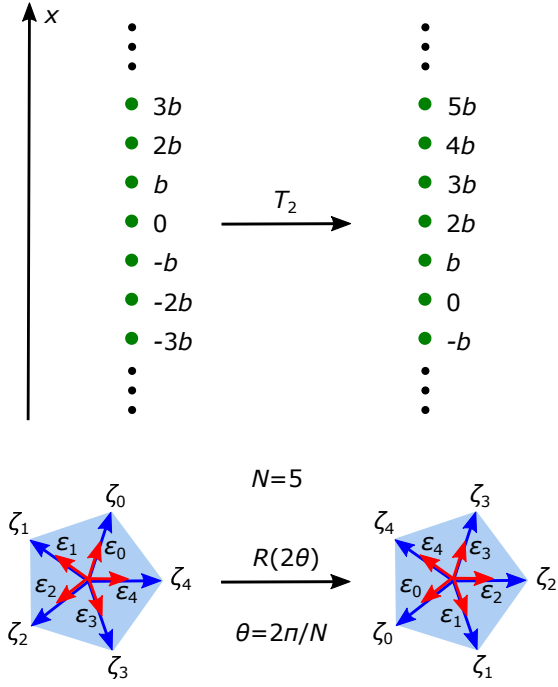}
 \caption{\label{figtoyspacetime} \textbf{Toy spacetime and GPT satisfying Postulate \hyperref[P1']{1'}.} Top: a spacetime in $1+1$ dimensions with discrete spacetime values given by (\ref{discretespacetime}) has as group of symmetries $\mathfrak{G}$, which is the group of space translations $T_k$ acting on spacetime points like in (\ref{discretetranslations}), for all integers $k$. Bottom: A physical system in this spacetime is described by a polygon theory presented in section \ref{secpolygon} (see Fig. \ref{figpol}), for an arbitrary integer $N\geq 3$. A space translation $T_k$ induces a rotation $R(k\theta)$ in the last two (of the three) dimensions of the states $\zeta\in\mathcal{S}$ and the effects $\varepsilon\in\mathcal{E}$, where $\theta=\frac{2\pi}{N}$. The rotations $R(k\theta)$ leave $\mathcal{S}$ and $\mathcal{E}$ invariant and satisfy $\varepsilon'(\zeta')=\varepsilon(\zeta)$, for arbitrary $\zeta\in\mathcal{S}$ and $\varepsilon\in\mathcal{E}$, where $\zeta'$ and $\varepsilon'$ are the transformed states and effects, respectively. The group of rotations $R(k\theta)$, with $k$ an arbitrary integer, is a nontrivial representation of $\mathfrak{G}$. Thus, Postulate \hyperref[P1']{1'} is satisfied. The case of $N=5$ and $k=2$ is illustrated.}
\end{figure}

The previous is only a simple example to illustrate how Postulate \hyperref[P1']{1'} could hold in a spacetime that is not Minkowski and with a GPT that is not quantum. We leave the investigation of this postulate for more general spacetimes and GPTs as an open problem.

\section{A set of physical postulates for finite dimensional quantum theory}
\label{particlepostulates}

\reff{In addition to Postulate \ref{P1first} introduced in section \ref{secmainpostulate}, we} use the following postulates in sections \ref{qubitsection} and \ref{quantumtheorysection} to reconstruct the qubit Bloch ball and finite dimensional quantum theory, and to derive that the number of spatial dimensions in Minkowski spacetime must be $n=3$. 

\begin{comment}
%COMMENT BEGINS
For completeness of this section we present again Postulate \hyperref[P1first]{1} introduced in section \ref{secmainpostulate}.

\begin{postulate1*}[\textbf{Poincar{\'e} Structure}]
\label{P1}
If spacetime is Minkowski in $1+n$ dimensions then there exists an arbitrarily large number of particles of the type \hyperref[particle]{$\mathcal{P}$}. The GPT describing each particle satisfies the following two conditions.
\begin{enumerate}
\item\textbf{Nontrivial Poincar{\'e} Invariance}. The particle's GPT is nontrivially invariant under $\mathfrak{Poin}$.

\item \textbf{Nontrivial Structure}. If consistency with Nontrivial Poincar{\'e} Invariance requires a class of states from the state space $\mathcal{S}$ of any of the particle's internal degrees of freedom to transform as a representation of a subgroup of $\mathfrak{Poin}$ then such a representation must be nontrivial.
\end{enumerate}
\end{postulate1*}

As discussed in section \ref{secmainpostulate}, this postulate (and its generalization given by Postulate \hyperref[P1']{1'}) is our main physical contribution to the literature of GPTs and reconstructions of quantum theory. We recall that $\mathfrak{Poin}$ is the proper orthochronous Poincar{\'e} group, comprising the Poincar{\'e} transformations that are continuously connected to the identity, which are the spacetime translations, the spatial rotations and the Lorentz boosts. For $n=3$, this is the subgroup of the Poincar{\'e} group under which physics is observed to be perfectly invariant \cite{Weinbergbook}.
%END COMMENT
\end{comment}

As Postulate \hyperref[P1first]{1}, Postulates \hyperref[P2]{2} and \hyperref[P3]{3} assume that there exist particles of the type $\mathcal{P}$, given by Definition \ref{particle}. Furthermore, Postulate \hyperref[P2]{2} assumes that spacetime is Minkowski in $1+n$ dimensions, with $n$ being an arbitrary positive integer. In order to avoid repetition, we do not state this explicitly in the postulates.

\begin{postulate}[\textbf{Existence of a Classical Limit}]
\label{P2}
For a particle of the type \hyperref[particle]{$\mathcal{P}$}, there exists a classical limit for the set of states. This postulate is two fold:
\begin{enumerate}
  \item \textbf{Existence of Classical Momentum}. There exists a class of states  $Z_{p,\zeta}^\text{class}\in\mathscr{S}$ describing a particle of the type $\mathcal{P}$ with classical well defined $1+n$ momentum $p$, and with internal degrees of freedom in an arbitrary state $\zeta\in\mathcal{S}$, for arbitrary physically realizable momentum $p$ for a particle of the type $\mathcal{P}$. Under an arbitrary proper orthochronous Poincar{\'e} transformation $P(x,\Lambda)\in\mathfrak{Poin}$, the states $Z_{p,\zeta}^\text{class}\in\mathscr{S}$ transform into the states $Z_{\Lambda p,\zeta'}^\text{class}\in\mathscr{S}$, with $\zeta'\in\mathcal{S}$.
  \item   \textbf{Existence of the Classical Bit}. The state space $\mathcal{S}$ for the internal degrees of freedom of a particle of the type $\mathcal{P}$ contains two or more states that can be perfectly distinguished in a measurement.
\end{enumerate}
\end{postulate}

We believe this postulate is very natural. To the best of our knowledge, physics allows a classical limit in which classical momentum and the classical bit exist, given by particular quantum states that we call ``classical''. 

To our knowledge, Postulate \hyperref[P2]{2.1} has not been considered before. We think this postulate is an important physical contribution to the literature of GPTs and reconstructions of quantum theory because it makes a first step in considering a particle's spacetime degrees of freedom and continuous dimensional degrees of freedom.
% As Lemma \ref{newlemma} given in section \ref{elementary} shows, this postulate together with Postulate \hyperref[P1first]{1} allows us to obtain a connection between the symmetries of Minkowski spacetime and the state space of a massive particle's internal degrees of freedom.

In this paper we only consider states with classical momentum. As shown by Lemma \ref{newlemma} in section \ref{technical}, this suffices to establish a connection between Minkowski spacetime and the state space $\mathcal{S}$ of the internal degrees of freedom of a particle of the type $\mathcal{P}$. 

\reff{We note that the non-trivial unitary representations of $\mathfrak{Poin}$ must be infinite dimensional \cite{Weinbergbook}. In quantum theory, a particle's quantum state can be transformed as a non-trivial unitary representation of $\mathfrak{Poin}$ by including in the transformations the spacetime degrees of freedom, namely the four momentum, which are continuous dimensional, and the spin  degrees of freedom, which are finite dimensional. Similarly, the states and effects considered here transform as representations of $\mathfrak{Poin}$, with the outcome probabilities remaining invariant, by including in the transformations the momentum degrees of freedom, which are continuous dimensional, and the internal degrees of freedom, which we have defined as finite-dimensional (see Lemma \ref{newlemma}). 

We further note that our derivation of finite dimensional quantum theory is via the internal degrees of freedom of particles of the type $\mathcal{P}$ (see Postulates  \hyperref[P3]{3} and  \hyperref[P7]{7} below). Regarding the momentum degrees of freedom, we only assume (in Postulate \hyperref[P2]{2.1}) that there are states with classical well defined momentum.} It would be interesting to investigate probabilistic theories where the spacetime degrees of freedom, like the momentum, are not only in classical states. We think that this requires a deeper analysis of continual dimensional general probabilistic theories and is thus out of the scope of this paper.

Postulate \hyperref[P2]{2.1} implies in particular that there exists a reference frame $F_{\text{rest}}$ in which a particle of the type $\mathcal{P}$ has momentum ${p_\text{rest}}=(m,0,0,\ldots,0)^{\text{t}}$ with $m>0$. In our notation, the first component is temporal and the other $n$ are spatial. We use units in which the speed of light is $c=1$. The set $\Pi_{p_\text{rest}}\equiv\{p=\Lambda {p_\text{rest}}\vert \Lambda\in \mathfrak{L}\}$ includes all physically possible $1+n$ momentums for a particle of the type $\mathcal{P}$, which has mass $m>0$. Mathematically, Postulate \hyperref[P2]{2.1} says that the considered theories include a set of states 
\begin{equation}
\label{eq:1}
\mathscr{S}^\text{class}\equiv \{Z_{p,\zeta}^{\text{class}}\vert p\in\Pi_{p_\text{rest}}, \zeta\in\mathcal{S}\}\subseteq\mathscr{S},
\end{equation}
and a set of effects
\begin{equation}
\label{eq:2}
\mathscr{E}^{\text{class}}\equiv \{\hat{E}_{p,\varepsilon}^{\text{class}}\vert p\in\Pi_{p_\text{rest}},\varepsilon\in\mathcal{E}\}\subseteq\mathscr{E},
\end{equation}
that perfectly distinguish the value of the momentum $p$ for states in $\mathscr{S}^\text{class}$, that is, such that
\begin{equation}
\label{eq:3}
\hat{E}_{p',\varepsilon}^{\text{class}}[Z_{p,\zeta}^{\text{class}}]=\delta(p'-p)\varepsilon\cdot \zeta,
\end{equation}
for all $p', p\in\Pi_{p_\text{rest}}$, all $\zeta\in\mathcal{S}$ and all $\varepsilon\in\mathcal{E}$. A measurement of the momentum for the states in $\mathscr{S}^\text{class}$ is given by the set of effects $\{\hat{E}_{p,u}^{\text{class}}\in\mathscr{E}^{\text{class}}\}_{p\in\Pi_{p_{\text{rest}}}}$, where $u\in\mathcal{E}$ is the unit effect on $\mathcal{S}$. Since the state of a particle with classical momentum $p$ transforms into a state with classical momentum $\Lambda p$ under a Poincar{\'e} transformation $P(x,\Lambda)\in \mathfrak{Poin}$, the states $Z_{p,\zeta}^{\text{class}}$ and effects  $\hat{E}_{p,\varepsilon}^{\text{class}}$ transform into
\begin{eqnarray}
\label{eq:5}
\hat{\bold{R}}^{\text{st}}(P(x,\Lambda))[Z_{p,\zeta}^{\text{class}}]&=&Z_{\Lambda p, R_p^{\text{st}}(P(x,\Lambda))\zeta}^{\text{class}},\nonumber\\
\label{eq:6}
\hat{\bold{R}}^{\text{ef}}(P(x,\Lambda))[\hat{E}_{p}^{\text{class}}(\varepsilon)]&=&\hat{E}_{\Lambda p,R_p^{\text{ef}}(P(x,\Lambda))\varepsilon}^{\text{class}},
\end{eqnarray}
where $\hat{\bold{R}}^{\text{st}}(P(x,\Lambda))$ and $\hat{\bold{R}}^{\text{ef}}(P(x,\Lambda))$ are representations of $\mathfrak{Poin}$ and $R^{\text{st}}_p(P(x,\Lambda)): \mathcal{S} \rightarrow \mathcal{S}$, $R^{\text{ef}}_p(P(x,\Lambda)):\mathcal{E} \rightarrow \mathcal{E}$ are transformations on the states and effects for the internal degrees of freedom of a particle of the type $\mathcal{P}$, respectively, which in general depend on $p$. 

In quantum theory, the states $Z_{p,\zeta}^{\text{class}}$ and the effects $\hat{E}_{\bar{p},\varepsilon}^{\text{class}}$ correspond to the projectors $\lvert a_{p}\rangle\langle a_{p} \rvert$ and $\lvert a_{\bar{p}}\rangle\langle a_{\bar{p}} \rvert$, where $\lvert a_p\rangle$ are the eigenstates of the four momentum operator with eigenvalues $p$, in which the states $\zeta\in\mathcal{S}$ and effects $\varepsilon\in\mathcal{E}$ are associated to a density matrix $\rho_{\zeta}$ and a POVM element $E_{\varepsilon}$ in a finite dimensional Hilbert space, respectively. A theory in which the states $Z_{p,\zeta}^{\text{class}}$, with $\zeta\in\mathcal{S}$ pure, are the only pure states in $\mathscr{S}$ is a classical theory in the momentum degrees of freedom. We are not restricting here to this particular case. We allow the possibility that there exist other pure states, although we do not investigate them in this paper.

Postulate \hyperref[P2]{2.2} is clearly satisfied by quantum theory because any quantum system with a nontrivial state space has a Hilbert space dimension greater than one and thus, containing a qubit in some subspace, can perfectly encode at least one classical bit. We note that this assumption has implications on the space of effects $\mathcal{E}$ associated to $\mathcal{S}$. That is, there exist at least two states $\zeta_0,\zeta_1\in\mathcal{S}$ and two effects $\varepsilon_0,\varepsilon_1\in\mathcal{E}$ satisfying $\varepsilon_j(\zeta_i)=\delta_{j,i}$, for all $i,j\in\{0,1\}$.

Postulate \hyperref[P2]{2.2} makes a weaker assumption than some postulates considered in previous reconstructions of quantum theory (e.g. \cite{CDP11,MM11,MMAP13}). For example, Ref. \cite{CDP11} assumes ``Perfect distinguishability'', which states that every state that is not completely mixed can be perfectly distinguished from some other state. Ref. \cite{MM11} assumes in ``Requirement 5'' that for a system with at most two distinguishable states the set of effects $\mathcal{E}$ equals the set of normalized effects $\mathcal{E}_\text{norm}$, i.e. all effects giving valid outcome probabilities for such a system exist in the theory. Similarly, Ref. \cite{MMAP13} assumes in the postulate ``Existence of an Information Unit'' that there exists a system called ``gbit'' for which $\mathcal{E}=\mathcal{E}_\text{norm}$. In this paper we do not make the strong assumption that $\mathcal{E}=\mathcal{E}_\text{norm}$ for any system. In fact, we derive this property for a particular system, the elementary system introduced below, from Postulate \hyperref[P2]{2.2} together with a few more of our postulates (see Lemmas \ref{adlemma1} and \ref{adlemma4} in section \ref{elementary}).

%Postulate \hyperref[P2]{2.2} is clearly satisfied by quantum theory because any quantum system with a nontrivial state space has a Hilbert space dimension greater than one and thus, containing a qubit in some subspace, can perfectly encode at least one classical bit. We note that this is a very weak property that is implicitly assumed by most of the axiomatic reconstructions of quantum theory. Postulate \hyperref[P2]{2.1}, being the analogue of Postulate \hyperref[P2]{2.2} for a continuous dimensional property, is also quite weak.

\begin{postulate}[\textbf{Minimality of the Elementary System}]
\label{P3}
There exists an elementary system in nature, defined as having a nontrivial state space $\mathcal{S}_{\text{elem}}$ with the smallest nontrivial finite dimension $d_{\text{elem}}\geq 1$ in nature. The GPT of the elementary system can be physically implemented in some of the internal degrees of freedom of a particle of the type \hyperref[particle]{$\mathcal{P}$}. The dimension $d_{\text{elem}}$ achieves the minimum value that is consistent with the set of considered postulates.
\end{postulate}

This is arguably a strong assumption. However, we think that among theories satisfying the same physical postulates, the theories that require the least number of real degrees of freedom are in some sense mathematically simpler and more efficient. We believe it is a natural assumption that physical theories should have mathematical structures that are as simple as possible, while still describing a broad range of physical phenomena, like the existence of entanglement as given by Postulate \hyperref[P6]{6} below, for instance. If two physical theories describe the same physical phenomena, where the first one needs $d_1$ real degrees of freedom to describe the elementary system and the second one needs $d_2>d_1$ real degrees of freedom for this, why would nature ``choose'' the second theory? Of course, the second theory could perhaps be more elegant than the first one by some standards, or it could involve mathematical calculations that are simpler than the first one. But, without knowing this a priori, we consider sensible that a measure of mathematical simplicity for a physical theory is the number of real degrees of freedom that it requires.

Nonetheless, we think that Postulate \hyperref[P3]{3} does not make a very strong assumption if we compare it with other postulates used in previous derivations of quantum theory. For example, Hardy's derivation \cite{H01} uses a  ``Simplicity Axiom", which roughly states that the number $d$ of real degrees of freedom to specify a state takes the minimum value consistent with the other considered axioms. Our postulate states that this only needs to hold for the system of smallest dimension $d_{\text{elem}}$, i.e. for the elementary system, while Hardy's axiom states that this holds for systems of arbitrary dimension $d$.

The assumption that the GPT of the elementary system can be physically implemented in some internal degrees of freedom of a particle of the type $\mathcal{P}$ seems also like a strong assumption because it gives a very special role to the particles of the type $\mathcal{P}$. However, this assumption has a more meaningful significance when considered together with Postulate \hyperref[P7]{7} below, which roughly states that any physical system with finite number of real degrees of freedom can be described by the GPT of a sufficiently large number of elementary systems.

In quantum theory, the elementary system is the qubit and any finite dimensional quantum state can be encoded in a sufficiently large number of qubits. As previously mentioned, in quantum theory, we can imagine particles of the type $\mathcal{P}$ to be electrons, for instance. In this case, since the electron has spin $\frac{1}{2}$, its spin degrees of freedom encode a qubit exactly. But a particle of the type $\mathcal{P}$ can be any other massive particle, like an atom.
Although the internal degrees of freedom of an atom have Hilbert space dimension larger than two, we can find a subspace of dimension two that describes a qubit.

Ideally, we would like to have derivations of quantum theory with the least number of postulates and with postulates that have the strongest physical motivations. Thus, we think it is interesting to investigate whether this postulate is really necessary as stated, or if
our reconstruction of the qubit Bloch ball and finite dimensional quantum theory can still be obtained with a relaxed version of this postulate in which  $d_{\text{elem}}$ is not assumed to achieve the minimum value that is consistent with the other postulates.

%\begin{postulate3'*}[\textbf{Elementary System}]
%\label{P3'}
%There exists an elementary system in nature, defined as having a nontrivial state space $\mathcal{S}_{\text{elem}}$ with the smallest finite dimension $d_{\text{elem}}>1$ in nature. The GPT of the elementary system can be physically implemented in some of the internal degrees of freedom of a particle of the type $\mathcal{P}$.
%\end{postulate3'*}

\begin{postulate}[\textbf{Continuous Reversibility}]
\label{P4}
For every pair of pure states there exists a continuous reversible transformation that transforms one state into the other.
\end{postulate}

This postulate was introduced by Hardy \cite{H01} and has been used in several reconstructions of quantum theory (e.g. \cite{DB09,MM11,TMSM12,MMAP13}). It is physically motivated by the continuity of time evolution, which holds in Minkowski spacetime, for instance. In quantum theory the set of reversible transformations is the set of unitary operations, which is continuous. We believe that this postulate should be relaxed when considering the possibility that spacetime could have a discrete structure, when investigating possible theories for quantum gravity, for instance. 

It is interesting to note that this postulate alone excludes the possibility that nature can be described by GPTs with discrete sets of pure states, like classical probabilistic theory of finite dimension or box world. This is because, as previously mentioned, reversible transformations take pure states into pure states. Thus, if a reversible transformation that takes a first pure state into a second pure state is continuous then there must be a continuous set of pure states connecting the first and second pure states.

\begin{postulate}[\textbf{Tomographic Locality}]
\label{P5}
The state of a composite system is totally characterized by the outcome probabilities of the local measurements on the subsystems.
\end{postulate}

Introduced by Barrett \cite{B07}, this postulate allows us to describe composite systems. It implies a simple relation between the number of real degrees of freedom $d_{AB}$ required to describe states of a composite system $AB$ and those of the subsystems, given by $1+d_{AB}=(1+d_A)(1+d_B)$, as initially considered by Hardy \cite{H01}. Tomographic locality has been considered in many reconstructions of quantum theory (e.g. \cite{H01,DB09,CDP11,MM11,H11,TMSM12,MMAP13}).

\begin{postulate}[\textbf{Existence of Entanglement}]
\label{P6}
The state space of any bipartite system contains at least one entangled state.
\end{postulate}

A state is \emph{entangled} if it cannot be written as a convex combination of product states. In our opinion, this postulate has great physical significance.
In quantum theory, many important properties of quantum information arise due to the existence of entangled states, like the violation of Bell inequalities \cite{Bell},  quantum teleportation \cite{teleportation}, superdense coding \cite{sdc}, and the existence of quantum computation algorithms that are exponentially faster than the best known classical algorithms \cite{DJ92,S94}, for instance. For these reasons, we believe it is physically sensible that nature allows the existence of entanglement. Postulate \ref{P6} has been considered before (e.g. \cite{TMSM12,MMAP13}).

\begin{postulate}[\textbf{Universal Encoding}]
\label{P7}
For any physical system, any state of finite dimension can be reversibly encoded in a sufficiently large number of elementary systems.
\end{postulate}

This postulate is similar to the postulate of ``Existence of an Information Unit" given in Ref. \cite{MMAP13}, which states that there exists a type of system called the ``gbit'' such that the state of any physical system can be reversibly encoded in a sufficiently large number of gbits. We note that the gbit plays the role of the qubit in quantum theory, like the elementary system in this paper does. We think that the postulate of Ref. \cite{MMAP13} should specify that this property holds for any physical system of finite dimension. This is because in quantum theory there are systems described by Hilbert spaces of continuous dimensions, which cannot be completely characterized by any finite number of qubits. The postulate of Ref. \cite{MMAP13} also assumes that gbits can interact, which is equivalent to Postulate \hyperref[P6]{6} (Existence of Entanglement) here. An important difference of Postulate \hyperref[P7]{7} with the postulate Existence of an Information Unit of Ref. \cite{MMAP13} is that we do not assume that all normalized effects are observable, which is arguably a strong assumption.

Postulate \ref{P7} says that at a fundamental level all physical systems with discrete degrees of freedom are described by the same mathematical theory.
This holds in quantum theory, where the elementary system is the qubit. An arbitrary finite dimensional quantum state is perfectly encoded in a finite number of qubits. This encoding is independent of the physical systems used to prepare the qubits, which can be, for example, the spin degrees of freedom of massive particles, the polarization degrees of freedom of massless particles, the energy levels of atoms or molecules, etc.

\section{A physical derivation of the elementary system and the qubit}
\label{qubitsection}

This section is divided in three subsections. Lemma \ref{newlemma}, presented in section \ref{technical}, is our main technical result, in the sense that the proofs of our most important results use this lemma. Lemma \ref{newlemma} roughly states that if spacetime is Minkowski in $1+n$ dimensions then Nontrivial Poincar{\'e} Invariance (Postulate \hyperref[P1first]{1.1}) and the Existence of Classical Momentum (Postulate \hyperref[P2]{2.1}) imply that there exists a class of states for a massive particle of the type $\mathcal{P}$ whose internal degrees of freedom transform as a representation of the group $\text{SO}(n)$. This result provides a first crucial step to establish a connection between Minkowski spacetime and the structure of the state space $\mathcal{S}$ of the internal degrees of freedom of a particle of the type $\mathcal{P}$. 

In section \ref{elementary} we assume that spacetime is Minkowski in $1+n$ dimensions and leave $n$ as a free variable. As is well known, and as discussed in section \ref{derivingMinkowski}, this follows from clear physical principles, like the principle of relativity and the invariance of the speed of light.
From Postulates \hyperref[P1first]{1} -- \hyperref[P3]{3}, and using Lemma \ref{newlemma}, Lemma \ref{adlemma1} below shows that the state space and the space of effects of the elementary system correspond to a $n-$dimensional Euclidean ball.

In section \ref{derivingqubit} we assume that spacetime is Minkowski in $1+3$ dimensions. From Postulates \hyperref[P1first]{1} -- \hyperref[P4]{4} and using Lemma \ref{adlemma1}, Lemma \ref{qubitlemma} below shows that the elementary system is locally equivalent to the qubit, i.e the state space, space of effects and set of reversible transformations of the elementary system correspond to those of the qubit.

\subsection{Main technical result: in Minkowski spacetime in $1+n$ dimensions some massive particles' internal degrees of freedom must transform as representations of $\text{SO(n)}$}
\label{technical}

%As previously mentioned, Lemma \ref{newlemma} below is our main technical result, in the sense that all our results are based on this lemma. It allows us to establish a first connection between the symmetries of Minkowski spacetime and the structure of the state space $\mathcal{S}$ of the internal degrees of freedom of a massive particle of the type $\mathcal{P}$.

\begin{lemma}
\label{newlemma}
We assume that spacetime is Minkowski in $1+n$ dimensions and that Postulates \hyperref[P1first]{1.1} and \hyperref[P2]{2.1} hold. 
%Thus, we assume that there exists a set of states $\mathscr{S}^\text{class}$ given by equation (\ref{eq:1}), whose elements $Z_p^{\text{class}}(\zeta)$ transform as given by equation (\ref{eq:5}) under proper orthochonous Poincar{\'e} transformations $P(x,\Lambda)\in \mathfrak{Poin}$, where $\hat{\bold{R}}\bigl(P(x,\Lambda)\bigr)$ is a representation of $\mathfrak{Poin}$.
We consider a state $Z_{p_{\text{rest}},\zeta}^{\text{class}}\in\mathscr{S}^\text{class}$ of a particle of the type $\mathcal{P}$, which has mass $m>0$, with well defined classical $1+n$ momentum $p_{\text{rest}}= (m,0,\ldots,0)^{\text{t}}$ in a given reference frame, with internal degrees of freedom in an arbitrary state $\zeta\in\mathcal{S}\subset \mathbb{R}^{d+1}$. We define $\Lambda_{p_{\text{rest}}}(p)\in\mathfrak{L}$ as the proper orthochronous Lorentz transformation that takes $p_{\text{rest}}$ into $p$, for all $p\in\Pi_{p_{\text{rest}}}\equiv\{\Lambda p_{\text{rest}}\vert \Lambda\in\mathfrak{L}\}$. From (\ref{eq:6}), under proper orthochronous Poincar{\'e} transformations $P\bigl(x,\Lambda_{p_{\text{rest}}}(p)\bigr)\in\mathfrak{Poin}$, the state $Z_{p_{\text{rest}},\zeta}^{\text{class}}$ transforms as
\begin{equation}
\label{adx1}
Z_{p_{\text{rest}},\zeta}^{\text{class}}\xrightarrow[]{P(x,\Lambda_{p_{\text{rest}}}(p))} Z_{p,\zeta'}^{\text{class}},
%\hat{\bold{R}}_\text{st}\Bigl(P\bigl(x,\Lambda_{p_{\text{rest}}}(p)\bigr)\Bigr)\bigl[Z_{p_{\text{rest}}}^{\text{class}}(\zeta) \bigr]=
\end{equation}
where the transformed state $\zeta'$ for the internal degrees of freedom is
\begin{equation}
\label{adx2}
\zeta' = R_{p_{\text{rest}}}^{\text{st}}\bigl(P(x,\Lambda_{p_{\text{rest}}}(p))\bigr)\zeta.
\end{equation}
Then, for any internal degrees of freedom of a particle of the type $\mathcal{P}$, it holds that the set of transformations $R_{p_{\text{rest}}}^{\text{st}}$ is a representation of the group SO$(n)$, that is,
\begin{equation}
\label{adx3}
R_{p_{\text{rest}}}^{\text{st}}(P(\vec{0},O_2))R_{p_{\text{rest}}}^{\text{st}}(P(\vec{0},O_1))=R_{p_{\text{rest}}}^{\text{st}}(P(\vec{0},O_2O_1)),
\end{equation}
for all $O_i\equiv\Bigl(\begin{smallmatrix}
  1 & 0 \\
  0 & \tilde{O_i} 
 \end{smallmatrix}\Bigr)$ with $\tilde{O_i}\in \text{SO}(n)$, and for all $i\in\{1,2\}$.
\end{lemma}

In our proof of Lemma \ref{newlemma}, the inclusion of the states $Z_{p,\zeta}^{\text{class}}\in\mathscr{S}^\text{class}$ with well defined classical $1+n$ momentum $p$, as follows from Postulate \hyperref[P2]{2.1}, allows us to use Wigner's method of induced representations \cite{W39,Weinbergbook}. Then, from Poincar{\'e} invariance (Postulate \hyperref[P1first]{1.1}), we show that under proper orthocrhonous Poincar{\'e} transformations in a reference frame in which the particle is stationary, with momentum $p_\text{rest} = (m,0,\ldots,0)^t$, the states $\zeta\in\mathcal{S}\subset\mathbb{R}^{d+1}$ for the particle's internal degrees of freedom must transform as representations of the little group for massive particles, the group that leaves $p_\text{rest}$ invariant, which is
$\text{SO}(n)$.

\begin{proof}[Proof of Lemma \ref{newlemma}]
%We define $\Lambda_{p_\text{rest}}(p)\in \mathfrak{L}$ as the proper orthochronous Lorentz transformation that takes the standard four momentum ${p_\text{rest}}\equiv(m,0,0,0)^t$, with $m>0$, to $p\in\Pi_{p_\text{rest}}\equiv\{ \Lambda {p_\text{rest}}\vert \Lambda\in \mathfrak{L}\}$. It follows that $\Lambda_{p_\text{rest}}({p_\text{rest}})$ is the identity on $\mathbb{R}^4$. 

A proper orthochronous Poincar{\'e} transformation $P(x,\Lambda)\in\mathfrak{Poin}$ denotes a translation by the spacetime vector $x\in\mathbb{R}^{1+n}$ and a proper orthochronous Lorentz transformation $\Lambda\in\mathfrak{L}$. From (\ref{neweq:l1}) and (\ref{eq:l1.1}), under a Poincar{\'e} transformation $P(x,\Lambda)$, a pair of spacetime and momentum $1+n$ vectors $(b,p)$ transform as
\begin{equation}
\label{ady9}
P(x,\Lambda)(b,p) =(x+\Lambda b, \Lambda p),
\end{equation}
and an arbitrary pair of Poincar{\'e} transformations compose as
\begin{equation}
\label{ady8}
P(x',\Lambda')\circ P(x,\Lambda) =P(x'+\Lambda' x,\Lambda' \Lambda),
\end{equation}
for all $x', x, b, p \in\mathbb{R}^{1+n}$ and all $\Lambda', \Lambda \in\mathfrak{L}$.

We consider a state $Z_{p,\zeta}^{\text{class}}\in\mathscr{S}^{\text{class}}$ with well defined classical $1+n$ momentum $p\in\Pi_{p_\text{rest}}$ and arbitrary state $\zeta\in\mathcal{S}\subset \mathbb{R}^{1+n}$ for any of the internal degrees of freedom of a particle of the type $\mathcal{P}$, in a reference frame $F$ in which the $1+n$ momentum is $p$. Under a Poincar{\'e} transformation $P^{-1}(x,\Lambda)\in\mathfrak{Poin}$, the reference frame $F$ is transformed into a reference frame $F'$. From Postulate \hyperref[P1first]{1.1}, in the frame $F'$, the state is
\begin{equation}
\label{ady0}
\hat{\bold{R}}^{\text{st}}(P(x,\Lambda))[Z_{p,\zeta}^{\text{class}}]=Z_{\Lambda p, R_p^{\text{st}}(P(x,\Lambda))\zeta}^{\text{class}},
\end{equation}
where $\hat{\bold{R}}^{\text{st}}(P(x,\Lambda))$ is a representation of $\mathfrak{Poin}$ and $R^{\text{st}}_p(P(x,\Lambda)): \mathcal{S} \rightarrow \mathcal{S}$ is a transformation on the states for the discrete degrees of freedom, for all $x\in\mathbb{R}^{1+n}$, all $p\in\Pi_{p_\text{rest}}$ and all $\Lambda, \Lambda'\in\mathfrak{L}$. In other words, the state $Z_{p,\zeta}^{\text{class}}$ is transformed into the state $\hat{\bold{R}}^{\text{st}}(P(x,\Lambda))[Z_{p,\zeta}^{\text{class}}]$ given by (\ref{ady0}) under the Poincar{\'e} transformation $P^{-1}(x,\Lambda)$ of the reference frame. In what follows we simply say that $Z_{p,\zeta}^{\text{class}}$ is transformed into $\hat{\bold{R}}^{\text{st}}(P(x,\Lambda))[Z_{p,\zeta}^{\text{class}}]$ under $P(x,\Lambda)$ (see Fig. \ref{fig2}). 

We define the Poincar{\'e} transformation
\begin{eqnarray}
\label{ady4}
P_{\vec{0},p_{\text{rest}}}^{\text{little}}(a,x,\Lambda,p)&\equiv& P(-\Lambda^{-1}_{p_{\text{rest}}}(\Lambda p)(x+a),\Lambda^{-1}_{p_{\text{rest}}}(\Lambda p))\nonumber\\
&&\!\!\!\!\!\!\quad \circ P(x+a-\Lambda x,\Lambda) \circ P(x,\Lambda_{p_{\text{rest}}}(p)),\nonumber\\
\end{eqnarray}
where $\vec{0}$ denotes the null spacetime vector $\vec{0}=(0,0,\ldots,0)^\text{t}$, for all $a,x\in\mathbb{R}^{1+n}$, $p\in\Pi_{p_{\text{rest}}}$ and $\Lambda\in\mathfrak{L}$. This transformation takes the pair $(\vec{0},p_{\text{rest}})$ to $(x,p)$, then to $(x+a,\Lambda p)$, and then back to $(\vec{0},p_{\text{rest}})$. Thus, it is an element of the little group of $(\vec{0},p_{\text{rest}})$, which is the subgroup of $\mathfrak{Poin}$ that leaves $(\vec{0},p_{\text{rest}})$ invariant. The little group of $(\vec{0},p_{\text{rest}})$ is SO$(n)$. This can be seen as follows. From (\ref{ady9}), an arbitrary Poincar{\'e} transformation $P(b,\Lambda')\in\mathfrak{Poin}$ transforms $(\vec{0},p_{\text{rest}})$ into $(b,\Lambda' p_{\text{rest}})$. Thus, in order that $P(b,\Lambda')$ belongs to the little group of $(\vec{0},p_{\text{rest}})$, it must hold that $b=\vec{0}$ and $\Lambda' p_{\text{rest}}= p_{\text{rest}}$, which requires that $\Lambda'\in\text{SO}(n)$, as 
$p_{\text{rest}}=(m,0,\ldots,0)^t$, with $m>0$.

\begin{comment}
%COMMENT BEGINS HERE
We show below that
\begin{eqnarray}
\label{y5}
&&P_{\vec{0},p_{\text{rest}}}^{\text{little}}(a',x+a,\Lambda',\Lambda p)  \circ P_{\vec{0},p_{\text{rest}}}^{\text{little}}(a,x,\Lambda,p)\nonumber\\
&&\qquad\qquad\qquad =P_{\vec{0},p_{\text{rest}}}^{\text{little}}(a+a',x,\Lambda'\Lambda, p),
\end{eqnarray}
for all $a', a, x\in\mathbb{R}^4$, $p\in\Pi_{p_\text{rest}}$ and $\Lambda, \Lambda'\in\mathfrak{L}$.
%COMMENT ENDS HERE
\end{comment}

We show below that
\begin{eqnarray}
\label{ady6}
&&R^{\text{st}}_{p_{\text{rest}}} \Bigl(P_{\vec{0},p_{\text{rest}}}^{\text{little}}(a',x+a,\Lambda',\Lambda p)\Bigr) R^{\text{st}}_{p_{\text{rest}}} \Bigl(P_{\vec{0},p_{\text{rest}}}^{\text{little}}(a,x,\Lambda,p)\Bigr) \nonumber\\
&&\quad = R^{\text{st}}_{p_{\text{rest}}} \Bigl(P_{\vec{0},p_{\text{rest}}}^{\text{little}}(a',x+a,\Lambda',\Lambda p)\circ P_{\vec{0},p_{\text{rest}}}^{\text{little}}(a,x,\Lambda,p) \Bigr),\nonumber\\
\end{eqnarray}
and that
\begin{equation}
\label{idrep}
R^{\text{st}}_{p_{\text{rest}}} \bigl(P\bigl(\vec{0},I\bigr)\bigr)=I^{\text{st}},
\end{equation}
where $I^{\text{st}}$ is the identity acting on $\mathcal{S}$, for all $a', a, x\in\mathbb{R}^{1+n}$, all $p\in\Pi_{p_\text{rest}}$ and all $\Lambda', \Lambda\in\mathfrak{L}$. We also show that $\Lambda_{p_{\text{rest}}}(p)$ can be consistently chosen in a way that it generates any $p\in \Pi_{p_\text{rest}}$ and that 
\begin{equation}
\label{ady00}
P_{\vec{0},p_{\text{rest}}}^{\text{little}}(a,x,O,p) = P(\vec{0},O),
\end{equation}
for any pure rotation $O$, i.e. for all $O=\Bigl(\begin{smallmatrix}
  1 & 0 \\
  0 & \tilde{O} 
 \end{smallmatrix}\Bigr)$ with $\tilde{O}\in \text{SO}(n)$, for all $a, x\in\mathbb{R}^{1+n}$ and for all $p\in\Pi_{p_\text{rest}}$. Thus, the transformations $P_{\vec{0},p_{\text{rest}}}^{\text{little}}(a,x,\Lambda,p)$ generate the whole little group $\text{SO}(n)$, for all $a, x\in\mathbb{R}^{1+n}$, $p\in\Pi_{p_\text{rest}}$ and $\Lambda \in \mathfrak{L}$. Therefore, it follows from (\ref{ady8}), (\ref{ady6}), (\ref{idrep}) and (\ref{ady00}), and from the fact that $\Lambda,\Lambda'\in \mathfrak{L}$ are arbitrary that $R^{\text{st}}_{p_\text{rest}}$ is a representation of the little group $\text{SO}(n)$, as claimed.

We show (\ref{ady6}) and (\ref{idrep}). From (\ref{ady0}), it is straightforward to see that in order that $\hat{\bold{R}}^{\text{st}}(P(x,\Lambda))$ be a representation of $\mathfrak{Poin}$, it must hold that
\begin{equation}
\label{ady1}
R^{\text{st}}_{\Lambda p}(P(b',\Lambda'))R^{\text{st}}_p(P(b,\Lambda))=R^{\text{st}}_p(P(b',\Lambda')\circ P(b,\Lambda)),
\end{equation}
and that
\begin{equation}
\label{idrep2}
R_p^{\text{st}}\bigl(P\bigl(\vec{0},I\bigr)\bigr)=I^{\text{st}},
\end{equation}
for all $b', b\in\mathbb{R}^{1+n}$, all $p\in\Pi_{p_\text{rest}}$ and all $\Lambda', \Lambda\in\mathfrak{L}$. From (\ref{ady4}), $P_{\vec{0},p_{\text{rest}}}^{\text{little}}(a,x,\Lambda,p)$ takes $p_{\text{rest}}$ to $p_{\text{rest}}$. Thus, (\ref{ady6}) and (\ref{idrep}) follow from (\ref{ady1}) and (\ref{idrep2}).

\begin{comment}
%COMMENT BEGINS
, we have
\begin{eqnarray}
\label{y20}
&&R^{\text{st}}_{p_{\text{rest}}} \Bigl(P_{\vec{0},p_{\text{rest}}}^{\text{little}}(a',x+a,\Lambda',\Lambda p)\Bigr) R^{\text{st}}_{p_{\text{rest}}} \Bigl(P_{\vec{0},p_{\text{rest}}}^{\text{little}}(a,x,\Lambda,p)\Bigr) \nonumber\\
&&\!\!\!\!\qquad = R^{\text{st}}_{p_{\text{rest}}} \Bigl(P_{\vec{0},p_{\text{rest}}}^{\text{little}}(a',x+a,\Lambda',\Lambda p) \circ P_{\vec{0},p_{\text{rest}}}^{\text{little}}(a,x,\Lambda,p) \Bigr),\nonumber\\
\end{eqnarray}
for all $a', a, x\in\mathbb{R}^4$, $p\in\Pi_{p_{\text{rest}}}$ and $\Lambda',  \Lambda\in\mathfrak{L}$. Thus, (\ref{y6}) follows straightforwardly from (\ref{y5}) and (\ref{y20}).
%COMMENT ENDS
\end{comment}

We show that $\Lambda_{p_{\text{rest}}}(p)$ can be consistently chosen in a way that it generates any $p\in \Pi_{p_\text{rest}}$. Let $\vec{p}$ be the $n$ momentum of $p$, i.e. the vector whose components correspond to the $n$ spatial dimensions of $p$. Let $\lVert \vec{p}\rVert >0$ be the Euclidean norm of $\vec{p}$, and let
\begin{equation}
\bar{p}\equiv\frac{\vec{p}}{\lVert \vec{p}\rVert}\nonumber
\end{equation}
be the unit vector in the direction of $\vec{p}$. We define
\begin{equation}
\label{adeq:a1}
\Lambda_{p_{\text{rest}}}(p)=Q(\bar{p})S(\lVert \vec{p}\rVert)Q^{-1}(\bar{p}),
\end{equation}
where $Q(\bar{p})$ is a pure rotation that takes the first-axis of the spatial dimensions into the axis $\bar{p}$ and
\begin{equation}
\label{eq:a1.1}
S(\lVert \vec{p}\rVert)=\left(\begin{smallmatrix}
\gamma & \sqrt{\gamma^2 -1} & 0 & 0 & \ldots & 0 &0 \\ 
\sqrt{\gamma^2 -1} & \gamma & 0 & 0& \ldots & 0 & 0\\
0 & 0 & 1 & 0& \ldots & 0&0\\
\vdots  & \vdots  & \vdots & \vdots    &\ddots &\vdots & \vdots\\
0 & 0 & 0 & 0&\ldots &0& 1\\
\end{smallmatrix} \right)\nonumber
\end{equation}
 is a pure boost in the first-axis that depends only on the magnitude of $\vec{p}$, and where
 \begin{equation}
 \gamma\equiv \frac{\sqrt{\lVert \vec{p}\rVert^2+m^2}}{m}.\nonumber
 \end{equation}
  We see that any momentum $p\in\Pi_{p_{\text{rest}}}$ can be obtained from ${p_\text{rest}}$ by such a transformation. Thus, we do not loss any generality by defining $\Lambda_{p_{\text{rest}}}(p)$ as in (\ref{adeq:a1}). 
 
We use the definition (\ref{adeq:a1}) of $\Lambda_{p_{\text{rest}}}(p)$, for all $p\in\Pi_{p_{\text{rest}}}$. We complete the proof by showing (\ref{ady00}).
From the definition (\ref{ady4}) of $P_{\vec{0},p_{\text{rest}}}^{\text{little}}(a,x,\Lambda,p)$, we have
\begin{equation}
\label{ady21}
P_{\vec{0},p_{\text{rest}}}^{\text{little}}(a,x,\Lambda,p)= P\bigl(\vec{0},\Lambda^{-1}_{p_{\text{rest}}}(\Lambda p)\Lambda \Lambda_{p_{\text{rest}}}(p)\bigr),
\end{equation}
for all $a, x\in\mathbb{R}^{1+n}$, all $p\in\Pi_{p_{\text{rest}}}$ and all $\Lambda', \Lambda \in\mathfrak{L}$. Thus, from (\ref{adeq:a1}) and (\ref{ady21}), if $O$ is a pure rotation, i.e if $O=\Bigl(\begin{smallmatrix}
  1 & 0 \\
  0 & \tilde{O} 
 \end{smallmatrix}\Bigr)$ with $\tilde{O}\in \text{SO}(n)$, we have
\begin{eqnarray}
\label{adeq:a2}
&&P_{\vec{0},p_{\text{rest}}}^{\text{little}}(a,x,O,p)\nonumber\\
&&\qquad= P\Bigl(\vec{0},Q(O\bar{p})S^{-1}(\lVert \vec{p}\rVert)Q'(O,\bar{p}) S(\lVert \vec{p}\rVert)Q^{-1}(\bar{p})\Bigr),\nonumber\\
\end{eqnarray}
where 
\begin{equation}
Q'(O,\bar{p})\equiv Q^{-1}(O\bar{p})OQ(\bar{p})\nonumber
\end{equation}
 is a pure rotation that takes the first-axis into $\bar{p}$, then into $O\bar{p}$, and then back into the first-axis, hence, it corresponds to a rotation around the first-axis, and thus commutes with $S(\lVert \bar{p}\rVert)$. It follows from (\ref{adeq:a2}) that $P_{\vec{0},p_{\text{rest}}}^{\text{little}}(a,x,O,p)=P(\vec{0},O)$, as claimed.

\end{proof}

\subsection{If spacetime is Minkowski in $1+n$ dimensions then the elementary system corresponds to an Euclidean $n-$ball}
\label{elementary}
%We assume that spacetime is Minkowski in $1+n$ dimensions, with $n$ being a free variable. Below we prove the following lemma.

\begin{lemma}
\label{adlemma1}
Suppose that spacetime is Minkowski in $1+n$ dimensions. If Postulates \hyperref[P1first]{1} -- \hyperref[P3]{3} hold then the state space and the space of effects of the elementary system corresponds to an Euclidean ball of dimension $d=n$, that is, $\mathcal{S}_{\text{elem}}= \mathcal{S}_{\text{ball}}^{(n)}$ and $\mathcal{E}_{\text{elem}}= \mathcal{E}_{\text{ball}}^{(n)}$, respectively.
\end{lemma}

The proof of Lemma \ref{adlemma1} uses Lemma \ref{newlemma}, given in section \ref{technical}, and Lemmas \ref{adlemma3} and \ref{adlemma4} given below, and is provided at the end of this subsection. 
%A brief summary of the proof is given below.

%The proof of Lemma \ref{adlemma1} uses the following lemmas.

Lemma \ref{adlemma3} uses Nontrivial Structure (Postulate \hyperref[P1first]{1.2}) and Lemma \ref{newlemma} to show that the inequality $d\geq n$ between the number $d$ of real internal degrees of freedom of a massive particle of the type $\mathcal{P}$ and the number of spatial dimensions $n$ must hold.

\begin{lemma}
\label{adlemma3}
If spacetime is Minkowski in $1+n$ dimensions and Postulates \hyperref[P1first]{1} and \hyperref[P2]{2.1} hold then the dimension $d$ for the state space $\mathcal{S}$ of any of the internal degrees of freedom of a particle of the type $\mathcal{P}$ is bounded by $d\geq n$. Let $p_{\text{rest}}=(m,0,\ldots,0)^{\text{t}}$ be the $1+n$ momentum of a particle of the type $\mathcal{P}$, which has mass $m>0$. If $d=n$ then in Lemma \ref{newlemma}, we have
\begin{equation}
R^{\text{st}}_{p_\text{rest}}(P(\vec{0},O))=P(\vec{0},O),\\
\end{equation}
for all pure rotations $O$, i.e. with $\tilde{O}\in\text{SO}(n)$. 
\end{lemma}

\begin{proof}
From Lemma \ref{newlemma}, the states $\zeta\in\mathcal{S}\subset\mathbb{R}^{d+1}$ for any of  the internal degrees of freedom of a stationary massive particle of the type $\mathcal{P}$ transform as $R^{\text{st}}_{p_\text{rest}}(P(x,\Lambda)) \zeta$ under a Poincar{\'e} transformation $P(x,\Lambda)\in \mathfrak{Poin}$, where $R^{\text{st}}_{p_\text{rest}}$ is a representation of $\text{SO}(n)$. 
From Postulate \hyperref[P1first]{1.2}, we discard the possibility that $R^{\text{st}}_{p_\text{rest}}(P(x,\Lambda))$ be the trivial representation $R^{\text{st}}_{p_\text{rest}}(P(x,\Lambda))=I^{\text{st}}$, where $I^{\text{st}}$ is the identity acting on $\mathcal{S}$. The nontrivial representation of $\text{SO}(n)$ with the smallest dimension is $\text{SO}(n)$ itself:
\begin{equation}
\label{ady22}
R^{\text{st}}_{p_\text{rest}}(P(\vec{0},O))=P(\vec{0},O),
\end{equation}
for all
 \begin{equation}
 \label{ady23}
 O\equiv\Bigl(\begin{smallmatrix}
  1 & 0 \\
  0 & \tilde{O} 
 \end{smallmatrix}\Bigr),
\end{equation} 
  with $\tilde{O}\in \text{SO}(n)$.
 
We can also  have equivalent representations 
\begin{equation}
\label{ady24}
R^{\text{st}}_{p_\text{rest}}(P(\vec{0},O))=L^{-1}P(\vec{0},O)L,
\end{equation}
for any invertible linear map $L$ and for all $O$ given by (\ref{ady23}). We note that because $P(\vec{0},O)$ is a $(1+n)\times (1+n)$ matrix (given by (\ref{ady23})), and since $L$ is an invertible linear map, the representation given by (\ref{ady24}) is a matrix acting on $\mathbb{R}^{1+d'}$ with $d'\geq n$. If we have the representation (\ref{ady24}), we can transform the states, effects and transformations by 
\begin{eqnarray}
\label{ady25}
\zeta\rightarrow \zeta_{\text{new}}&\equiv&L\zeta,\nonumber\\
\varepsilon\rightarrow\varepsilon_{\text{new}}&\equiv& (L^{-1})^{\text{t}}\varepsilon,\nonumber\\
\tau\rightarrow \tau_{\text{new}}&\equiv&L\tau L^{-1},
\end{eqnarray}
for all $\zeta\in\mathcal{S}$, all $\varepsilon\in\mathcal{E}$ and all $\tau\in\mathcal{T}$, 
without changing the physics. This is because the outcome probabilities do not change:
 \begin{eqnarray}
 \label{ady26}
 \varepsilon_{\text{new}}(\zeta_{\text{new}})&=&\varepsilon(\zeta),\nonumber\\
 \varepsilon_{\text{new}}(\tau_{\text{new}}\zeta_{\text{new}})&=&\varepsilon(\tau\zeta),
 \end{eqnarray}
 for all $\zeta\in\mathcal{S}$, all $\varepsilon\in\mathcal{E}$ and all $\tau\in\mathcal{T}$.
 
Thus, from (\ref{ady26}), if we have a representation $R^{\text{st}}_{p_\text{rest}}$ given by (\ref{ady24}), we can apply the transformations (\ref{ady25}) and obtain a new representation given by (\ref{ady22}) that gives the same outcome probabilities. It follows that, in general, the dimension of the representation $R^{\text{st}}_{p_\text{rest}}$, and thus of the state space $\mathcal{S}$, satisfies $d\geq n$. Furthermore, if $d=n$ then (\ref{ady22}) holds, as claimed.
\end{proof}

Lemma \ref{adlemma4} below proves, from Postulates \hyperref[P1first]{1} and \hyperref[P2]{2} and using Lemma \ref{adlemma3}, that if the state space of the elementary system is an Euclidean ball of dimension $n$, i.e. if $\mathcal{S}_\text{elem}=\mathcal{S}_\text{ball}^{(n)}$, then the space of effects of the elementary system is $\mathcal{E}_\text{elem}=\mathcal{E}_\text{ball}^{(n)}$. The Existence of the Classical Bit (Postulate \hyperref[P1first]{2.2}) means that there is a pair of states that can be perfectly distinguished in a single measurement. Since $\mathcal{S}_\text{elem}=\mathcal{S}_\text{ball}^{(n)}$, this requires that there exists an effect of the form $\varepsilon_0=\frac{1}{2}\zeta_0\in\mathcal{E}_{\text{elem}}$ with $\zeta_0=\bigl(\begin{smallmatrix}
  1  \\
  r_0 
 \end{smallmatrix}\bigr)\in\mathcal{S}_\text{ball}^{(n)}$ pure, that is, with $r_0\in\mathbb{R}^n$ and $\lVert r_0\rVert=1$. Thus, we obtain from Poincar{\'e} invariance that by applying all the Poincar{\'e} transformations that correspond to the group of spatial rotations $\text{SO}(n)$, all extremal effects in $\mathcal{E}_\text{ball}^{(n)}$ must be included in the set of effects $\mathcal{E}_{\text{elem}}$ of the elementary system. Since $\mathcal{E}_{\text{elem}}$ is convex, we obtain that $\mathcal{E}_{\text{elem}}=\mathcal{E}_\text{ball}^{(n)}$.

\begin{lemma}
\label{adlemma4}
Suppose that spacetime is Minkowski in $1+n$ dimensions and Postulates \hyperref[P1first]{1} -- \hyperref[P3]{3} hold. If the elementary system has the state space $\mathcal{S}_{\text{elem}}= \mathcal{S}_{\text{ball}}^{(\text{n})}$, then its space of effects is given by $\mathcal{E}_{\text{elem}}= \mathcal{E}_{\text{ball}}^{(\text{n})}$.
\end{lemma}

\begin{proof}
First, we note that since $\mathcal{S}_{\text{elem}}= \mathcal{S}_{\text{ball}}^{(n)}$, the dimension of $\mathcal{S}_{\text{elem}}$ is $d=n$. Let $\mathcal{E}_{\text{elem}}$ be the set of effects associated to $\mathcal{S}_{\text{elem}}= \mathcal{S}_{\text{ball}}^{(n)}$. We show that $\mathcal{E}_{\text{elem}}=\mathcal{E}_{\text{ball}}^{(n)}$, where $\mathcal{E}_{\text{ball}}^{(n)}$ was defined as the convex hull of the zero effect $\varepsilon_{\bold{0}}\equiv\bigl(\begin{smallmatrix}
0\\ \bold{0}
\end{smallmatrix} \bigr)$, the unit effect $u\equiv\bigl(\begin{smallmatrix}1\\ \bold{0}
\end{smallmatrix} \bigr)$, and the extremal effects $\varepsilon_t\equiv\frac{1}{2}\bigl(\begin{smallmatrix}
1\\ t
\end{smallmatrix} \bigr)$, where $\bold{0},t \in\mathbb{R}^n$,  $\bold{0}$ is the null vector, and $\lVert t \rVert=1$ (see Fig. \ref{fig1}).

As previously mentioned, we have in general that 
\begin{equation}
\label{ady27}
\mathcal{E}_{\text{elem}}\subseteq \mathcal{E}_{\text{elem}}^{\text{norm}},
\end{equation}
where
\begin{equation}
\label{ady28}
\mathcal{E}_{\text{elem}}^{\text{norm}}\equiv \{\varepsilon\in \mathbb{R}^{1+n}\vert 0\leq \varepsilon(\zeta)\leq 1 ~\forall \zeta\in\mathcal{S}_{\text{ball}}^{(n)}\}\nonumber
\end{equation}
is the set of normalized effects associated to $\mathcal{S}_{\text{elem}}=\mathcal{S}_{\text{ball}}^{(n)}$. It is not difficult to see that 
\begin{equation}
\label{ady29}
\mathcal{E}_{\text{elem}}^{\text{norm}}= \mathcal{E}_{\text{ball}}^{(n)}.
\end{equation}
It follows from (\ref{ady27}) and (\ref{ady29}) that
\begin{equation}
\label{ady30}
\mathcal{E}_{\text{elem}}\subseteq \mathcal{E}_{\text{ball}}^{(n)}.
\end{equation}
Thus, from (\ref{ady30}), we only need to show that all effects $\varepsilon\in \mathcal{E}_{\text{ball}}^{(n)}$ are elements of $\mathcal{E}_{\text{elem}}$, that is, we need to show that
\begin{equation}
\label{ady31}
\mathcal{E}_{\text{ball}}^{(n)}\subseteq \mathcal{E}_{\text{elem}}.
\end{equation}

We show (\ref{ady31}). From Postulate \hyperref[P3]{3}, the elementary system can be physically implemented with some internal degrees of freedom of a particle of the type $\mathcal{P}$. We focus on these degrees of freedom in what follows. Thus, we can consider states $Z_{{p_\text{rest}},\zeta}^{\text{class}}\in \mathscr{S}^\text{class}$ and effects $\hat{E}_{{p_\text{rest}},\varepsilon}^{\text{class}}\in\mathscr{E}^{\text{class}}$ for a particle of the type $\mathcal{P}$ such that the internal degrees of freedom correspond to an elementary system, i.e. with $\zeta\in\mathcal{S}_\text{elem}$ and $\varepsilon\in\mathcal{E}_\text{elem}$.

From Postulate \hyperref[P2]{2.2}, there exists a pair of states $\zeta_0,\zeta_1\in\mathcal{S}_{\text{elem}}$ and a pair of effects $\varepsilon_0,\varepsilon_1\in \mathcal{E}_{\text{elem}}$ such that
\begin{equation}
\label{ady32}
\varepsilon_i(\zeta_j)=\delta_{i,j},
\end{equation}
for all $i,j\in\{0,1\}$. Since $\mathcal{S}_{\text{elem}}= \mathcal{S}_{\text{ball}}^{(n)}$, it is easy to see that this condition requires $\zeta_i$ to be pure, i.e. such that 
\begin{equation}
\label{ady33}
\zeta_i=\bigl(\begin{smallmatrix}
1\\ r_i
\end{smallmatrix} \bigr), \qquad \text{with}~\lVert r_i\rVert = 1,
\end{equation}
and $r_i\in\mathbb{R}^n$, for all $i\in\{0,1\}$, and satisfying 
 \begin{equation}
 \label{ady34}
 r_1=-r_0;
 \end{equation}
  and that $\varepsilon_i$ are extremal effects of the form 
  \begin{equation}
  \label{ady35}
  \varepsilon_i=\frac{1}{2}\zeta_i,
\end{equation}
for all $i\in\{0,1\}$.

From Postulate \hyperref[P3]{3}, the elementary system can be physically implemented with some internal degrees of freedom of a particle of the type $\mathcal{P}$. Thus, we can consider the state $Z_{{p_\text{rest}},\zeta_0}^{\text{class}}\in \mathscr{S}^\text{class}$ and the effect $\hat{E}_{{p_\text{rest}},\varepsilon_0}^{\text{class}}\in\mathscr{E}^{\text{class}}$, with $\zeta\in\mathcal{S}_\text{elem}$ and $\varepsilon\in\mathcal{E}_\text{elem}$. It follows straightforwardly from (\ref{eq:3}) and from (\ref{ady32}) that
\begin{equation}
\label{ady36}
\hat{E}_{{p_\text{rest}},\varepsilon_0}^{\text{class}}\bigl[Z_{{p_\text{rest}},\zeta_0}^{\text{class}}\bigr]=1.
\end{equation} 
From Poincar{\'e} invariance, this outcome probability remains the same after the Poincar{\'e} transformation $P(x,\Lambda)\in \mathfrak{Poin}$. More precisely, from  
(\ref{x0.1}), (\ref{eq:6}) and (\ref{ady36}), we have
\begin{equation}
\label{ady37}
\hat{E}_{\Lambda {p_\text{rest}},\varepsilon'_0}^{\text{class}}\bigl[Z_{\Lambda {p_\text{rest}},\zeta'_0}^{\text{class}}\bigr]=1,
\end{equation}
where
\begin{eqnarray}
\label{ady38}
\zeta'_0&=&R^{\text{st}}_{p_\text{rest}}(P(x,\Lambda))\zeta_0,\nonumber\\
\varepsilon'_0&=&R^{\text{ef}}_{p_\text{rest}}(P(x,\Lambda))\varepsilon_0,
\end{eqnarray}
are the transformed states and effects for the discrete degrees of freedom, after the Poincar{\'e} transformation $P(x,\Lambda)\in \mathfrak{Poin}$, respectively.

From (\ref{eq:3}), we see that (\ref{ady37}) is only possible if $\zeta'_0$ is a pure state 
\begin{equation}
\label{ady39}
 \zeta_0'=\bigl(\begin{smallmatrix}1\\ r_0'
\end{smallmatrix} \bigr),\qquad ~\text{with} ~\lVert r_0' \rVert =1,
\end{equation}
and $r_0'\in\mathbb{R}^n$, and if $\varepsilon'_0$  is an extremal effect
\begin{equation}
\label{ady40}
\varepsilon_0'=\frac{1}{2}\zeta_0'.
\end{equation}
Thus, from (\ref{ady38}) and (\ref{ady40}), we have
\begin{equation}
\label{ady41}
\varepsilon'_0=\frac{1}{2}R^{\text{st}}_{p_\text{rest}}(P(x,\Lambda))\zeta_0,
\end{equation}
where $P(x,\Lambda)\in \mathfrak{Poin}$. In (\ref{ady41}), we take $x=\vec{0}$ and $\Lambda=O$ a pure rotation, i.e. with $\tilde{O}\in\text{SO}(n)$. Since the dimension of the state space for the elementary system is $d=n$, from Lemma \ref{adlemma3} we have that  
\begin{equation}
R^{\text{st}}_{p_\text{rest}} (P(\vec{0},O))=O,\nonumber
\end{equation}
for all pure rotations $O$, i.e. for all $\tilde{O}\in\text{SO}(n)$. Thus, it follows that
 \begin{equation}
\label{ady42}
\varepsilon'_0=\frac{1}{2}O\zeta_0,
\end{equation}
 where $\tilde{O}\in\text{SO}(n)$. Therefore, from (\ref{ady42}), we can generate all extremal effects $\varepsilon'_0$ given by (\ref{ady39}) and (\ref{ady40}),
with
\begin{equation}
\label{ady43}
r'_0=\tilde{O} r_0,\nonumber
\end{equation}
for all $\tilde{O}\in\text{SO}(n)$. That is, all extremal effects in $\mathcal{E}_{\text{ball}}^{(n)}$ are elements of $\mathcal{E}_{\text{elem}}$. As previously mentioned, the zero effect $\varepsilon_{\bold{0}}$ and the unit effect $u$ are included in the space of effects. Since $\mathcal{E}_{\text{elem}}$ is convex, any convex combination of the extremal effects, the zero effect $\varepsilon_{\bold{0}}$ and the unit effect $u$ is an element of $\mathcal{E}_{\text{elem}}$, that is $\mathcal{E}_{\text{ball}}^{(n)}\subseteq \mathcal{E}_{\text{elem}}$, as claimed. 
\end{proof}

Having stated and proved Lemmas  \ref{newlemma} -- \ref{adlemma4}, we proceed to show Lemma \ref{adlemma1}.

%\subsubsection{Proof of Lemma \ref{adlemma1}}

\begin{proof}[Proof of Lemma \ref{adlemma1}]
From Lemma \ref{adlemma3}, the dimension of the state space $\mathcal{S}$ for the internal degrees of freedom of a particle of the type $\mathcal{P}$ satisfies $d\geq n$. From Postulate \hyperref[P3]{3}, the elementary system can be physically implemented in some internal degrees of freedom of a particle of the type $\mathcal{P}$ and has the state space $\mathcal{S}_{\text{elem}}$ of minimum dimension $d_{\text{elem}}$ consistent with the considered postulates. Since consistency with Postulates \hyperref[P1first]{1} and \hyperref[P2]{2} is satisfied, the elementary system achieves $d_{\text{elem}}=n$. It will be seen later that this is consistent with Postulates \hyperref[P4]{4} -- \hyperref[P7]{7} too. Thus, we assume that the case $d_{\text{elem}}=n$ is achieved. 

It follows from the previous discussion and from Lemma \ref{adlemma3} that
\begin{equation}
\label{ady44}
R_{\text{dist}}^{p_\text{rest}}(P(\vec{0},O))=O,
\end{equation}
for all $\tilde{O}\in\text{SO}(n)$. We show that in this case $\mathcal{S}_{\text{elem}}=\mathcal{S}_{\text{ball}}^{(n)}$. Consider a state $Z_{{p_\text{rest}},\zeta}^{\text{class}}\in \mathscr{S}^\text{class}$. After a Poincar{\'e} transformation $P(x,\Lambda)\in \mathfrak{Poin}$, we obtain from (\ref{eq:6}) that
$Z_{{p_\text{rest}},\zeta}^{\text{class}}$ transforms into the state
\begin{equation}
\hat{\bold{R}}^{\text{st}}(P(x,\Lambda))[Z_{{p_\text{rest}},\zeta}^{\text{class}}]=Z_{\Lambda {p_\text{rest}},\zeta'}^{\text{class}},\nonumber
\end{equation}
where
\begin{equation}
\zeta' = R^{\text{st}}_{p_{\text{rest}}}(P(x,\Lambda))\zeta
\end{equation}
is the transformed state for the internal degrees of freedom.
Thus, $R^{\text{st}}_{p_\text{rest}}(P(x,\Lambda))$ is an allowed transformation on $\mathcal{S}_{\text{elem}}$ and $\zeta'\in\mathcal{S}_{\text{elem}}$.

Consider a state 
\begin{equation}
\label{ady45}
\zeta_r\equiv\bigl(\begin{smallmatrix}
1\\ r
\end{smallmatrix} \bigr)\in\mathcal{S}_{\text{elem}},
\end{equation} whose vector $r\in\mathbb{R}^n$ has the biggest Euclidean norm
\begin{equation}
\lVert r \rVert >0
\end{equation}
 among all states in $\mathcal{S}_{\text{elem}}$. Without loss of generality, we can take 
\begin{equation}
\lVert r \rVert =1,
\end{equation}
as we argue. In general, the state $\zeta_r$ with the biggest Euclidean norm $\lVert r \rVert$ must satisfy $\lVert r \rVert > 0$. Otherwise, $\mathcal{S}$ would have a single state $\zeta_{\bold{0}}=(\begin{smallmatrix}
1\\ \bold{0}
\end{smallmatrix} \bigr)$. This would mean that $R^{\text{st}}_{p_\text{rest}}(P(\vec{0},\Lambda))$ is a trivial representation $R^{\text{st}}_{p_\text{rest}}(P(x,\Lambda))=I^{\text{st}}$ acting on $\mathbb{R}$. Since this case is discarded from Postulate \hyperref[P1first]{1.2}, we have $\lVert r \rVert >0$. If $\lVert r \rVert \neq1$ we can rescale the states and effects by applying the transformation
\begin{equation}
L_c=\text{diag}(1,c,\ldots,c),
\end{equation}
 with 
 \begin{equation}
 c=\lVert r \rVert ^{-1},
 \end{equation}
  to $\mathcal{S}_{\text{elem}}$ and $(L_c)^{-1}$ to $\mathcal{E}_{\text{elem}}$, which leaves all outcome probabilities unchanged, and so describes the same physics. 

From (\ref{ady44}), the state $\zeta_r$ is transformed into the state $\zeta_{\tilde{O}r}$ by applying the transformation $P(\vec{0},O)$, for all $\tilde{O}\in\text{SO}(n)$. This means that $\mathcal{S}_{\text{elem}}$ contains  a set of states  $\zeta_r$ with $r\in\mathbb{R}^n$ defining a unit sphere in $n-$dimensional Euclidean space. Since $\mathcal{S}_{\text{elem}}$ is convex, any state $\zeta_r$ with $r$ in a unit $n-$ball is in $\mathcal{S}_{\text{elem}}$. Furthermore, we said that the biggest norm of the vectors $r$ is $1$. It follows that the state space $\mathcal{S}_{\text{elem}}$ must be precisely an Euclidean $n-$ball $\mathcal{S}_{\text{ball}}^{(n)}$, defined in Fig. \ref{fig1}.

Finally, it follows from Lemma \ref{adlemma4} that $\mathcal{E}_{\text{elem}}=\mathcal{E}_{\text{ball}}^{(n)}$.
\end{proof}

\subsection{If spacetime is Minkowski in $1+3$ dimensions then the elementary system is locally identical to the qubit}
\label{derivingqubit}

%Now we assume that spacetime is Minkowski in $1+3$ dimensions. Thus, Lemmas \ref{newlemma} -- \ref{adlemma4} hold with $n=3$.

%The following lemma shows, from Postulates \hyperref[P1first]{1} -- \hyperref[P4]{4} and using Lemma \ref{adlemma1} that the state space, the space of effects and the group of reversible transformations of the elementary system are those of the qubit. 

\begin{lemma}
\label{qubitlemma}
If spacetime is Minkowski in $1+3$ dimensions and Postulates \hyperref[P1first]{1} -- \hyperref[P4]{4} hold then the state space, the space of effects and the group of reversible transformations of the elementary system are those of the qubit, i.e. $\mathcal{S}_{\text{elem}}= \mathcal{S}_{\text{BB}}$, $\mathcal{E}_{\text{elem}}= \mathcal{E}_{\text{BB}}$ and $\mathcal{R}_{\text{elem}}= \mathcal{R}_{\text{BB}}$.
\end{lemma}

\reff{This means that the measurement statistics and reversible transformations for the elementary system are identical to those of the qubit. Therefore, the elementary system is locally identical to the qubit.}

\begin{proof}[Proof of Lemma \ref{qubitlemma}]
From Postulates \hyperref[P1first]{1} -- \hyperref[P3]{3} and Lemma \ref{adlemma1}, the state space and the space of effects of the elementary system correspond to an Euclidean ball of dimension $d=n=3$, that is, 
$\mathcal{S}_{\text{elem}}= \mathcal{S}_{\text{ball}}^{(3)}$ and $\mathcal{E}_{\text{elem}}= \mathcal{E}_{\text{ball}}^{(3)}$. The state space and the space of effects of the qubit correspond to the Bloch ball: $\mathcal{S}_{\text{BB}}\equiv \mathcal{S}_{\text{ball}}^{(3)}$ and $\mathcal{E}_{\text{BB}}\equiv \mathcal{E}_{\text{ball}}^{(3)}$ (see Fig. \ref{fig1}). Thus, we have $\mathcal{S}_{\text{elem}}= \mathcal{S}_{\text{BB}}$ and $\mathcal{E}_{\text{elem}}= \mathcal{E}_{\text{BB}}$.

From Continuous Reversibility (Postulate \hyperref[P4]{4}), every pair of pure states must be connected by a continuous reversible transformation. The set of pure states in the Bloch ball is the Bloch sphere $\mathcal{S}_{\text{BB}}\equiv \mathcal{S}_{\text{sphere}}^{(3)}$. Any continuous reversible transformation that takes pure states into pure states corresponds to a rotation of the Bloch vector. The group of rotations of the Bloch vector is $\mathcal{R}_{\text{BB}}$ (see Fig. \ref{fig1}).
\end{proof}

\section{A physical derivation of finite dimensional quantum theory and the number of spatial dimensions in Minkowski spacetime}
\label{quantumtheorysection}

The results of this section are twofold. First, in section \ref{firstderivation} we assume that spacetime is Minkowski in $1+3$ dimensions. From Postulates \hyperref[P1first]{1} -- \hyperref[P7]{7} and using Lemma \ref{qubitlemma} and the results of Ref. \cite{TMSM12,MMAP13}, Theorem \ref{firsttheorem} shows that any physical system of any finite dimension can be described by finite dimensional quantum theory \reff{(see a summary of the proof in Fig. \ref{fig3}).}

%We give a brief summary of our proof of Theorem \ref{firsttheorem} (see Fig. \ref{fig3}). We assume that spacetime is Minkowski in $1+3$ dimensions. From Postulates \hyperref[P1first]{1} -- \hyperref[P4]{4}, Lemma \ref{qubitlemma} states that the elementary system is locally identical to the qubit. The result of Ref. \cite{TMSM12} states that for any theory in which the individual systems are identical qubits and that satisfies Continuous Reversibility (Postulate \hyperref[P4]{4}) and Tomographic Locality (Postulate \hyperref[P5]{5}), if the theory admits any continuous reversible entangling interaction between systems, i.e. the Existence of Entanglement (Postulate \hyperref[P6]{6}) holds, then the allowed states, measurements, and transformations must be identical to those in quantum theory. Thus, by having $N$ massive particles of the type $\mathcal{P}$ encoding each an elementary system, which are identical qubits, in their internal degrees of freedom, we obtain from Postulates \hyperref[P1first]{1} -- \hyperref[P6]{6} that the internal degrees of freedom of these particles are described by $N-$qubit theory. Finally, as argued by Ref. \cite{MMAP13}, Universal Encoding (Postulate \hyperref[P7]{7}) states that for any physical system, any state of finite dimension can be reversibly encoded in a sufficient large number of elementary systems. It follows that any physical system of any finite dimension can be described by finite dimensional quantum theory.

\begin{figure*}
\includegraphics[scale=0.32]{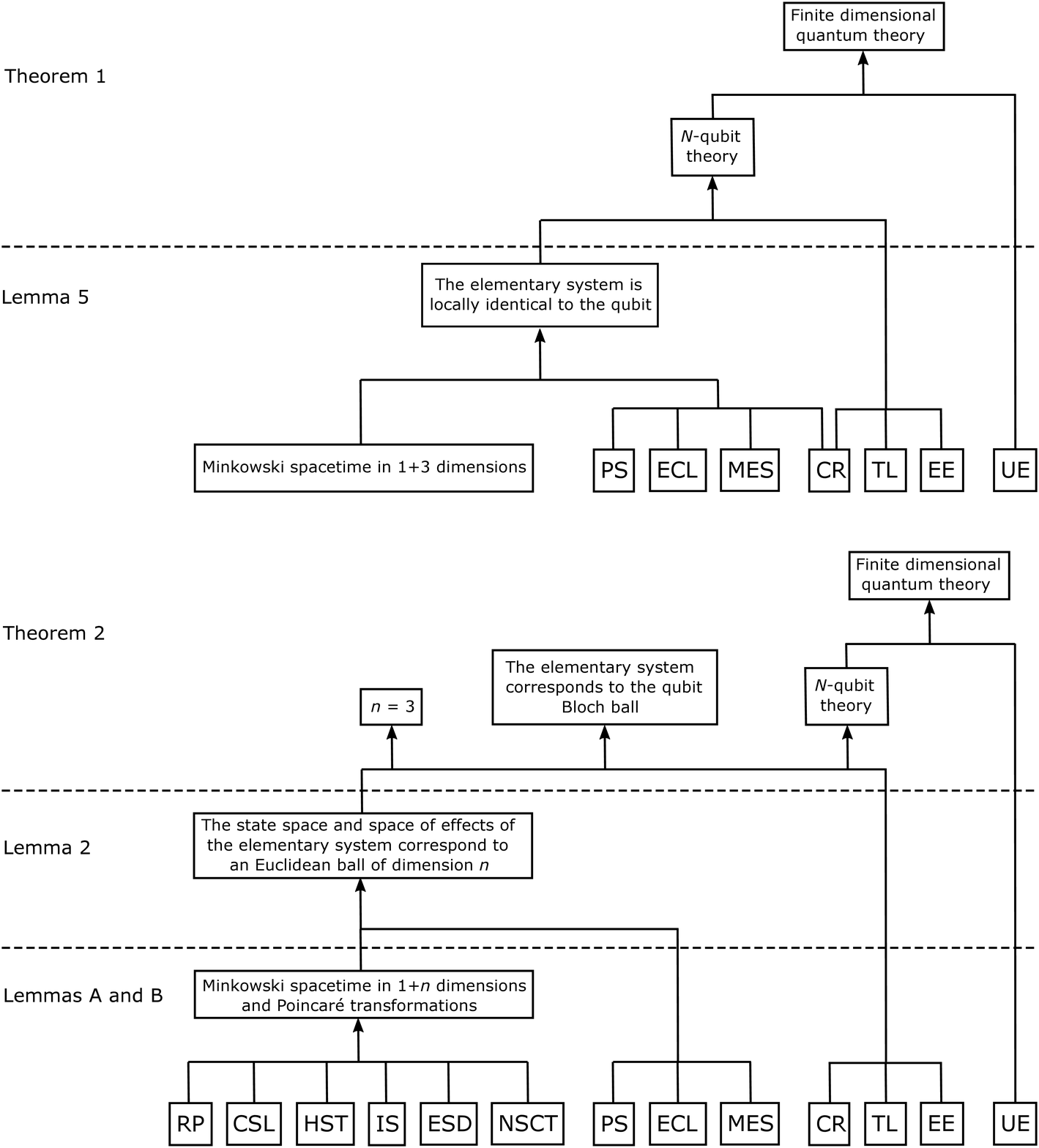}
 \caption{\label{fig3} \textbf{Summary of our reconstructions of Minkowski spacetime and finite dimensional quantum theory.} The considered physical principles and postulates are the Relativity Principle (RP), Constancy of the Speed of Light (CSL), Homogeneity of Space and Time (HST), Isotropy of Space (IS), Euclidean Spatial Distance (ESD), Non-Singularity of Coordinate Transformations (NSCT), Poincar{\'e} Structure (PS), Existence of a Classical Limit (ECL), Minimality of the Elementary System (MES), Continuous Reversibility (CR), Tomographic Locality (TL), Existence of Entanglement (EE) and Universal Encoding (UE), given in sections \ref{derivingMinkowski}, \ref{secmainpostulate} and \ref{particlepostulates}. The dotted lines indicate which lemmas and theorems in this paper prove the stated results. Top: Assuming that spacetime is Minkowski in $1+3$ dimensions, Lemma \ref{qubitlemma} in section \ref{derivingqubit} shows that the elementary system is locally identical to the qubit. This result is used by Theorem \ref{firsttheorem} in section \ref{firstderivation} to reconstruct finite dimensional quantum theory. Bottom: Minkowski spacetime in $1+n$ dimensions and the Poincar{\'e} transformations are derived by Lemmas \ref{Minkowski} and \ref{Poincare} in section \ref{derivingMinkowski} and  Appendix \ref{app} from well established physical principles and postulates. Then, assuming that spacetime is Minkowski in $1+n$ dimensions, Lemma \ref{adlemma1} shows in section \ref{elementary} that the state space and space of effects of the elementary system corresponds to an Euclidean ball of dimension $n$. This result is used by Theorem \ref{secondtheorem} in section \ref{secondderivation} to derive that the number of spatial dimensions is $n=3$, the elementary system corresponds to the qubit Bloch ball, and that any physical system of any finite dimension can be described by finite dimensional quantum theory.}
\end{figure*}

Second, in section \ref{secondderivation} we assume that spacetime is Minkowski in $1+n$ dimensions and leave $n$ as a free variable. As discussed in section \ref{derivingMinkowski}, this follows from well established physical principles.  From Postulates \hyperref[P1first]{1} -- \hyperref[P7]{7} and using Lemma \ref{adlemma1} and the results of Refs. \cite{TMSM12,MMAP13,MMPA14}, Theorem \ref{secondtheorem} shows that the elementary system is the qubit, the number of spatial dimensions is $n=3$, and any physical system of any finite dimension can be described by finite dimensional quantum theory \reff{(see a summary of the proof in Fig. \ref{fig3}).}

\subsection{A physical derivation of finite dimensional quantum theory if spacetime is Minkowski in $1+3$ dimensions}
\label{firstderivation}

%We assume that spacetime is Minkowski in $1+3$ dimensions. Below we prove the following theorem.

\begin{theorem}
\label{firsttheorem}
Consider that spacetime is Minkowski in $1+3$ dimensions and that Postulates \hyperref[P1first]{1} -- \hyperref[P7]{7} hold. Then, any physical system of any finite dimension can be described by finite dimensional quantum theory.
\end{theorem}

To prove Theorem \ref{firsttheorem}, we use Lemma \ref{qubitlemma} given above, and Lemma \ref{lemmadelatorre} given below. The following Lemma is shown in Ref. \cite{TMSM12} (Theorem 2 of Ref. \cite{TMSM12}).

\begin{lemma}
\label{lemmadelatorre}
Consider any locally tomographic theory in
which the individual systems are identical qubits. If the
theory admits any continuous reversible entangling interaction between
systems, then the allowed states, measurements, and
transformations must be identical to those in quantum
theory.
\end{lemma}

\begin{proof}[Proof of Theorem \ref{firsttheorem}]
From Postulates \hyperref[P1first]{1} -- \hyperref[P3]{4}, and from Lemma \ref{qubitlemma}, the elementary system is locally identical to the qubit, i.e. the state space, the space of effects and the set of reversible transformations of the elementary system correspond to those of the qubit: $\mathcal{S}_{\text{elem}}= \mathcal{S}_{\text{BB}}$, $\mathcal{E}_{\text{elem}}= \mathcal{E}_{\text{BB}}$ and $\mathcal{R}_{\text{elem}}= \mathcal{R}_{\text{BB}}$.

Now we consider $N$ massive particles of the type $\mathcal{P}$ where each of them encodes an elementary system, which is locally identical to a qubit, in its internal degrees of freedom. For these $N$ elementary systems, which are $N$ identical local qubits, it follows from Continuous Reversibility (Postulate \hyperref[P4]{4}), Tomographic Locality (Postulate \hyperref[P5]{5}) and the Existence of Entanglement (Postulate \hyperref[P6]{6}), and from Lemma \ref{lemmadelatorre} that the allowed states, measurements and transformations must be identical to those of quantum theory.

Finally, as argued in Ref. \cite{MMAP13}, Universal Encoding (Postulate \hyperref[P7]{7}) implies that any finite dimensional state space of any physical system can be perfectly encoded in a sufficiently large number of elementary systems. Thus, any finite dimensional state space in the theory, together with its space of effects and set of transformations is described by finite dimensional quantum theory. 
\end{proof}

\subsection{A physical derivation of the qubit, of finite dimensional quantum theory and of $n=3$ if spacetime is Minkowski in $1+n$ dimensions}
\label{secondderivation}

%We assume that spacetime is Minkowski in $1+n$ dimensions, with $n$ being a free variable. Below we prove the following theorem.

\begin{theorem}
\label{secondtheorem}
Suppose that spacetime is Minkowski in $1+n$ dimensions, for $n\in\mathbb{N}$, and that Postulates \hyperref[P1first]{1} -- \hyperref[P7]{7} hold. Then, the elementary system is the qubit, the number of spatial dimensions is $n=3$, and any physical system of any finite dimension can be described by finite dimensional quantum theory.
\end{theorem}

To prove Theorem \ref{secondtheorem}, we use Lemmas \ref{adlemma1} and \ref{lemmadelatorre} given above, and Lemma \ref{lemmamasanesmuller} given below.
%We recall Lemma \hyperref[lemmadelatorrerepeated]{6}, used in section \ref{firstderivation} and shown in Ref. \cite{TMSM12} (Theorem 2 of Ref. \cite{TMSM12}).
%\begin{lemma6*}
%\label{lemmadelatorrerepeated}
%Consider any locally tomographic theory in which the individual systems are identical qubits. If the theory admits any continuous reversible entangling interaction between systems, then the allowed states, measurements, and transformations must be identical to those in quantum theory.
%\end{lemma6*}
The following lemma is shown in Ref. \cite{MMPA14} (Theorem 1 of Ref. \cite{MMPA14}).

\begin{lemma}
\label{lemmamasanesmuller}
Consider a bipartite system $AB$, where the local system $A$ and the local system $B$ have the state space and space of effects of an Euclidean ball of dimension $n$, $\mathcal{S}_{\text{ball}}^{(n)}$ and  $\mathcal{E}_{\text{ball}}^{(n)}$, respectively, for $n\in\mathbb{N}$. Consider any group $\mathcal{R}$ of continuous reversible transformations that acts transitively on the set of pure states of $\mathcal{S}_{\text{ball}}^{(n)}$, i.e on $\mathcal{S}_{\text{sphere}}^{(n)}$, with $\mathcal{R}$ different to the group of qubit reversible transformations: $\mathcal{R}_{\text{BB}}\equiv\Bigl\{\tau\equiv\Bigl(\begin{smallmatrix}
  1 & 0 \\
  0 & \tilde{\tau} 
 \end{smallmatrix}\Bigr)\big\vert \tilde{\tau}\in\text{SO}(3)\Bigr\}$.
For any connected group $\mathcal{R}_{AB}$ acting on the set of pure states of the bipartite system $AB$ satisfying $\mathcal{R}\times\mathcal{R}\subseteq \mathcal{R}_{AB}$ and $(\varepsilon\otimes\varepsilon)(\mathcal{R}_{AB}(\zeta\otimes\zeta))\subseteq[0,1]$, for any  effect $\varepsilon\in\mathcal{E}_{\text{ball}}^{(n)}$ and any sate $\zeta\in\mathcal{S}_{\text{ball}}^{(n)}$, there is no entanglement interaction between the systems $A$ and $B$.
\end{lemma}

\begin{proof}[Proof of Theorem \ref{secondtheorem}]
From Postulates \hyperref[P1first]{1} -- \hyperref[P3]{3}, and from Lemma \ref{adlemma1}, the state space and the space of effects of the elementary system corresponds to an Euclidean ball of dimension $d=n$, that is, $\mathcal{S}_{\text{elem}}= \mathcal{S}_{\text{ball}}^{(n)}$ and $\mathcal{E}_{\text{elem}}= \mathcal{E}_{\text{ball}}^{(n)}$, respectively. The set of pure states for the elementary system is given by $\mathcal{S}_{\text{sphere}}^{(n)}$.

Now consider a bipartite system $AB$, where the local systems $A$ and $B$ are elementary systems,  with state spaces $\mathcal{S}_A=\mathcal{S}_{\text{ball}}^{(n)}$ and $\mathcal{S}_B=\mathcal{S}_{\text{ball}}^{(n)}$, and with space of effects $\mathcal{E}_A=\mathcal{E}_{\text{ball}}^{(n)}$ and $\mathcal{E}_B=\mathcal{E}_{\text{ball}}^{(n)}$, respectively. This can be implemented by having two massive particles of the type $\mathcal{P}$ where each of them encodes an elementary system in its internal degrees of freedom, for instance. The joint state space $\mathcal{S}_{AB}$ must include the tensor product of local states. Similarly, the joint space of effects $\mathcal{E}_{AB}$ must include the tensor product of local effects. That is, $\mathcal{S}_A\otimes_{\text{min}}\mathcal{S}_B\subseteq \mathcal{S}_{AB}$ and $\mathcal{E}_A\otimes_{\text{min}}\mathcal{E}_B\subseteq \mathcal{E}_{AB}$, where  $\mathcal{S}_A\otimes_{\text{min}}\mathcal{S}_B\equiv \text{convex hull}\{\zeta\otimes \zeta'\vert \zeta\in\mathcal{S}_A,\zeta'\in\mathcal{S}_B\}$ is the minimal tensor product, and where $\mathcal{E}_A\otimes_{\text{min}}\mathcal{E}_B\equiv \text{convex hull}\{\varepsilon\otimes \varepsilon'\vert \varepsilon\in\mathcal{E}_A,\varepsilon'\in\mathcal{E}_B\}$.

From Continuous Reversibility (Postulate \hyperref[P4]{4}), for every pair of pure states there exists a continuous reversible transformation that transforms one state into  the other. Let the group of continuous reversible transformations acting on the set of pure states $\mathcal{S}_{\text{sphere}}^{(n)}$ of the elementary system be the group $\mathcal{R}$. This group must act transitively on $\mathcal{S}_{\text{sphere}}^{(n)}$. Let $\mathcal{R}_{AB}$ be a connected group of continuous reversible transformations acting on the set of pure states of the bipartite state space $\mathcal{S}_{AB}$. As mentioned above, we must have $\mathcal{S}_A\otimes_{\text{min}}\mathcal{S}_B\subseteq \mathcal{S}_{AB}$. Thus, $(\zeta\otimes\zeta)\in\mathcal{S}_{AB}$, for any  state $\zeta\in\mathcal{S}_{\text{ball}}^{(n)}$. Similarly, we must have $\mathcal{R}_{AB}(\zeta\otimes\zeta)\in\mathcal{S}_{AB}$, for any state $\zeta\in\mathcal{S}_{\text{ball}}^{(n)}$. It must also hold that $(\varepsilon\otimes\varepsilon)\in\mathcal{E}_{AB}$, for any effect $\varepsilon\in\mathcal{E}_{\text{ball}}^{(n)}$. Therefore, it must hold that $(\varepsilon\otimes\varepsilon)(\mathcal{R}_{AB}(\zeta\otimes \zeta))\subseteq[0,1]$, for any  effect $\varepsilon\in\mathcal{E}_{\text{ball}}^{(n)}$ and for any state $\zeta\in\mathcal{S}_{\text{ball}}^{(n)}$. Clearly, the tensor product of two local continuous reversible interactions acting respectively on the subsystems $A$ and $B$ is a continuous reversible interaction acting on the bipartite system $AB$. That is, we have $\mathcal{R}\times\mathcal{R}\subseteq\mathcal{R}_{AB}$. It follows from Continuous Reversibility (Postulate \hyperref[P4]{4}), Tomographic Locality (Postulate \hyperref[P5]{5}) and the Existence of Entanglement (Postulate \hyperref[P6]{6}), and from Lemmas \ref{lemmadelatorre} and \ref{lemmamasanesmuller}, that the bipartite state space $\mathcal{S}_{AB}$, the bipartite space of effects $\mathcal{E}_{AB}$, and the bipartite set of allowed transformations $\mathcal{T}_{AB}$ are identical to those of quantum theory for a two qubit system and the dimension of the state space is $n=3$. This means in particular that the elementary system is identical to a qubit.

Now we consider $N$ massive particles of the type $\mathcal{P}$ where each of them encodes an elementary system, which is locally identical to a qubit, in its internal degrees of freedom. For these $N$ elementary systems, which are $N$ identical local qubits, it follows from Continuous Reversibility, Tomographic Locality and the Existence of Entanglement, and from Lemma \ref{lemmadelatorre} that the allowed states, measurements and transformations must be identical to those of quantum theory.

Finally, as argued in Ref. \cite{MMAP13}, Universal Encoding (Postulate \hyperref[P7]{7}) implies that any finite dimensional state space of any physical system can be perfectly encoded in a sufficiently large number of elementary systems. Thus, any finite dimensional state space in the theory, together with its space of effects and set of transformations is described by finite dimensional quantum theory. 

\end{proof}

\section{Discussion and open problems}
\label{discussion}
%\subsection{Conclusions}

%We have found a connection between the mathematical structures of spacetime and quantum theory. Although deriving Minkowski spacetime from physical principles is well established, we have left the number of spatial dimensions $n$ as a free variable. From Poincar{\'e} invariance and other physical postulates, we have derived finite dimensional quantum theory and that $n=3$. It would be interesting to investigate variations of our derivation. For example, can Minimality of the Elementary System (Postulate \ref{P3}) be discarded to make our derivation stronger? More generally, it is interesting to investigate further connections between the mathematical formalisms of spacetime and quantum theory.

%We reconstructed the qubit Bloch ball from Poincar{\'e} invariance and other physical properties. Then, we followed the arguments of previous works \cite{TMSM12,MMAP13} to reconstruct finite dimensional quantum theory. It is interesting to investigate how fundamental the principle of Poincar{\'e} invariance is in our reconstruction, and more generally, in the foundations of quantum theory. 

Quantum theory and relativity are the most fundamental theories in physics. In our view, it would not be surprising if they were connected at a deep level. In this paper we have proposed the postulate of \hyperref[P1first]{Poincar{\'e} Structure} suggesting that the state space and the space of measurements of some physical systems are constrained in such a way that the states and measurements must transform as nontrivial representations of the group of symmetry transformations of Minkowski spacetime, which is the proper orthochronous Poincar{\'e} group. From this and other physically sensible postulates, and with the help of results of Refs. \cite{TMSM12,MMAP13,MMPA14}, we have reconstructed finite dimensional quantum theory and derived the number of spatial dimensions of Minkowski spacetime.

Am important piece in our reconstruction is Lemma \ref{newlemma}, which roughly says that, from Poincar{\'e} invariance, the states and effects for a massive particle's internal degrees of freedom that has classical well defined $1+n$ momentum in Minkowski spacetime of $1+n$ dimensions must transform as representations of $\text{SO}(n)$, with the measurement outcome probabilities remaining invariant.  This is an encouraging reason to investigate what sets of states $\mathcal{S}$ and effects $\mathcal{E}$ are consistent with this condition when such representations are nontrivial. 

In particular, it would be interesting to investigate for which bipartite systems with local states spaces $\mathcal{S}$ and local spaces of effects $\mathcal{E}$ of this type there can exist entanglement. If by imposing \hyperref[P4]{Continuous Reversibility} and \hyperref[P5]{Tomographic Locality} it can be shown that these type of bipartite systems can only have entanglement if the local states transform as the representation  of $\text{SO}(n)$ given by $\text{SO}(n)$ itself, then our result follows without the assumption we have made in our postulate \hyperref[P3]{Minimality of the Elementary System} that the dimension $d_\text{elem}$ of the elementary system takes the minimum value that is consistent with the postulates. Investigating the previous question would be an interesting extension of the result of Ref. \cite{MMPA14} used in our analysis, which roughly says that for a bipartite system with local state and effect spaces given by Euclidean balls of dimension $d$ satisfying continuous reversibility and tomographic locality, entanglement can only exist if $d=3$, in which case the bipartite system is described by quantum theory.

It would also be interesting to investigate the sets of correlations arising from local measurements on bipartite (and multipartite) systems with local state and effect spaces $\mathcal{S}$ and $\mathcal{E}$ of the type mentioned above, where the states and effects transform as nontrivial representations of $\text{SO}(n)$ and where the outcome probabilities remain invariant. In particular,
can the quantum Tsirelson bound \cite{C80} on the CHSH Bell inequality \cite{CHSH69} be violated by some
bipartite system of the form described? Considering this question is to some extent motivated by the results of Refs. \cite{BBBEW10,AACHKLMP10} which show that the correlations obtained for bipartite systems with local measurements described by quantum theory that satisfy the no-signalling principle can be obtained with quantum theory, even if the joint system is in principle not quantum.

Finally, as previously mentioned, an important motivation to investigate foundational physical principles of quantum theory is to explore new theories that follow by modifying these principles. The problem of unifying gravity and quantum theory is a compelling reason to do so \cite{H07,H16}. The postulate of \hyperref[P1']{Structure from the Spacetime Symmetries} that we have proposed in this paper generalizes the postulate of \hyperref[P1first]{Poincar{\'e} Structure} to arbitrary spacetimes with arbitrary groups of symmetry transformations. This postulate allows us to explore possible modifications of quantum theory in spacetimes that are not Minkowski. It would be very interesting to investigate the implications for quantum theory within the framework of GPTs arising from this postulate in curved spacetimes of general relativity, and more broadly in spacetimes of natural modifications of general relativity.

\begin{acknowledgments}
The author acknowledges helpful conversations with Serge Massar, Stefano Pironio and Lucien Hardy. The author began this work at the Laboratoire d'Information
Quantique, Universit\'{e} libre de Bruxelles, with financial
support from the European Union under the project
QALGO, from the F.R.S.-FNRS under the project
DIQIP and from the InterUniversity Attraction Poles of
the Belgian Federal Government through project Photonics@be. The author continued and completed this work at the Centre for Quantum Information and Foundations, DAMTP, University of Cambridge, with financial support from the UK Quantum
  Communications Hub grant no. EP/T001011/1.
\end{acknowledgments}

\appendix
\section{Proofs of Lemmas \ref{Minkowski} and \ref{Poincare}}
\label{app}

\begin{proof}[Proof of Lemma \ref{Minkowski}]
We present a proof that is close to the one given by Ref. \cite{LandauLifshitzbook}. Let $F$ and $F'$ be any two inertial reference frames. Let $E_0$ and $E_1$ be two spacetime events with respective spacetime coordinates $x$ and $y$ in  reference frame $F$, and with spacetime coordinates $x'$ and $y'$ in reference frame $F'$. Let the spacetime intervals between these spacetime events be $\Delta s$ and $\Delta s'$ in the reference frames $F$ and $F'$, respectively. \reff{We define these by 
\begin{eqnarray}
\label{M0}
(\Delta s)^2 = -(y_0-x_0)^2+\sum_{i=1}^n\bigl(y_i-x_i\bigr)^2,\nonumber\\
(\Delta s')^2 = -(y_0'-x_0')^2+\sum_{i=1}^n\bigl(y_i'-x_i'\bigr)^2.
\end{eqnarray}
From} Postulate \ref{ESD}, the spatial distance $\lvert \vec{y}-\vec{x}\rvert$ between the space locations $\vec{x}$ and $\vec{y}$ in the reference frame $F$ is given by Euclidean geometry: $\lvert\vec{y}-\vec{x}\rvert=\sqrt{\sum_{i=1}^n (y_i-x_i)^2}$. Similarly, in the reference frame $F'$, we have $\lvert\vec{y}'-\vec{x}'\rvert=\sqrt{\sum_{i=1}^n (y_i'-x_i')^2}$. Thus, from (\ref{M0}), we have
\begin{eqnarray}
\label{M1}
(\Delta s)^2 &=& -(y_0-x_0)^2+\lvert\vec{y}-\vec{x}\rvert^2,\nonumber\\
(\Delta s')^2 &=& -(y_0'-x_0')^2+\lvert\vec{y}'-\vec{x}'\rvert^2.
\end{eqnarray}

By definition, Minkowski spacetime in $1+n$ dimensions is the set of spacetime points $x\in\mathbb{R}^{1+n}$ such that the spacetime interval  $\Delta s$ 
is the same in all inertial reference frames, for any pair of spacetime points $x,y\in\mathbb{R}^{1+n}$. Since $F$ and $F'$ are arbitrary inertial reference frames, it remains to show that $\Delta s=\Delta s'$, for any pair of spacetime points $x,y\in\mathbb{R}^{1+n}$.

We first show that if $\Delta s=0$ then $\Delta s'=0$. We assume that
\begin{equation}
\label{M1.1}
\Delta s =0
\end{equation}
holds. From (\ref{M1}) and (\ref{M1.1}), we have
\begin{equation}
\label{M2}
\lvert\vec{y}-\vec{x}\rvert=\lvert y_0-x_0\rvert.
\end{equation}
Without loss of generality let $y_0\geq x_0$. Since we are using units in which the speed of light in vacuum is unity, i.e $c=1$, (\ref{M2}) means that a light signal leaving the space location $\vec{x}=(x_1,\ldots,x_n)$ at time $x_0$ and travelling through vacuum in a straight line reaches the space location $\vec{y}$ at time $y_0$. That is, in the frame $F$, the spacetime events $E_0$ and $E_1$ can correspond to a light signal travelling through vacuum from the spacetime point $x$ to the spacetime point $y$.

From Principle \ref{RP}, in the frame $F'$ the spacetime events $E_0$ and $E_1$ can also correspond to a light signal travelling through vacuum from the spacetime point $x'$ to the spacetime point $y'$. From Principle \ref{CSL}, the speed of light in vacuum is $c=1$ in both reference frames $F$ and $F'$. Thus, we have 
\begin{equation}
\label{M3}
\lvert\vec{y}'-\vec{x}'\rvert=\lvert y_0'-x_0'\rvert.
\end{equation}
It follows from (\ref{M1}) and from (\ref{M3}) that the spacetime interval $\Delta s'$ between the spacetime points $x'$ and $y'$ in the frame $F'$ satisfies $\Delta s'=0$,
as claimed.

We now consider that
\begin{equation}
\label{M5}
\Delta s\neq 0.
\end{equation}
We show that $\Delta s'=\Delta s$. We define $a$ by
\begin{equation}
\label{M6}
\Delta s=a \Delta s'.
\end{equation}
From (\ref{M5}), we have that $a\neq 0$. In general, $a$ is a function of the spacetime coordinates $x,y,x'$ and $y'$, and of the velocity $\vec{v}$ of the reference frame $F'$ with respect to the reference frame $F$. However, due to Principle \ref{HST}, $a$ cannot depend on any spacetime coordinates, as otherwise some spacetime coordinates would be treated in a special way. Furthermore, due to Principle \ref{IS}, $a$ cannot depend on the direction of $\vec{v}$, as otherwise different directions in space would be treated differently. Thus, $a$ can only depend on the magnitude $v$ of $\vec{v}$.

We now consider a third reference frame $F''$. Let $x''$ and $y''$ be the respective spacetime coordinates of the spacetime events $E_0$ and $E_1$ in the reference frame $F''$, and let $\Delta s''$ be the spacetime interval between $x''$ and $y''$. That is, from (\ref{M1}), we have $(\Delta s'')^2 = (y_0''-x_0'')^2-\lvert\vec{y}''-\vec{x}''\rvert^2$. Let $\vec{v}'$ be the velocity of $F''$ with respect to $F$ and let $\vec{v}''$ be the velocity of $F''$ with respect to $F'$. Thus, from (\ref{M6}), we have
\begin{eqnarray}
\label{M7}
\Delta s&=&a(v) \Delta s',\\
\label{M8}
\Delta s&=&a(v') \Delta s'',\\
\label{M9}
\Delta s'&=&a(v'') \Delta s''.
\end{eqnarray}
From (\ref{M5}), (\ref{M7}) and (\ref{M8}), we have $a(v)\neq 0$ and $a(v')\neq 0$. Thus, we can divide (\ref{M7}) and (\ref{M8}) by $a(v)$ and $a(v')$, respectively, and substitute $\Delta s'$ and $\Delta s''$ in (\ref{M9}), to obtain
\begin{equation}
\label{M10}
\frac{ \Delta s}{a(v)}=\frac{a(v'')\Delta s}{a(v')}.
\end{equation}
From (\ref{M5}), we can divide (\ref{M10}) by $\Delta s$ and multiply by $a(v')\neq 0$, to obtain
\begin{equation}
\label{M11}
\frac{ a(v')}{a(v)}=a(v'').
\end{equation}
Since $\vec{v}$ is the velocity of $F'$ with respect to $F$, $\vec{v}'$ is the velocity of $F''$ with respect to $F$, and $\vec{v}''$ is the velocity of $F''$ with respect to $F'$, we see that the magnitude $v''$ not only depends on the magnitudes $v$ and $v'$, but also on the angle $\theta$ between $\vec{v}$ and $\vec{v}'$. Thus, we see that the right hand side of (\ref{M11}) depends on $\theta$, but the left hand side does not. Since $\vec{v}$ and $\vec{v}'$ are arbitrary, so is their angle $\theta$. Thus, we see that (\ref{M11}) can only hold if $a$ is a constant. It follows from (\ref{M11}) that this constant is $a=1$. Thus, it follows from (\ref{M6}) that $\Delta s'=\Delta s$, as claimed.
\end{proof}

\begin{proof}[Proof of Lemma \ref{Poincare}]
We reproduce a proof given by Ref. \cite{WeinbergGravitationandCosmologybook}.
Let $F$ and $F'$ be an arbitrary pair of inertial reference frames. We consider two spacetime events $E_0$ and $E_1$ with spacetime coordinates separated by infinitesimal spacetime intervals $ds$ and $ds'$ in the reference frames $F$ and $F'$, respectively. \reff{These are given by
\begin{eqnarray}
\label{M12}
ds^2&=&\sum_{\alpha,\beta} \eta_{\alpha\beta}dx_\alpha dx_\beta,\\
\label{M13}
ds'^2&=&\sum_{\gamma,\delta} \eta_{\gamma\delta}dx'_\gamma dx'_\delta,
\end{eqnarray}
where $\eta$ is the metric given by (\ref{metric}),} where $\alpha,\beta,\gamma$ and $\delta$ run over $\{0,1,\ldots,n\}$, and where $x_\alpha$ and $x'_\alpha$ are the spacetime coordinates in the reference frames $F$ and $F'$, respectively, for all $\alpha\in\{0,1,\ldots,n\}$.

From Postulate \ref{NCT}, the coordinate transformation $x\rightarrow x'$ is non-singular. Thus, the functions $x'(x)$ and $x(x')$ are well behaved differentiable functions and the matrix $\frac{\partial x'_\alpha}{\partial x_\beta}$ has a well defined inverse $\frac{\partial x_\beta}{\partial x'_\alpha}$. We express the infinitesimal intervals $dx'_\gamma$ as a function of the infinitesimal intervals $dx_\alpha$, for all $\alpha,\gamma\in\{0,1,\ldots,n\}$, by taking partial derivatives. From (\ref{M13}), we have
\begin{equation}
\label{M14}
ds'^2=\sum_{\alpha,\beta,\gamma,\delta} \eta_{\gamma\delta}\frac{\partial x'_\gamma}{\partial x_\alpha} \frac{\partial x'_\delta}{\partial x_\beta} dx_\alpha dx_\beta.
\end{equation}

By definition of Minkowski spacetime in $1+n$ dimensions, the spacetime interval between the spacetime events $E_0$ and $E_1$ is the same in all inertial reference frames. Thus, we have
\begin{equation}
\label{M15}
ds^2=ds'^2.
\end{equation}
It follows from (\ref{M12}), (\ref{M14}) and (\ref{M15}) that
\begin{equation}
\label{M16}
\sum_{\alpha,\beta}\biggl(\sum_{\gamma\delta}\eta_{\gamma\delta}  \frac{\partial x'_\gamma}{\partial x_\alpha} \frac{\partial x'_\delta}{\partial x_\beta}-\eta_{\alpha\beta}\biggr)dx_\alpha dx_\beta=0.
\end{equation}
In order that (\ref{M16}) holds for arbitrary $dx_\alpha$, we must have
\begin{equation}
\label{M17}
\eta_{\alpha\beta}=\sum_{\gamma,\delta}\eta_{\gamma\delta} \frac{\partial x'_\gamma}{\partial x_\alpha} \frac{\partial x'_\delta}{\partial x_\beta},
\end{equation}
for all $\alpha,\beta\in\{0,1,\ldots,n\}$.

We take the partial derivative of (\ref{M17}) with respect to $x_\epsilon$, and obtain
\begin{equation}
\label{M18}
0=\sum_{\gamma,\delta}\eta_{\gamma\delta} \biggl(\frac{\partial^2 x'_\gamma}{\partial x_\epsilon\partial x_\alpha} \frac{\partial x'_\delta}{\partial x_\beta}+\frac{\partial x'_\gamma}{\partial x_\alpha} \frac{\partial^2 x'_\delta}{\partial x_\epsilon\partial x_\beta}\biggr),
\end{equation}
for all $\alpha,\beta,\epsilon\in\{0,1,\ldots,n\}$. We add to (\ref{M18}) the same equation with $\alpha$ and $\epsilon$ interchanged, and we subtract the same equation with $\epsilon$ and $\beta$ interchanged, to obtain
\begin{eqnarray}
\label{M19}
0&=&\sum_{\gamma,\delta}\eta_{\gamma\delta} \biggl(\frac{\partial^2 x'_\gamma}{\partial x_\epsilon\partial x_\alpha} \frac{\partial x'_\delta}{\partial x_\beta}+\frac{\partial x'_\gamma}{\partial x_\alpha} \frac{\partial^2 x'_\delta}{\partial x_\epsilon\partial x_\beta}\nonumber\\
&&\quad\qquad\qquad+\frac{\partial^2 x'_\gamma}{\partial x_\alpha\partial x_\epsilon} \frac{\partial x'_\delta}{\partial x_\beta}+\frac{\partial x'_\gamma}{\partial x_\epsilon} \frac{\partial^2 x'_\delta}{\partial x_\alpha\partial x_\beta}\nonumber\\
&&\quad\qquad\qquad-\frac{\partial^2 x'_\gamma}{\partial x_\beta\partial x_\alpha} \frac{\partial x'_\delta}{\partial x_\epsilon}-\frac{\partial x'_\gamma}{\partial x_\alpha} \frac{\partial^2 x'_\delta}{\partial x_\beta\partial x_\epsilon}\biggr),\nonumber\\
\end{eqnarray}
for all $\alpha,\beta,\epsilon\in\{0,1,\ldots,n\}$. We see from the definition of the metric matrix $\eta_{\gamma\delta}$ given by (\ref{metric}), that in (\ref{M19}), the second and last terms cancel each other, and so do the fourth and fifth terms. Thus, since the first and third terms are equal, we obtain
\begin{equation}
\label{M20}
0=2\sum_{\gamma,\delta}\eta_{\gamma\delta} \biggl(\frac{\partial^2 x'_\gamma}{\partial x_\epsilon\partial x_\alpha} \frac{\partial x'_\delta}{\partial x_\beta}\biggr),
\end{equation}
for all $\alpha,\beta,\epsilon\in\{0,1,\ldots,n\}$.

It is easy to see from (\ref{metric}) that the matrix $\eta_{\gamma\delta}$ is invertible. As mentioned above, the matrix $\frac{\partial x'_\delta}{\partial x_\beta}$ is also invertible. Thus, the matrix $\sum_\delta \eta_{\gamma\delta}\frac{\partial x'_\delta}{\partial x_\beta}$ is also invertible.  It follows that (\ref{M20}) has only the trivial solution
\begin{equation}
\label{M21}
0=\frac{\partial^2 x'_\gamma}{\partial x_\epsilon\partial x_\alpha} ,
\end{equation}
for all $\alpha,\gamma,\epsilon\in\{0,1,\ldots,n\}$. The general solution to (\ref{M21}) is given by
\begin{equation}
\label{M22}
x'_\gamma=\sum_{\delta=0}^{n}\Lambda_{\gamma\delta} x_\delta+a_\gamma,
\end{equation}
where $\Lambda_{\gamma\delta},a_\gamma\in\mathbb{R}$ are constants, for all $\gamma,\delta\in\{0,1,\ldots,n\}$. By substituting (\ref{M22}) in (\ref{M17}), we obtain
\begin{equation}
\label{M23}
\eta_{\alpha\beta}=\sum_{\gamma=0}^n\sum_{\delta=0}^n\eta_{\gamma\delta}\Lambda_{\gamma\alpha}\Lambda_{\delta\beta}.
\end{equation}
Thus, from (\ref{M22}) and (\ref{M23}), we see that the coordinate transformation $x\rightarrow x'$ is a Poincar{\'e} transformation $P(a,\Lambda)$ defined by (\ref{neweq:l1}) and (\ref{newnew1}), as claimed.
\end{proof}


%\bibliography{Lorentzbiblio}
%\bibliography{postdocbiblio}
\end{document}